\documentclass[reqno,11pt]{amsart}
\usepackage{hyperref}
\usepackage{graphicx}
\usepackage{slashed}
\usepackage{amscd}
\usepackage{tensor}
\usepackage{amssymb}
\usepackage{esint}
\usepackage{tikz}
\usepackage{tikz-cd} 
\usepackage{enumitem}
\usepackage{accents}
\usepackage[off]{auto-pst-pdf} 
\usepackage{pst-grad} 
\usepackage{pst-plot} 
\usepackage[mathscr]{eucal}
\numberwithin{equation}{section}
\allowdisplaybreaks[1]

\title[On the Mathematical Foundations of Causal Fermion Systems]{On the Mathematical Foundations of \\[0.1cm] Causal Fermion Systems in Minkowski Space}

\author[M.\ OPPIO]{MARCO OPPIO}

\theoremstyle{thm}
\newtheorem{theorem}{Theorem}[section]
\newtheorem{lemma}[theorem]{Lemma}
\newtheorem{corollary}[theorem]{Corollary}
\newtheorem{proposition}[theorem]{Proposition}
\newtheorem{definition}[theorem]{Definition}
\newtheorem{remark}[theorem]{Remark}
\newtheorem{notation}[theorem]{Notation}
\newtheorem{assumption}[theorem]{Assumption}

\newcommand{\restr}{\!\!\restriction}
\newcommand{\supp}{\text{supp }}
\newcommand{\V}[1]{{\bf{#1}}}

\DeclareFontFamily{OT1}{rsfso}{}
\DeclareFontShape{OT1}{rsfso}{m}{n}{ <-7> rsfso5 <7-10> rsfso7 <10-> rsfso10}{}
\DeclareMathAlphabet{\mycal}{OT1}{rsfso}{m}{n}

\DeclareMathOperator{\ran}{ran}

\def\scF{\mathscr{F}}
\def\scB{\mathscr{B}}

\def\scH{\mathscr{H}}

\def\Sol{\mathscr{H}_{m}^{sc}}
\def\sol{\mathscr{H}_m}

\newcommand{\Sl}{\mbox{$\prec \!\!$ \nolinebreak}}
\newcommand{\Sr}{\mbox{\nolinebreak $\succ$}}

\def\bpsi{\boldsymbol{\psi}}

\def\cD{\mathcal{D}}
\def\cE{{\mathcal{E}}}
\def\scF{{\ca F}}

\def\sM{{\mathsf M}}
\def\sN{{\mathsf N}}

\def\sP{{\mathsf P}}

\def\sU{\mathscr{U}}

\def\sD{\mathsf{D}}

\def\rF{\mathrm{F}}

\def\bC{{\mathbb C}}           
\def\bI{{\mathbb I}}

\def\cF{\mathcal{F}}

\def\bN{{\mathbb N}}
\def\bM{{\mathbb M}}

\def\bR{{\mathbb R}}

\def\bZ{{\mathbb Z}}

\def\scF{{\mathfrak F}}
\def\gG{{\mathfrak G}}
\def\gg{{\mathfrak g}}

\def\gJ{{\mathfrak J}}

\def\gR{{\mathfrak R}}

\def\Sol{\mathscr{H}_{m}^{sc}}
\def\sol{\mathscr{H}_m}

\def\scF{\mathscr{F}}
\def\scB{\mathscr{B}}

\def\scH{\mathscr{H}}

\begin{document}

\maketitle
\vspace{-0.8cm}
\begin{center}
	{	{\footnotesize\textit{marco.oppio@ur.de}}\\[0.5cm]
			{\em \footnotesize Fakult\"at f\"ur Mathematik,\  Universit\"at Regensburg \\D-93040 Regensburg, Germany\\}}
	\vspace{0.8em}
	
	{\footnotesize{\rm NOVEMBER 2020}}
\end{center}
\begin{abstract}
	The emergence of the concept of a causal fermion system  is revisited and further investigated for the vacuum Dirac equation in Minkowski space. After a brief recap of the Dirac equation and its solution space, in order to allow for the effects of a possibly nonstandard structure of spacetime at the Planck scale, a regularization by a smooth cutoff in momentum space is introduced, and its properties are discussed. Given an ensemble of solutions, we recall the construction of a local correlation function, which realizes spacetime in terms of operators. It is shown in various situations that the local correlation function maps spacetime points to operators of maximal rank and that it is closed and  homeomorphic onto its image.
	It is inferred that the corresponding causal fermion systems are regular and have a smooth manifold structure.
	The cases considered include a Dirac sea vacuum and systems involving a finite number of particles and antiparticles. 
\end{abstract}

\tableofcontents

\section{Introduction}

In most attempts to quantum gravity, a widely accepted principle is the existence of a minimal observable length, generally associated to the Planck length $l_P\sim 10^{-35} m$. 
In order to understand this problem in a simple way, imagine to probe the microscopic structure of spacetime down to the Planck scale. Then the uncertainty principle would lead to energy densities which are large enough to change the structure of spacetime itself drastically.
This argument poses severe constraints to any attempt of defining a notion of localization.

Sticking to the simplest case of Minkowski space, one possible way to implement the existence of a minimal length is to introduce a cutoff regularization in momentum space. Roughly speaking, if the spacetime uncertainty is believed to be bounded from below by $l_P$, then one would expect an upper bound in the momentum uncertainty of the order of $\hbar\, l_P^{-1}$. This can be realized by taking out the momenta of the wave functions of interest which lie above such an upper scale (for mathematical simplicity we stick to a \textit{smooth} cutoff in this paper, see Section \ref{sectionreg}).
This procedure is reminiscent of the usual ultraviolet regularization frequently used in the renormalization program in quantum field theory as a
technical tool to remove divergences.
However, in our setting the regularization 
has a physical significance as it effectively describes the microscopic structure of physical spacetime in the presence of a minimal length. 

As shown in \cite{FF}, the first consequence of implementing the minimal length as a momentum cutoff is the existence of natural realizations of spacetime in terms of finite-rank operators on the one-particle Hilbert space (the so-called \textit{local correlation function}, see Theorem \ref{teoremaesistenzaF}). 
More concretely, given any closed subspace of the one-particle Hilbert space (referred to as an ensemble of solutions) one constructs a function\footnote{Given a Hilbert space $\mathscr{H}$ we denote by $\scB(\mathscr{H})$ the Banach space of bounded  operators on $\mathscr{H}$.}
$$
F^\varepsilon:\bR^{1,3}\rightarrow \mathscr{B}(\scH),\quad \langle u| F^\varepsilon(x)v\rangle_{\mathscr{H}}:=-u_\varepsilon(x)^\dagger\,\gamma^0\,v_\varepsilon(x),
$$ 
where the lower index $\varepsilon$ specifies that the function $u$ has been regularized by the cutoff and $\gamma^0$ is the zeroth Dirac matrix. Each operator $F^\varepsilon(x)$ has rank at most four, and its range is formed by solutions which are relevant at the point $x\in\bR^{1,3}$ (the meaning of this statement will become clear in Section \ref{sectionCFS}). In this way the local information encoded in the fermionic wave functions is absorbed into a suitable redefinition of spacetime.
The structure consisting of an \textit{ensemble of solutions} together with the corresponding \textit{local correlation function}
(and equipped with a canonical \textit{Borel measure} 
 on its image, see Definition \ref{defCFS} and the remark thereafter) defines a \textit{causal fermion system}.

There is a distinguished ensemble of solutions which provides us with a faithful realization of spacetime and is invariant under spacetime translation and as such gives a sensible candidate for a vacuum structure: the family of negative-energy solutions of the Dirac equation. This resembles the original idea of Dirac of the vacuum as a quantum system where all the negative-energy states are occupied. Although in the standard framework of quantum field theory such a concept is no longer used, in our setting Dirac's idea finds a new possibility of interpretation.
The addition of particles corresponds then to an extension of the solution ensemble within the positive-energy subspace, while the addition of antiparticles is realized by a restriction of such a family within the negative one.

Some features of the local correlation functions, like continuity, injectivity and regularity (the property of having maximal rank), have already been analyzed in the case of a Dirac sea vacuum (see for example \cite[Section 1.2.3]{FF} and \cite[Section 4]{FG}). In this paper we delve into this matter more systematically and study under which assumptions a given ensemble of solutions leads to a local correlation function which, besides being continuous, is also closed and homeomorphic onto its image. The topological and differentiable structures of Minkowski space can then be lifted through $F^\varepsilon$ to $\scB(\mathcal{\scH})$, realizing spacetime as a four-dimensional manifold consisting of bounded operators.
\null
In Theorem \ref{teoremamanifold} it is proved that this is always true in the case of a Dirac sea vacuum and in presence of particles, while sufficient assumptions are provided when antiparticles are taken into account.

This shows an asymmetry between matter and anti-matter in our formalism. The reason of this is not difficult to understand. 
The addition of positive-energy solutions simply adds information to the vacuum configuration, which is already rich enough for its local correlation function to provide a faithful representation of spacetime. On the other hand, removing negative-energy solutions from the vacuum ensemble might cause a critical loss of information, ending up in a too poor local correlation function which is, for instance, no longer injective (see the counterexample in Section \ref{subsectininjective}). This can be prevented by restricting attention to negative-energy solutions which are sufficiently ``spread out" in space and do not vary too much on the microscopic scale, where the structure of spacetime has indeed been modified by the introduction of a regularization (see Sections \ref{sectionholes} and \ref{sectionsmooth}). 

The lack of features like injectivity or regularity in the general case and its connection to the microscopic behavior of the wave functions seems to be of physical and mathematical interest, but it will not be discussed in this paper and could be pursued in the future.

Clearly, the whole construction presented here (and originally introduced in \cite{FF}) depends heavily  on the choice of the regularization cutoff. Nevertheless, in the theory of causal fermion systems, the belief is that a distinguished \textit{physically meaningful regularization} does exist and arises naturally as a minimizer of an action principle (see the \textit{causal action principle} in \cite{FF}, Section 1.1.1).
An attempt to construct such ``optimal'' regularizations can be found in~\cite{reg}. Here we do not enter such constructions, because our focus is to analyze the analytic properties of the local correlation function for a simple class of regularizations.

\section{{Standard Results on the Dirac Equation}}

In this chapter we review some standard results on the Dirac equation. In order not to distract from the main scope of the paper, all the proofs of this section are postponed to the appendix.

We assume that an inertial reference frame has been assigned and  Minkowski space realized accordingly as the Euclidean space $\bR^{1,3}$ equipped with the Minkowski inner product $\eta$ with signature convention $(+,-,-,-)$. 
For simplicity of notation, we use natural units $\hbar=c=1$.

\subsection{The Equation and its Solutions Space}

In this section we introduce the basic notions on the  \textit{Dirac equation} and its solution spaces which are relevant for the theory of causal fermion systems. We will focus exclusively on the case of \textit{strictly positive mass $m>0$}.

The starting point is the Dirac first-order linear differential operator:
$$
\mathcal{D}:=i\gamma^\mu\partial_\mu - m \bI_4:\mathcal{C}^\infty(\bR^{1,3},\bC^4)\rightarrow \mathcal{C}^\infty(\bR^{1,3},\bC^4).
$$
The space of smooth solutions of $\mathcal{D}$, i.e. the set of  $f\in\mathcal{C}^\infty(\bR^{1,3},\bC^4)$ such that
\begin{equation}\label{solutiondirac}
	i\gamma^\mu \partial_\mu f =mf.\footnote{In the remainder of the paper we will use Feynman notation $\slashed{a}:=a_\mu \gamma^\mu$ for $a\in\bC^4$ and  $\slashed\partial:=\gamma^\mu\partial_\mu$. }
\end{equation}
 is denoted by $\ker \mathcal{D}$.

Equation (\ref{solutiondirac}) is a \textit{symmetric hyperbolic system of partial differential equations} and as such it admits \textit{unique global solutions} if regular initial data are assigned on a given Cauchy surface 
$$
\Sigma_t:=\{(t,\V{x})\in\bR^{1,3}\:|\: \V{x}\in\bR^3 \}
$$ 
(for details see for example \cite{intro} or \cite[Section 5.3]{john}).
\begin{theorem}\label{existenceuniqueness}
	Referring to equation (\ref{solutiondirac}), the following statements hold.
	\begin{itemize}[leftmargin=2.5em]
		\vspace{0.05cm}
		\item[\rm{(i)}] Let $f,g\in \ker\cD$  be such that $f\restr_{\Sigma_t}=g\restr_{\Sigma_t}$ for some $t\in\bR$, then $f=g$.\\[-0.5em]
		\item[\rm{(ii)}] For any $t\in\bR$ and $\varphi\in \mathcal{C}_0^\infty(\bR^3,\bC^4) \footnote{By $\mathcal{C}_0^\infty(X,Y)$ we denote the linear space of compactly supported smooth functions from $X$ to $Y$.}$ there exists\footnote{The existence of solutions can be proved in the more general case $\varphi\in\mathcal{C}^\infty(\bR^3,\bC^4)$.} $\mathrm{E}_t(\varphi)\in \ker\cD$ with $$\mathrm{E}_t(\varphi)\restr_{\Sigma_t}=\varphi.$$
	\end{itemize}	
	
\end{theorem}
\noindent Points (i) and (ii) in the theorem above guarantee that, for every $t\in\bR$, the function
\begin{equation}\label{funzioneE}
	\mbox{E}_t:\mathcal{C}^\infty_0(\bR^3,\bC^4)\ni\varphi\mapsto \mbox{E}_t(\varphi)\in \ker\cD
\end{equation}
is a well-defined \textit{injective} linear mapping whose image is the set:
\begin{equation}\label{imageE}
	\mbox{E}_t\big(\mathcal{C}_0^\infty(\bR^3,\bC^4)\big)=\{f\in \ker\cD\:|\: f\restr_{\Sigma_t}\in \mathcal{C}_0^\infty(\bR^3,\bC^4) \}.
\end{equation}

It can be proved that the support of the generated solution $\mathrm{E}_t(\varphi)$ is contained in the causal propagation of the support of $\varphi$, as we expect from finite propagation speed.
\begin{proposition}\label{causalpropag}
	For every $s,t\in\bR$ and $\varphi\in \mathcal{C}_0^\infty(\bR^3,\bC^4)$ the following holds.
	\begin{itemize}[leftmargin=2.5em]
		\vspace{0.08cm}
		\item[\rm{(i)}] $\supp \mathrm{E}_t(\varphi)\subset \{t\}\times\supp\varphi + \{x\in\bR^{1,3}\:|\: \eta(x,x)\ge 0 \}$\\[-0.4em]
		\item[\rm{(ii)}] $\mathrm{E}_t(\varphi)\restr_{\Sigma_s}\in \mathcal{C}^\infty_0(\bR^3,\bC^4)$ and  $\mathrm{E}_s(\mathrm{E}_t(\varphi)\restr_{\Sigma_s})=\mathrm{E}_t(\varphi)$.
	\end{itemize}
\end{proposition}

From this proposition we see that the choice of the Cauchy surface does not play any role in the definition of (\ref{imageE}). More precisely:
$$
\mathrm{E}_t\big(\mathcal{C}_0^\infty(\bR^3,\bC^4)\big)=\mbox{E}_s(\mathcal{C}_0^\infty(\bR^3,\bC^4))\quad\mbox{for any }s,t\in\bR.
$$
\begin{definition}
	The image of $\mbox{\em E}_t$ is called the \textbf{space of smooth solutions of the Dirac equation with spatially compact support} and denoted by $\Sol$.
\end{definition}

This space of solutions can be equipped with a pre-Hilbert space structure. Given any ${t\in\bR}$ we consider the sesquilinear function\footnote{The \textit{Euclidean inner product} of $\bR^n$ and $\bC^n$ are denoted by $a^\dagger b$ and $a\cdot b$, respectively. In both cases the corresponding norm is denoted by $|a|$. }
\begin{equation}\label{innerproduct}
	\Sol\times\Sol\ni (f,g)\mapsto (f|g)_t:=(f\restr_{\Sigma_t}\!|\,g\restr_{\Sigma_t})_{\mathcal{L}^2}=\int_{\bR^3}f(t,\V{x})^\dagger g(t,\V{x})\, d^3\V{x} \in\bC.
\end{equation}
Again, the choice of the Cauchy surface plays no role, due to \textit{current conservation}. 
\begin{proposition}\label{currentconserv}
	For any $t\in\bR$ the function $(\cdot|\cdot)_t$ defines a Hermitian inner product. With respect to it and the $\mathcal{L}^2$-norm in the domain, the isomorphism
	$$
	\mathrm{E}_t:\mathcal{C}_0^\infty(\bR^3,\bC^4)\longrightarrow\scH_m^{sc}
	$$ 
	is a linear isometry.  Moreover, $(\cdot|\cdot)_t=(\cdot|\cdot)_0$.
\end{proposition}

Having in mind the construction of a one-particle Hilbert space, the natural next step consists in taking the completion of $\Sol$ with respect to the inner product \eqref{innerproduct}.  In the general case, where no background (complete) Hilbert space is given, the completion of a pre-Hilbert space is constructed in a purely abstract way, by taking as linear space the set of equivalence classes of Cauchy sequences and extending the inner product to it by continuity. Since we aim at building a quantum theory of wave functions, it is important to show that even the limit points can be realized in terms of measurable functions on spacetime.  

A natural space where to embed our space of smooth solutions is the set of locally square-integrable functions
$$
\mathcal{L}_{loc}^2(\bR^{1,3},\bC^4)\supset \Sol.
$$ 
This space can be equipped with the structure of a complete metric space (see for example Lemma 5.17 in \cite{mv}) such that the arising notion of convergence is equivalent to the
requirement:
\vspace{0.2em}

\begin{center}
\textit{$u_n\to u$ in  $\mathcal{L}^2_{loc}(\bR^{1,3},\bC^4)$\\[0.4em]   if and only if\\[0.4em] $\|u_n\restr_B-u\restr_B\!\|_{\mathcal{L}^2}\to 0$ for all bounded open sets  $B\subset\bR^{1,3}$.}
\end{center}
\vspace{0.8em}

\noindent In the following we will make use of the sets
$
R_T:=[-T,T]\times\bR^3.
$
\begin{lemma}\label{lemmaconverging}
	Given the pre-Hilbert topology of $\Sol$, the canonical embedding
	 $$
	 \Sol\hookrightarrow\mathcal{L}^2_{loc}(\bR^{1,3},\bC^4)
	$$ 
	 is Cauchy continuous\footnote{A function between metric spaces is said to be Cauchy continuous if it maps Cauchy sequences to Cauchy sequences}, in particular it is continuous. 
\end{lemma}
At this point, bearing in mind the way the abstract completion of a pre-Hilbert space is constructed, we can  characterize the completion of $\Sol$ with respect to $(\cdot|\cdot)_0$ by assigning to every Cauchy sequence in $\Sol$ the corresponding limit in $\mathcal{L}_{loc}^2(\bR^{1,3},\bC^4)$ and extending the inner product $(\cdot|\cdot)_0$ by continuity in the obvious way.  
\begin{proposition}\label{defSOL}
	The linear space
	$$
	\sol:=\overline{\Sol}^{\,\, \mathcal{L}^2_{loc}}\subset\mathcal{L}_{loc}^2(\bR^{1,3},\bC^4),
	$$ 
	equipped with the continuous extension of $(\cdot|\cdot)_0$ is a Hilbert space, called the \textbf{one-particle Hilbert space of the Dirac equation}. Its elements are called the \textbf{physical solutions of the Dirac equation}. Its inner product is denoted  by $(\cdot|\cdot)_m$.
\end{proposition}
To support our choice of $\sol$ as completion of $\Sol$, we also state the following technical result.
\begin{lemma}\label{lemmaconverging2}
	Let $\{f_n \}_n$ be a Cauchy sequence in $\Sol$. Then the following holds.
	\begin{itemize}[leftmargin=2.5em]
		\vspace{0.08cm}
		
		\item[\rm{(i)}] For every $T>0$ the function  $u=\lim_{n\to\infty} f_n\in\sol$ fulfills
		\begin{equation}\label{identityintegralstrip}
		u\restr_{R_T}\,\in \mathcal{L}^2(R_T,\bC^4)\quad\text{and}\quad
			\|u\restr_{R_T}\!\|_{\mathcal{L}^2}=\sqrt{2T}\|u\|_m.
		\end{equation}
		\item[\rm{(ii)}] There exists a subsequence $\{f_{\sigma(n)} \}_n$ which converges to $u$ pointwise a.e.
	\end{itemize}
	
\end{lemma}

The elements of $\sol$ can still be interpreted as solutions of a partial differential equation, even though in a weak sense, for they are not regular functions in general. Nevertheless, as expected, in the case of a smooth function, the Dirac equation is solved in the ordinary sense.
\begin{theorem}\label{weaksol}
	Every  $u\in\sol$ is a weak solution of the Dirac equation (\ref{solutiondirac}), i.e.
	$$
	\int_{\bR^4} u(x)^\dagger(\cD^* \varphi)(x)\, d^4x=0\quad \mbox{for all }\varphi\in \mathcal{C}_0^\infty(\bR^{1,3},\bC^4).
	$$
	with $\cD^*:=i(\gamma^\mu)^\dagger\partial_\mu+m\bI_4$ the formal adjoint of $\cD$.
	If $u\in\mathcal{C}^\infty(\bR^{1,3},\bC^4)$, then $u\in\ker\mathcal{D}$.
\end{theorem}

At this point,  exploiting the density of $\mathcal{C}_0^\infty(\bR^3,\bC^4)$  within $\mathcal{L}^2(\bR^3,\bC^4)$, it is possible to extend (uniquely) the isometries $\mbox{E}_t:\mathcal{C}^\infty_c(\bR^3,\bC^4)\rightarrow \Sol$ to unitary operators of Hilbert spaces (we keep the same notation):
\begin{equation}
	\mathrm{E}_t:\mathcal{L}^2(\bR^3,\bC^4)\rightarrow \sol.
\end{equation}
Each of these operators admits an inverse 
\begin{equation}
	\mathrm{E}_t^{-1}:\sol\to \mathcal{L}^2(\bR^3,\bC^4)
\end{equation} 
which plays the role of a \textit{trace operator}, assigning to every physical solution its ``restriction" to $\Sigma_t$. This is really the case for the smooth functions of $\sol$ with spatially compact support, while in the general case this is only true in a weak sense, for they merely consist of equivalence classes of measurable functions.

\begin{remark} A few comments follow.
	\begin{itemize}[leftmargin=2.5em]
		\vspace{0.08cm}
		\item[\rm{(i)}] It could be possible to characterize the solutions in $\sol$ in terms of Sobolev spaces, more precisely as elements of $\mathcal{H}^1_{loc}(\bR^{1,3},\bC^4)$. However these methods would require longer preparation and we do not need them here.\\[-0.5em]
		\item[\rm{(ii)}] The space $\mathcal{L}^2(\bR^3,\bC^4)$ can be interpreted as the set of generalized \textit{initial data} for~(\ref{solutiondirac}).\\[-0.5em]
		\item[\rm{(iii)}] Due to unitary equivalence, both $\mathcal{L}^2(\bR^3,\bC^4)$ and $\sol$ can be taken as one-particle Hilbert spaces.  The elements of the former are called \textbf{wave functions}.
	\end{itemize}
\end{remark} 

To summarize, the functions $\mathrm{E}_t$ allow us to interpret the solutions of the Dirac equation in two equivalent ways: either as \textit{functions globally defined on spacetime $\bR^{1,3}$} - that is, the elements of $\sol$ - or in terms of \textit{evolving wave functions within $\mathcal{L}^2(\bR^3,\bC^4)$} - that is, as curves
\begin{equation}\label{unitaryev}
	\bR\ni t\mapsto (\mathrm{E}_t)^{-1}\,\mathrm{E}_0(\psi)\in \mathcal{L}^2(\bR^3,\bC^4),
\end{equation}\\[-0.6em]
\noindent for  arbitrary initial data $\psi\in\mathcal{L}^2(\bR^3,\bC^4)$. This latter description fits better to the standard formulation of quantum mechanics. In the next section we will study the feature of this \textit{evolution map}.

\subsection{The Hamiltonian Operator and its Spectral Decomposition}\label{sectionhamiltoniana}

In both classical and quantum mechanics, with due mathematical differences, the \textit{Hamiltonian} is defined as the \textit{generator of time evolution} and it is generally (but not always) identified with an observable physical quantity of the system: the \textit{energy}. In our framework, time evolution is given in terms of a strongly continuous one-parameter group of unitary operators. Stone Theorem guarantees the existence of a unique self-adjoint generator: this is the Hamiltonian we are looking for.

The time evolution operators were defined in the previous section in (\ref{unitaryev}). We need to prove that they do define a strongly continuous one-parameter group of unitary operators.
\vspace{0.2cm}

\begin{definition}
	The \textbf{time-evolution operator} is defined for every $t\in\bR$ by
	$$
	\mathrm{U}_t:=\mathrm{E}_t^{-1}\,\mathrm{E}_0.
	$$
\end{definition}
Restricting to the dense subspace of compactly supported smooth functions, the action of such a mapping is given by
\begin{equation}
\mathcal{C}^\infty_0(\bR^3,\bC^4)\ni \varphi \mapsto \mbox{E}_0(\varphi)\restr_{\Sigma_t}  \in \mathcal{C}^\infty_0(\bR^3,\bC^4).
\end{equation}

\noindent As expected, the family $\{\mathrm{U}_t\}_{t\in\bR}$ fulfills all the properties of a linear unitary evolution.
\begin{proposition}\label{continuityU}
	The function $\bR\ni t\mapsto \mathrm{U}_t$ is a strongly continuous one-parameter group of unitary operators. The corresponding self-adjoint generator (the \textbf{Hamiltonian}) is given by
	$$
	\mathcal{H}:=-i\gamma^0\sum_{i=1}^3\gamma^i\partial_i+m\gamma^0\quad \left(\mathrm{U}_t=e^{-it\mathcal{H}}\right)
	$$
	with the first Sobolev space $\mathcal{H}^1(\bR^3,\bC^4)$ as domain. The set $\mathcal{C}_0^\infty(\bR^3,\bC^4)$ is a core for~$\mathcal{H}$. 
\end{proposition}

The spectral properties of the Hamiltonian are easier to understand if analyzed in momentum space by means of the (unitary) Fourier Transform\footnote{In this paper the Fourier Transform on $\bR^3$ is defined with respect to the Euclidean inner product:
	\begin{equation}\label{fouriertransform}
	\cF(f)(\V{k}):=\int_{\bR^3}\frac{d^3 \V{x}}{(2\pi)^{3/2}}\,f(\V{x})\,e^{-i\V{x}\cdot \V{k}},\quad \cF^{-1}(g)(\V{x})=\int_{\bR^3}\frac{d^3 \V{k}}{(2\pi)^{3/2}}\,g(\V{k})\,e^{i\V{x}\cdot \V{k}},
	\end{equation}
	while the Fourier Transform on $\bR^{1,3}$ is carried out with respect to the Minkowski inner product:
	\begin{equation}\label{fouriertransformM}
	\cF(f)(k):=\int_{\bR^4}\frac{d^4 x}{(2\pi)^{2}}\,f(x)\,e^{i\eta(x,k)},\quad \cF^{-1}(g)(x)=\int_{\bR^4}\frac{d^4 k}{(2\pi)^{2}}\,g(k)\,e^{-i\eta(x,k)}
	\end{equation}
	which fits better to a relativistic setting. 
}
$$
\cF : \mathcal{L}_x^2(\bR^3,\bC^4)\rightarrow\mathcal{L}_p^2(\bR^3,\bC^4),\quad \|\cF(\psi)\|_{\mathcal{L}^2}=\|\psi\|_{\mathcal{L}^2},
$$
where the lower indices were added just to make clear the distinction between 
\begin{center}
	\textit{position $(x)$ and momentum $(p)$ representations}. 
	\end{center}

\begin{definition}\label{defdistrib}
	Let $\psi\in\mathcal{L}_x^2(\bR^3,\bC^4)$. Then $\cF(\psi)\in\mathcal{L}^2_p(\bR^3,\bC^4)$ is called the \textbf{three-dimensional momentum distribution} of $\psi$.
\end{definition}
The Fourier Transform is an isometric isomorphism on the space of (spinor-valued) Schwartz functions $\mathcal{S}(\bR^3,\bC^4):=\mathcal{S}(\bR^3,\bC)\oplus\mathcal{S}(\bR^3,\bC)\oplus\mathcal{S}(\bR^3,\bC)\oplus \mathcal{S}(\bR^3,\bC)$ (see for example \cite{moretti}):
$$ 
\cF(\mathcal{S}_x(\bR^3,\bC^4))=\mathcal{S}_p(\bR^3,\bC^4).
$$
We can apply this transformation to our operator $\mathcal{H}$ and work directly in momentum space. In what follows, we will make use of the \textbf{energy function}:
$$
\omega:\bR^3\ni \V{k} \mapsto \sqrt{\V{k}^2+m^2}\in\bR.
$$
\begin{theorem}\label{theoremfourierH}
In momentum space, the operator $\mathcal{H}$ is the multiplication operator
	\begin{equation}\label{defh}
	\hat{\mathcal{H}}\varphi=h\cdot \varphi,\quad\mbox{with}\quad h:\bR^3\ni \V{k}\mapsto \gamma^0\sum_{i=1}^3\gamma^ik^i+m\gamma^0\in \bM(4,\bC),
	\end{equation}
	defined on the dense domain
	$$
	D(\hat{\mathcal{H}}):=\{\varphi\in\mathcal{L}_p^2(\bR^3,\bC^4)\:|\: \omega\cdot \varphi\in \mathcal{L}_p^2(\bR^3,\bC^4) \}.
	$$ 
	In particular the associated one-parameter group reads:
	$$
	e^{-it\hat{\mathcal{H}}}: \mathcal{L}_p^2(\bR^3,\bC^4)\ni \psi\mapsto e^{-ith}\cdot \psi \in\mathcal{L}^2_p(\bR^3,\bC^4).
	$$
\end{theorem}

The spectral features of the Hamiltonian $\mathcal{H}$ are easier to analyze in momentum representation than in position representation, in that in the former settings the analysis boils down to studying the matrix $h$.

Note that, for any choice of $\V{k}$, the matrix $h(\V{k})$ is symmetric (with respect to the Euclidean inner product of $\bC^4$) and has eigenvalues $\pm \omega(\V{k})$, both two-fold degenerate (for details follow the discussion in \cite[Section 9.2]{dimock} with the appropriate modifications; see also \cite[Section 2.2]{DM}). 

The linear space $\bC^4$ decomposes into two orthogonal subspaces,
$$
\bC^4= W_\V{k}^+\oplus W_{\V{k}}^-,
$$
which are the images of the following orthogonal projections on $\bC^4$:

\begin{equation}\label{localprojectorposneg}
p_\pm (\V{k}):= \frac{\slashed{k}+m}{2k^0}\gamma^0\bigg|_{k^0=\pm \omega(\V{k})}=\frac{1}{2}\left(\bI_4\mp \frac{\boldsymbol{k}\cdot\boldsymbol{\gamma}}{\omega(\V{k})}\,\gamma^0\pm \frac{m}{\omega(\V{k})}\gamma^0\right).
\vspace{0.2cm}
\end{equation}
\begin{proposition}\label{proporthogonalityp}
	Referring to (\ref{localprojectorposneg}), for every $\V{k}\in\bR^3$ it holds that:
	\begin{itemize}[leftmargin=2.5em]
		\vspace{0.09cm}
		\item[\rm{(i)}] $p_\pm(\V{k})^\dagger=p_\pm(\V{k})$\\[-0.5em]
		\item[\rm{(ii)}] $p_\pm(\V{k})^2=p_\pm(\V{k}),\ \  p_+(\V{k})\,p_-(\V{k})=0,\ \  p_+(\V{k})+p_-(\V{k})=\bI_4$.
	\end{itemize}
\end{proposition}
\vspace{0.5em}

\noindent For every $\V{k}\in\bR^3$, the following four vectors form an orthogonal basis of $\bC^4$:\\[0.1em]
\begin{equation}\label{fundamentalspinors}
\chi^+_{\uparrow\downarrow}(\V{k}):=\left(
\begin{matrix}
e_{\uparrow\downarrow}\\[2ex]
\dfrac{\V{\sigma}\cdot\V{k}}{\omega(\V{k})+m} e_{\uparrow\downarrow}
\end{matrix}
\right)\in W_\V{k}^+,
\quad
\chi^-_{\uparrow\downarrow}(\V{k}):=\left(
\begin{matrix}
-\dfrac{\V{\sigma}\cdot\V{k}}{\omega(\V{k})+m} e_{\uparrow\downarrow}\\[2ex]
e_{\uparrow\downarrow}
\end{matrix}
\right)\in W_\V{k}^-
\end{equation}\\[-0.5em]

\noindent where $e_\uparrow = (1, 0)^t$ and $e_{\downarrow} = (0,1)^t$. The upwards and downwards arrows are chosen in connection with the physical interpretation of these vectors as the spinors carried by the \textbf{spin up} and \textbf{spin down} plane-wave solutions (see also Proposition \ref{spindecomposition}).

\begin{proposition}\label{propPpm}
	The multiplication operators defined by
	$$
	\hat{P}_\pm :\mathcal{L}_p^2(\bR^3,\bC^4)\ni \psi \mapsto p_{\pm}\cdot \psi\in\mathcal{L}_p^2(\bR^3,\bC^4),
	$$
	fulfill the following properties.
	\begin{itemize}[leftmargin=2.5em]
		\vspace{0.08cm}
		\item[\rm{(i)}] $\hat{P}_\pm\in\scB(\mathcal{L}_p^2(\bR^3,\bC^4))$ and $(\hat{P}_\pm)^\dagger=\hat{P}_\pm$,\\[-0.5em]
		\item[\rm{(ii)}] $(\hat{P}_\pm)^2=\hat{P}_\pm,\ \ \hat{P}_+\,\hat{P}_-=0$ and $\hat{P}_++\hat{P}_-=\bI$.
	\end{itemize}
\end{proposition}
\begin{notation}
	On $\mathcal{L}^2_x(\bR^3,\bC^4)$ these projectors are denoted by 
	$$
	P_\pm := \cF^{-1}\circ \hat{P}_\pm\circ \cF
	$$
\end{notation}
Accordingly, the Hilbert space decomposes into two orthogonal subspaces:
$$
\mathcal{L}_p^2(\bR^3,\bC^4)=\hat{P}_-(\mathcal{L}_p^2(\bR^3,\bC^4))\oplus \hat{P}_+(\mathcal{L}_p^2(\bR^3,\bC^4)).
$$

\begin{lemma}\label{lemmaproizionieschwat}
	The following statements are true.
	\begin{itemize}[leftmargin=2.5em]
		\vspace{0.08cm}
		
		\item[\rm{(i)}] $\overline{\hat{P}_\pm(\mathcal{S}_p(\bR^3,\bC^4))}=\hat{P}_\pm(\mathcal{L}_p^2(\bR^3,\bC^4))$\\[-0.5em]
		\item[\rm{(ii)}] $\hat{P}_\pm(\mathcal{S}_p(\bR^3,\bC^4))=\hat{P}_\pm (\mathcal{L}_p^2(\bR^3,\bC^4))\cap\mathcal{S}_p(\bR^3,\bC^4)$
	\end{itemize}
	\vspace{0.1cm}
\end{lemma}

It is now possible to explicit the action of the Hamiltonian on these orthogonal subspaces.
Notice that the domain $\mathcal{S}_p(\bR^3,\bC^4)$ is a dense invariant core for $\hat{\mathcal{H}}$ and on its elements the Hamiltonian acts as a multiplication operator.

\begin{theorem}\label{theoremenergy}
	Let $\varphi\in\mathcal{S}_p(\bR^3,\bC^4)$ and $\psi\in\mathcal{L}_p^2(\bR^3,\bC^4)$. Then the following holds.
	\begin{itemize}[leftmargin=2.5em]
		\vspace{0.08cm}
		\item[\rm{(i)}] $\hat{\mathcal{H}} \varphi = (+\omega)\cdot \hat{P}_+(\varphi) + (-\omega )\cdot \hat{P}_-(\varphi)$.\\[-0.5em]
		\item[\rm{(ii)}] $e^{-it\hat{\mathcal{H}}}\psi= e^{-i\omega t}\hat{P}_+(\psi) + e^{i\omega t} \hat{P}_- (\psi).$\\[-0.5em]
		\item[\rm{(iii)}] $\sigma(\hat{\mathcal{H}})=\sigma_c(\hat{\mathcal{H}})=(-\infty,-m]\cup[m,\infty)$
	\end{itemize}
\end{theorem}

At this point, we may wonder how the momentum distributions in $\mathcal{S}_p(\bR^3,\bC^4)$ look like when represented as elements of $\sol$  by means of the unitary operator $\mbox{E}_0\circ\cF^{-1}$. We have:
\begin{equation*}\label{mappingLH}
\begin{split}
\mathcal{L}_p(\bR^3,\bC^4)\ni \psi\mapsto u_\psi:&=\mathrm{E}_0\circ \cF^{-1} (\psi) =\\
&= \mathrm{E}_t\circ e^{-it\mathcal{H}} \circ\cF^{-1}(\psi)=\\
&=\mathrm{E}_t\circ\cF^{-1}(e^{-it\hat{\mathcal{H}}}\psi)\in\sol.
\end{split}
\end{equation*}
For Schwartz functions we then have the following result.
\begin{proposition}\label{generalsolutionsmooth}
	Let $\varphi\in\mathcal{S}_p(\bR^3,\bC^4)$. Then $u_\varphi\in\sol\cap \mathcal{C}^\infty(\bR^{1,3},\bC^4)$ and
	\begin{equation}\label{fourierexpansion}
	u_\varphi(t,\V{x})=\int_{\bR^3} \frac{d^3\V{k}}{(2\pi)^{3/2}} \left(\varphi_+ (\V{k}) e^{-i(\omega(\V{k})t-\V{k}\cdot\V{x})}+\varphi_-(\V{k}) e^{-i(-\omega(\V{k})t-\V{k}\cdot\V{x})}\right),
	\end{equation}
	with  $\varphi_\pm := \hat{P}_\pm(\varphi)\in\mathcal{S}_p(\bR^3,\bC^4)$. 
\end{proposition}

Theorem \ref{theoremenergy} shows that the subspaces $\hat{P}_\pm(\mathcal{S}_p(\bR^3,\bC^4))$ can be interpreted as  positive- and negative-energy eigenspaces of the Hamiltonian. Sticking to $\sol$ as our favorite realization of the one-particle Hilbert space, we  consider the projectors
$$
\sP_{\pm}:= \hat{\mathrm{E}}_0\circ \hat{P}_\pm\circ \hat{\mathrm{E}}_0^{-1}=\mathrm{E}_0\circ P_\pm\circ\mathrm{E}_0^{-1} \in\scB(\sol),\quad \hat{\mathrm{E}}_0:=\mathrm{E_0}\circ \cF^{-1}
$$
and give the following definition.
\begin{definition}
	The subspaces $\sP_{\pm}(\sol)$ are called the \textbf{positive- and negative-energy subspaces} and denoted by $\sol^\pm$. Their elements are called the \textbf{positive- and negative-energy  physical solutions} of (\ref{solutiondirac}).
\end{definition}

\begin{remark} A few remarks follow.
	\begin{itemize}[leftmargin=2.5em]
		\vspace{0.08cm}
		
		\item[\rm{(i)}]
		The elements in the corresponding orthogonal subspaces 
		$$
		P_\pm(\mathcal{L}_x^2(\bR^3,\bC^4))\subset\mathcal{L}_x^2(\bR^3,\bC^4)
		$$ 
		will be referred to as the \textbf{positive- and negative-energy wave functions}.\\[-0.5em]
		\item[\rm{(ii)}] Explicitly, for every $\psi\in\mathcal{L}_p^2(\bR^3,\bC^4)$ we have 
		$$
		\sP_\pm(u_\psi)=\hat{\mathrm{E}}_0(\hat{P}_\pm(\psi))=u_{\hat{P}_\pm(\psi)}\in\scH_m^\pm
		$$\\[-2.1em]
		\item[\rm{(ii)}] For simplicity of notation, in the remainder of this paper, the subscript $0$ in $\mathrm{E}_0$ and $\hat{\mathrm{E}}_0$ will be dropped.
	\end{itemize}
\end{remark}
\subsection{Four-Momentum Representation and the Fermionic Projectors}

Bearing in mind the distributional identity 
$$
\delta(k^2-m^2)=\frac{\delta(k^0-\omega(\V{k}))}{2\omega(\V{k})}+\frac{\delta(k^0+\omega(\V{k}))}{2\omega(\V{k})},
$$
the (positive- and negative-energy components of the) physical solution of the Dirac equation with three-dimensional momentum distribution $\varphi\in \mathcal{S}_p(\bR^3,\bC^4)$ can be restated  as
\begin{equation}\label{abstractrestatement}
\begin{split}
\hat{\mathrm{E}}(\hat{P}_\pm(\varphi))(t,\V{x})&=\sP_\pm (u_{\varphi})(t,\V{x})=\\
&=\int_{\bR^3} \frac{d^3\V{k}}{(2\pi)^{3/2}}\, p_\pm(\V{k})\, \varphi(\V{k})\, e^{-i(\pm\omega(\V{k})t-\V{k}\cdot\V{x})}=\\ 
&=\pm\int_{\bR^3} \frac{d^3\V{k}}{(2\pi)^{3/2}}\, \frac{\slashed{k}+m}{2\omega(\V{k})}\, (\gamma^0 \varphi(\V{k}))\, e^{-i\eta(k,x)}\bigg|_{k^0=\pm\omega(\V{k})}\!\!\,= \\
&=\pm\int_{\bR^4} \frac{d^4 k}{(2\pi)^2}\,\delta(k^2-m^2)\,\Theta(\pm k^0)\,(\slashed{k}+m)\,\tilde{\varphi}(k)\,e^{-i\eta(k,x)}.
\end{split}
\end{equation}
In the above equations $x=(t,\V{x})\in\bR^{1,3}$ and $\tilde{\varphi}$ is any function of $\mathcal{S}_p(\bR^{1,3},\bC^4)$ whose value on the mass-shell is:
\begin{equation}\label{extension}
\tilde{\varphi}(\pm\omega(\V{k}),\V{k})=\sqrt{2\pi}\,\gamma^0\varphi(\V{k}).
\end{equation}
\begin{remark}
	Note that an extension as in (\ref{extension}) always exists. Take for example 
	$$
	\tilde{\varphi}(k^0,\V{k}):=\sqrt{2\pi}\,\gamma^0\varphi(\V{k})f(\omega(\V{k})-|k^0|),
	$$ 
	for some arbitrary $f\in\mathcal{C}_0^\infty(\bR,[0,\infty))$ with 
	$$
	\supp f\subset [-m/2,m/2]\quad\mbox{and}\quad f\equiv 1\mbox{ and }[-m/4,m/4].
	$$
	This function belongs to $\mathcal{S}_p(\bR^{1,3},\bC^4)$, as follows from Lemma \ref{chiusuraschwartz}.
\end{remark}

The final identity obtained above provides a more compact way to denote solutions.

\begin{proposition}\label{mostgeneralsolution}
	For any $f\in\mathcal{S}_x(\bR^{1,3},\bC^4)$, the function $P(\,\cdot\,,f)$ defined on $\bR^{1,3}$ by
	\begin{equation}\label{fermionicprojdef}
	P_\pm(x,f):=\int\frac{d^4 k}{(2\pi)^2}\, \delta(k^2-m^2)\,\Theta(\pm k^0)\,(\slashed{k}+m)\,\cF(f)(k)\,e^{-i\eta(k,x)}
	\end{equation}
	is a smooth solution of (\ref{solutiondirac}) and has the following representation:
	\begin{equation}\label{representationsol}
	P_\pm((t,\V{x}),f)= \pm\int_{\bR^3}\frac{d^3\V{k}}{(2\pi)^{2}}\, p_\pm(\V{k})\,\gamma^0\, \cF(f)(\pm \omega(\V{k}),\V{k})\, e^{-i(\pm\omega(\V{k})t-\V{k}\cdot\V{x})}.
	\end{equation}
	More precisely, the following identity holds:
	\begin{equation}\label{identificationsolutions}
	\{P_\pm(\,\cdot\,, f)\:|\: f\in\mathcal{S}_x(\bR^{1,3},\bC^4) \}= \hat{\mathrm{E}}(\hat{P}_\pm(\mathcal{S}_p(\bR^3,\bC^4)))\subset \scH_m^\pm\cap \mathcal{C}^\infty(\bR^{1,3},\bC^4).
	\vspace{0.1cm}
	\end{equation}
\end{proposition}
\begin{proposition}\label{propdistribproof}
	For any $x\in\bR^{1,3}$ the mappings
	\begin{equation}\label{distributionP}
	P_\pm(x,\,\cdot\,):\mathcal{S}_x(\bR^{1,3},\bC^4)\ni f\mapsto P_\pm(x,f)\in \bC
	\end{equation}
	are tempered distributions.
\end{proposition}
Proposition \ref{mostgeneralsolution} shows that it is possible to represents all the physical solutions $u_\varphi$ with $\varphi\in\mathcal{S}_p(\bR^3,\bC^4)$ via the formula $P_\pm(\,\cdot\,,f)$, with $f$ ranging within $\mathcal{S}_x(\bR^{1,3},\bC^4)$. This association defines a linear operator.
\begin{definition}\label{deffermionicprojectors}
	The \textbf{(unregularized) fermionic projectors onto the positive and negative spectrum} are defined as the linear mappings
	\begin{equation}\label{deffermoinicprojcetorunreg}
	P_\pm: \mathcal{S}_x(\bR^{1,3},\bC^4)\ni f\mapsto P_\pm(\,\cdot\,,f)\in \scH_m^\pm\cap \mathcal{C}^\infty(\bR^{1,3},\bC^4).
	\end{equation}
	Their difference 
	$$
	P_c:=P_--P_+
	$$
    is called the \textbf{(unregularized) causal propagator }
\end{definition}
By acting with $P_c$ on $\mathcal{S}_x(\bR^{1,3},\bC^4)$ we can realize all solutions $\hat{\mathrm{E}}(\mathcal{S}_p(\bR^3,\bC^4))$.

\begin{proposition}\label{realizationallsolut}
	For every $f\in\mathcal{S}_x(\bR^{1,3},\bC^4)$ there exist $f_\pm\in\mathcal{S}_x(\bR^{1,3},\bC^4)$ such that
	$$
	P_\pm(\,\cdot\,,f)=P_c(\,\cdot\,,f_\pm).
	$$
	More precisely, the causal propagator fulfills
	$$
	P_c(\mathcal{S}_x(\bR^{1,3},\bC^4))=\hat{\mathrm{E}}(\mathcal{S}_p(\bR^3,\bC^4)). 
	$$
\end{proposition}


%

It can be shown with formal computations that the (unregularized) fermionic projectors can be represented as integral operators
\begin{equation}\label{integralrepresentation}
P_\pm(x,f)=\int_{\bR^4}P_\pm(x,y)  f(y)\, d^4y,
\end{equation}
with integral kernel
\begin{equation}\label{xydistribution}
P_\pm(x,y):=\int\frac{d^4 k}{(2\pi)^4}\, \delta(k^2-m^2)\,\Theta(\pm k^0)\,(\slashed{k}+m)\, e^{-i\eta(x-y,k)}.
\end{equation} 
\noindent A calculation similar as in \ref{abstractrestatement} shows that these integral kernels have the following equivalent three-dimensional representation:
	\begin{equation}\label{3dexpression}
		P_\pm(x,y)= \pm\int_{\bR^3}\frac{d^3\V{k}}{(2\pi)^{4}}\, p_\pm(\V{k})\,\gamma^0\, e^{-i(\pm\omega(\V{k})(t_x-t_y)-\V{k}\cdot(\V{x-y}))},
	\end{equation}
It should be stressed, however, that both integrals above do not exist in the Lebesgue sense. The equations above are meaningful only in the sense of  Fourier Transform of tempered distributions. More elementary, the integral representation of $P(x,f)$ in \eqref{integralrepresentation} is defined with the prescription of reversing the order of integration when the integral kernel \eqref{xydistribution} is plugged in.
\vspace{0.5em}
	
\begin{center}
 $P_\pm(x,y)$ is referred to as the \textbf{kernel\footnote{In order to avoid confusion, by kernel we mean more precisely integral kernel and not null space.} of the (unregularized)\\[0.3em] fermionic projector $P_\pm$}.
\end{center}
\vspace{0.5em}

As a matter of fact, it turns out that away from the light-cone the distributions $P_\pm$ are indeed regular. Let us denote the light-cone at the origin by
$$
L_0:=\{\xi\in\bR^{1,3}\:|\: \xi^2=0 \}.
$$
\begin{proposition}\label{kernelP}
There exist unique functions $\mathcal{P}_\pm\in C^\infty(\bR^{1,3}\setminus L_0,\bC^4)$ such that
\begin{equation}\label{smoothrepresent}
P_\pm(x,\varphi)=\int_{\bR^{1,3}}\mathcal{P}_\pm(x-y)\varphi(y)\,d^4y\quad \mbox{for all }\varphi\in C^\infty_0(\bR^{1,3}\setminus L_0,\bC^4).
\end{equation}
 In particular, for any $y\in\bR^{1,3}$ and $a\in\bC^4$ the functions $\mathcal{P}_\pm(\,\cdot\,,y)a$ are smooth solutions of the Dirac equation away from the light-cone centered at $y$.
\end{proposition}
\begin{notation}
	For the sake of simplicity, we will make use of the notation
	$$
	P_\pm(\,\cdot\,,y)a:=\mathcal{P}_\pm(\,\cdot\,-y)a.
	$$
	where representation \eqref{smoothrepresent} is implicitly assumed.
\end{notation}

The following result follows directly from \eqref{xydistribution}. 
\begin{proposition}
	The \textbf{kernel of the (unregularized) causal propagator}\footnote{Also known as \textit{causal fundamental solution}, see \cite[Section 2.1.3]{FF}. Note that our choice for $P_c(x,y)$ differ from the kernel $k_m(x-y)$ used in \cite{FF} by a minus sing.} is 
	$$
	P_c(x,y):=P_-(x,y)-P_+(x,y)=-\int_{\bR^4}\frac{d^4 k}{(2\pi)^4}\, \delta(k^2-m^2)\,\epsilon(k^0)\,(\slashed{k}+m)\,e^{-i\eta(x-y,k)},
	$$
	with $\epsilon(k^0):=\Theta(k^0)-\Theta(-k^0)$ the step distribution.
\end{proposition}

\begin{remark}\label{defPk} A few remarks follow.
	\begin{itemize}[leftmargin=2.5em]
		\vspace{0.08cm}
		
		\item[\rm{(i)}] The term ``causal''  comes from the fact that the corresponding kernel vanishes for spatially separated spacetime points. This is not true for the single terms $P_\pm(x,y)$ (compare with \cite[Section 2.1.3]{FF}).\\[-0.5em]
		\item[\rm{(ii)}] The operator $P_{nc}:=P_-+P_+$ fulfills analogous properties but is not causal. Its corresponding kernel\footnote{This can be found in \cite[Section 2.1.3]{FF} under the different notation $p_m(x,y)$}  is given by
		$$
		P_{nc}(x,y):=\int_{\bR^4}\frac{d^4 k}{(2\pi)^4}\, \delta(k^2-m^2)\,(\slashed{k}+m)\,e^{-i\eta(x-y,k)},
		$$\\[-1.3em]
		\item[\rm{(iii)}] The solutions $P_-(x,y)$ and $-P_+(x,y)$ correspond to the negative and positive energy components of $P_{c}(x,y)$, respectively, and should not be confused with the advanced and retarded Green functions (see \cite[Section 2.1.3]{FF}).\\[-0.5em]
		\item[\rm{(iv)}] Using that $\cF\big(a\,\delta_y^{(4)}\big)(k)=(2\pi)^{-2}\,a\,e^{i\eta(y,k)}$ we can formally rewrite (\ref{xydistribution}) as:
		$$
		P_\pm(x,y)a=P_\pm\big(x,a\delta_y^{(4)}\big)\ \mbox{ for } a\in\bC^4,
		$$
		in agreement with  the notation used in $P_\pm(x,f)$ if we interpret $a\delta_y^{(4)}$ as a function.\\[-0.5em]
		\item[\rm{(v)}] The function $P_\pm(\,\cdot\,,(0,\V{y}))a$ can be interpreted as the the evolution in time of the initial data on $\Sigma_0$ given by
		\begin{equation}\label{loca}
		\begin{split}
		\Omega^\pm(\V{x};a,\V{y}):&=\pm\int_{\bR^3}\frac{d^3\V{k}}{(2\pi)^4}\,p_\pm(\V{k})\,(\gamma^0a)\, e^{i\V{k}\cdot(\V{x}-\V{y})}=\\
		&=\pm P_\pm\big((2\pi)^{-1}\,\gamma^0 a\,\delta_{\V{y}}\big)(\V{x}).
		\end{split}
		\end{equation}
		The integrals in (\ref{loca}) are actually ill-defined in the Lebesgue sense and have a meaning only as tempered distributions in the variable $\V{x}\in\bR^3$. 
		Also, the second identity, which is based on $\cF(b\,\delta_{\V{y}}^{(3)})(\V{k})=(2\pi)^{-3/2}\,b\,e^{-i\V{k}\cdot\V{y}}$, is purely formal.
		Accordingly, the corresponding distributional solution  $P_c(\,\cdot\,,(0,\V{y}))a$ is generated by the initial data:
		\begin{equation}\label{locatot}
		\Omega(\,\cdot\,;a,\V{y}):=\Omega^-(\,\cdot\,;a,\V{y})-\Omega^+(\,\cdot\,;a,\V{y})=-(2\pi)^{-1}\,\gamma^0 a\,\delta_{\V{y}},
		\end{equation}
		which corresponds to a localized state at position $\V{y}\in\bR^3$.
		The distributional solutions arising from (\ref{loca}) correspond then to its negative- and (minus) its positive-energy components, respectively. These can be interpreted as the negative- and positive-energy (distributional) solutions which  are as far as possible localized in space. 
		As will be shortly discussed in Section \ref{decayproperteis}, it is not possible to construct physical (normalizable) solutions in the positive or negative spectrum which have compactly supported initial data.\\[-0.5em]
		\item[\rm{(vi)}]
		It should be mentioned that the solutions in point (iv) differ from the eigenstates of the Newton-Wigner position operator by a factor $\sqrt{2\omega\cdot (\omega+m)^{-1}}$ in the momentum distribution (see for example \cite[Section 2]{NW}). 
	\\[-0.5em]
		\item[\rm{(vii)}] In the reminder of the paper we will make use of the following notation:
		\begin{equation*}
		\begin{split}
		\hat{P}_\pm(k)&:=\delta(k^2-m^2)\,\Theta(\pm k^0)\,(\slashed{k}+m),\\
		\hat{P}_c(k)&:=\hat{P}_-(k)-\hat{P}_+(k)= -\delta(k^2-m^2)\,\epsilon( k^0)\,(\slashed{k}+m).
		\end{split}
		\end{equation*}
		As distributions: $\cF(P_\pm(\,\cdot\,,f))= \hat{P}_\pm\cdot \cF(f)$ for every $f\in\mathcal{S}_x(\bR^{1,3},\bC^4)$.
	\end{itemize}
\end{remark}

We already know that to every physical solution in $\scH_m$ it is possible to assign a three-dimensional momentum distribution via the function $\hat{\mathrm{E}}^{-1}$ (see Definition \ref{defdistrib}). The above discussion shows that this is possibile, for smooth solutions, also in the four-dimensional momentum space, even though just in a distributional sense.

\begin{definition}\label{defdistrib4}
	For every $f\in\mathcal{S}_x(\bR^{1,3},\bC^4)$, the tempered distribution $\hat{P}_c\cdot \cF(f)$ is called the \textbf{four-dimensional momentum distribution} of $P_c(\,\cdot\,,f)$.
\end{definition}


To conclude this section, we  point out that there exists an additional and more detailed representation of the smooth solutions discussed so far. Indeed, making use of Proposition \ref{generalsolutionsmooth} and the basis vectors (\ref{fundamentalspinors}), we see that any positive- or negative-energy solution can be written in terms of spin up and spin down plane-wave solutions. 
\begin{proposition}\label{spindecomposition}
	The set $\hat{\mathrm{E}}(P_\pm(\mathcal{S}_p(\bR^3,\bC^4)))$ is spanned by the solutions
	\begin{equation}\label{basicsolution}
	\begin{split}
	u_\uparrow^\pm(t,\V{x})&:=\int_{\bR^3}\frac{d^3\V{k}}{(2\pi)^{3/2}}\, \lambda_{\uparrow}^{\pm}(\V{k})\,\chi_{\uparrow}^\pm(\V{k})\,e^{-i(\pm \omega(\V{k})t-\V{x}\cdot\V{k})}\\ u_\downarrow^\pm(t,\V{x})&:=\int_{\bR^3}\frac{d^3\V{k}}{(2\pi)^{3/2}}\, \lambda_{\downarrow}^{\pm}(\V{k })\,\chi_{\downarrow}^\pm(\V{k})\,e^{-i(\pm\omega(\V{k})t-\V{x}\cdot\V{k})}
	\end{split}
	\end{equation}
	for arbitrary functions $\lambda_{\uparrow,\downarrow}\in\mathcal{S}_p(\bR^3,\bC)$.
\end{proposition}
In the appendix (see Lemma \ref{estimate}) two estimates concerning these functions are carried out. These will prove useful later on.

\subsection{A Few Words on the Decay Properties of Dirac Solutions}\label{decayproperteis}
In the analysis of the smooth structure of the  vacuum causal fermion system the decay properties of the smooth solutions of (\ref{solutiondirac}) will become crucial.  

As a first remark, we  point out that the smooth solutions of the Dirac equation which lie either in $\sol^-$ or $\sol^+$ are spread all over spacetime: in particular  they cannot be localized in bounded regions. This is made mathematically precise by  the following theorem whose proof can be found in \cite[Corollary 1.7]{thaller}.

\begin{theorem}\label{nevervanishing}
	Let $u\in\hat{\mathrm{E}}(\hat{P}_\pm(\mathcal{S}_p(\bR^3,\bC^4)))$ be different from zero. Then
	$$
	\supp (u\restr_{\Sigma_t})=\bR^3\quad\mbox{for all $t\in\bR$}
	$$
\end{theorem}
In particular this theorem has the following corollary:
$$
\supp u=\bR^{1,3}\quad \mbox{for all } u\in\hat{\mathrm{E}}(\hat{P}_\pm(\mathcal{S}_p(\bR^3,\bC^4)))\setminus\{0\}
$$
Indeed, if this were not the case, then it would be possible to find a non-empty open set $\Omega\subset\bR^{1,3}$ such that $u\restr_\Omega=0$. As a consequence, there must exist some $t\in\bR$ such that the open set $\Omega_t:=\Sigma_t\cap\Omega$ is not empty and $u\restr_{\Omega_t}=0$, contradicting Theorem \ref{nevervanishing}.

Note that this theorem does not state that the general solution in $\hat{\mathrm{E}}(\hat{P}_\pm(\mathcal{S}_p(\bR^3,\bC^4)))$ vanishes nowhere, but that the open set of points where it does not vanish is dense (the same is true for its restriction to Cauchy surfaces)

Nevertheless, it can be proved that these solutions decay to zero at space and time infinities, as one could expect. In order to see this, consider  the Dirac operator $\cD:=i\slashed{\partial}-m$ defined in (\ref{solutiondirac}) and fix any element $f\in \ker\cD$. Using the identity $\slashed{\partial}^2=\square$\footnote{The symbol $\square$ denotes the Laplacian operator $\partial_\mu\partial^\mu$}, we see that the components of $f$ satisfy the Klein-Gordon equation:
\begin{equation}\label{key}
\square f^\mu= -m^2 f^\mu.
\end{equation}
Therefore, we can apply \cite[Theorem 7.2.1]{horm} which is reported here for simplicity.
\begin{theorem}
	Let $\varphi\in\mathcal{S}(\bR^n,\bC)$  and  $N\in\bZ$ with $N<-(n+1)/2$ be such that
	\begin{equation}\label{polynbound}
	|D^\alpha \hat{\varphi}(\boldsymbol{\xi})|\le C_\alpha(1+|\boldsymbol{\xi}|)^{N-|\alpha|} \quad\mbox{for every multi-index $\alpha$}.
	\end{equation}
	Then the solution of the Klein-Gordon equation given by
	\begin{equation}\label{KGsol}
	f(t,\V{x}):=\frac{1}{(2n)^n}\int_{\bR^{n}}e^{i(t\sqrt{1+|\boldsymbol{\xi}|^2}+\V{x}\cdot\boldsymbol{\xi})}\,\hat{\varphi}(\boldsymbol{\xi})\,d^n\boldsymbol{\xi}
	\end{equation}
	 satisfies, for $|t|+|\V{x}|\ge 1$,
	\begin{equation}\label{generalinequality}
	|f(t,\V{x})|\le C(|t|+|\V{x}|)^{N+1}\,(1+(t^2-|\V{x}|^2)_+)^{M_+/2}\,(1+(|\V{x}|^2-t^2)_+)^{M_-},
	\end{equation}
where $M_-$ is arbitrary, $C=C(M_-)$ and $M_+:=\max\{0,-\frac{n}{2}-N-1\}$,
\end{theorem}
Of course, the same result applies if the sign of the energy is changed and the mass is different from one.

At this point, focus on solutions (\ref{basicsolution}) with negative energy (the positive-energy case is analogous). The functions $\hat{\varphi}:=\lambda_{\uparrow}^-\chi_{\uparrow}^-, \lambda_{\downarrow}^-\chi_{\downarrow}^-$ belong to $\mathcal{S}_p(\bR^3,\bC^4)$ and therefore the form of the four components of the functions $u^-_{\uparrow,\downarrow}$ matches  with \eqref{KGsol}. Inequality \eqref{polynbound} is fulfilled for any negative integer $N$ because $\hat{\varphi}$ is a Schwartz function. If we take $N= -3$, then $N< -(n+1)/2$ with $n=3$ and $M_+=1/2$. Finally, if we choose $M_-=0$, then inequality \eqref{generalinequality} gives, for $|t|+|\V{x}|\ge 1$,
\begin{equation*}
	|u_{\uparrow,\downarrow}^-(t,\V{x})|\le C\frac{\big(1+(t^2-|\V{x}|^2)_+\big)^{1/4}}{(|t|+|\V{x}|)^{2}}
\end{equation*}
The same argument applies to the positive-energy solutions. Since any solution can be written as linear composition of the solutions $u^\pm_{\uparrow,\downarrow}$ the following result follows directly.
\begin{proposition}\label{hormander}
	Let $u\in\hat{\mathrm{E}}(\mathcal{S}_p(\bR^3,\bC^4))$. Then there exists $C>0$ such that
	\begin{equation*}
		|u(t,\V{x})|\le C\frac{\big(1+(t^2-|\V{x}|^2)_+\big)^{1/4}}{(|t|+|\V{x}|)^{2}}\quad\mbox{if $\ |t|+|\V{x}|\ge 1$}.
	\end{equation*}
In particular,
	$
	\lim_{n\to \infty}|u(x_n)|= 0
	$
	whenever $\lim_{n\to\infty}|x_n|=\infty.$
\end{proposition}

\section{Microscopic Investigation:  Regularization}

In this chapter we discuss the regularization of the physical solutions. As anticipated in the introduction, led by the conjecture that the nature of spacetime is not continuous on a microscopic scale or, being more cautious, that it looks differently from what we expect from a macroscopic point of view, we can make an attempt and modify the theory on this scale by restricting its domain of validity up to a scale of the order of the Planck length. For the sake of generality, we denote this critical microscopic scale by $\varepsilon$ and let it vary within some interval $(0,\varepsilon_{max})$.

\subsection{Some Quantitative Assumptions on the Regularization Parameter}

Let us denote by $m_P$ the \textit{Planck mass}, whose value in natural units\footnote{In natural units:
	$
	[\mbox{energy}]=[\mbox{momentum}]=[\mbox{mass}]=[\mbox{length}]^{-1}=[\mbox{time}]^{-1}
	=\mbox{eV}$} reads: 
$$
m_P\sim 1.22 \times 10^{28} \mbox{ eV}
$$ 
As a comparison, the rest masses of the electron, the proton and the top quark are, respectively:
$$
m_e\sim 5.11\times 10^{5}\mbox{ eV}, \quad m_p\sim 9.38\times 10^{8}\mbox{ eV}, \quad m_t\sim 1.73\times 10^{11}\mbox{ eV},
$$
The corresponding mass ratios are given by
$$
\frac{m_e}{m_P}\sim 4.19\times 10^{-23},\quad \frac{m_p}{m_P}\sim 7.69\times 10^{-20},\quad \frac{m_t}{m_P}\sim 1.42\times 10^{-17}.
$$
Directly related with the Planck mass, we can define the \textit{Planck length} $l_P$ by the relation 
$$
l_P =\frac{1}{m_P}\sim 8.20 \times 10^{-29}\mbox{ eV}^{-1}\quad \left(\sim 1.62\times 10^{-35}m\right).
$$ 
The length $l_P$ is a good candidate for the microscopic scale $\varepsilon$, at least concerning the order of magnitude. Nevertheless, we make the following weaker and cautious assumption, which will be taken for granted in the rest of the paper, if not specified otherwise.

\begin{assumption}\label{assumptionepsilon}
	Let $m>0$ denote the positive mass of the system as in \eqref{solutiondirac}.  Then the microscopic scale $\varepsilon$ is bounded from above by
	$$
	m\varepsilon\le m\varepsilon_{max} = 10^{-15}
	$$
\end{assumption}

\subsection{Regularization: Momentum  Cutoff or Spacetime Mollification?}\label{sectionreg}
As discussed in the previous sections, to every solution $u\in \scH_m$ a three-dimensional momentum distribution $\hat{\mathrm{E}}^{-1}(u)\in\mathcal{L}_p^2(\bR^3,\bC^4)$ can be associated (see Definition \ref{defdistrib}). Similarly, restricting to the dense subspace spanned by the solutions $P(\,\cdot\,,f)$, also a four-dimensional momentum distribution can be defined (see Definition \ref{defdistrib4}). These representations of the solutions allow us to better understand how the minimal length can be implemented.

Roughly speaking, if a lower bound $\Delta x\gtrsim \varepsilon$ exists in spacetime , then we may expect an upper bound in momentum space of the order $\Delta k\lesssim \varepsilon^{-1}$. This can be accomplished by introducing a cutoff in momentum space. More precisely, the idea is to multiply the momentum distributions by a smooth function which decays fast enough to take out (or at least weaken) the contribution coming from momenta larger than $\varepsilon^{-1}$. This can be carried out in both the three- and four-dimensional momentum spaces.

It is possible to make some assumptions on the cutoff functions, led by physical intuition. For example, it is sensible to assume that these functions are spherically symmetric in the three-momentum coordinates, for there should  be no distinguished direction. Similarly, in the four-dimensional space, it sounds reasonable to assume that these functions are symmetric under inversion of the sign of energy. We want now to discuss this in detail.

\begin{remark}[A Warning on  Poincar\'e Invariance]
Before entering the discussion, we want to warn the reader that introducing a cutoff regularization in momentum space  (as is done  in this work and will become clear soon)  unavoidably breaks Poincar\'e symmetry. More precisely, all the constructions carried out in this paper are reference dependent.
This is evident, for example, if we compare the kernel of the  fermionic projector in \eqref{xydistribution} with its regularization in  Proposition \ref{rfp}: of the manifest Poincar\'e invariance of $P_\pm(x,y)$ only rotation and translation invariance survive the introduction of a cutoff.
Nevertheless,  although undesirable, a  loss of Poincar\'e invariance may not be too surprising if one keeps in mind that the microscopic nature of physical spacetime is completely unknown, and there is no reason why Poincar\'e symmetry should be preserved at scales of the order of the Planck length.  
\end{remark}

\begin{definition}\label{curoff4d}
	A function $\mathfrak{G}\in\mathcal{S}_p(\bR^{1,3},\bR)\setminus\{0 \}$ which fulfills:
	\begin{itemize}[leftmargin=2.5em]
		\vspace{0.09cm}
		\item[\rm{(i)}] $\mathfrak{G}(k)\ge 0$ for any $k\in\bR^{1,3}$,\\[-0.5em]
		\item[\rm{(ii)}] $\mathfrak{G}(-k^0, \mathrm{R}\V{k})=\mathfrak{G}(k^0,\V{k})$ for any $k\in\bR^{1,3}$ and $\mathrm{R}\in \mathbb{O}(3)$,
	\end{itemize}
	is called a \textbf{four-momentum cutoff.}
\end{definition}

\begin{remark}
	The paradigm of such a cutoff function is obviously the Gaussian:
	$$
	\gG(k):= A e^{-\frac{|k|^2}{2B^2}}\quad\mbox{for all }k\in\bR^{1,3}.
	$$
\end{remark}

As is clear from  \ref{abstractrestatement},  the four-momenta which contribute to the Fourier expansion of the solutions are  exclusively the four-vectors lying on the mass shell $|k^0|=\omega(\V{k})$.   This means that we can in fact focus on the three-momenta $\V{k}\in\bR^3$. More precisely, we can restrict our attention to the  function
\begin{equation}\label{restriction}
\mathfrak{g}:\bR^3\ni\V{k}\mapsto \mathfrak{G}(\omega(\V{k}),\V{k})\in\bR,
\end{equation}
which selects the values of the cutoff on the mass shell. Then the following holds.
\begin{proposition}\label{regul3d}
	Referring to (\ref{restriction}),  the function $\mathfrak{g}$ belongs to $\mathcal{S}_p(\bR^3,\bR)$ and satisfies
	\begin{itemize}[leftmargin=2.5em]
		\vspace{0.09cm}
		\item[\rm{(i)}] $\mathfrak{g}(\V{k})\ge 0$ for any $\V{k}\in\bR^3$\\[-0.5em]
		\item[\rm{(ii)}] $\mathfrak{g}(\mathrm{R}\,\V{x})=\mathfrak{g}(\V{k})$ for any $\V{k}\in\bR^3$ and $\mathrm{R}\in \mathbb{O}(3)$,
	\end{itemize}
\end{proposition}
\begin{proof}
	Points (i) and (ii) follow directly from the analogous properties of $\mathfrak{G}$. The fact that $\mathfrak{g}$ is a Schwartz function follows from the fact that $\mathfrak{G}$ is a Schwartz function itself and from Lemma \ref{chiusuraschwartz}.
\end{proof}
Of course, in order not to get a trivial regularization, we need  to assume that such a function does not vanish entirely. This can be achieved, for example, if $\gG$ is strictly positive on a sufficiently large ball intersecting the mass shell.
\begin{definition}\label{definitionreg}
	A function $\mathfrak{g}\in \mathcal{S}_p(\bR^3,\bR)\setminus\{0\}$ fulfilling (i) and (ii) of Proposition \ref{regul3d} is called a \textbf{three-momentum cutoff.}
\end{definition}

We have shown that a four-momentum cutoff induces a three-momentum cutoff. It turns out that also the other way around is possible, although not in a unique way. This shows that the two concepts are indeed interchangeable.
\begin{proposition}\label{extensiong}
	Let $\mathfrak{g}$ be a three-momentum cutoff. Then there exists a four-momentum cutoff $\mathfrak{G}$ such that 
	$$
	\mathfrak{g}(\V{k})=\mathfrak{G}(\omega(\V{k}),\V{k})\quad\mbox{for all }\V{k}\in\bR^3.
	$$
	In particular, $\mathfrak{G}$ does not vanish identically on the mass shell.
\end{proposition}
\begin{proof}
	Let $\mathfrak{g}$ be as in Definition \ref{definitionreg} and define $\mathfrak{G}(k^0,\V{k}):=\mathfrak{g}(\V{k})f(\omega(\V{k})-|k^0|)$ for some  arbitrarily fixed function $f\in\mathcal{C}_0^\infty(\bR,[0,\infty))$ which fulfills $\supp f\subset [-m/2,m/2]$ and $f\equiv 1$ on $[-m/4,m/4]$. Again, exploiting Lemma \ref{chiusuraschwartz}, is can be proved that $\mathfrak{G}$ belongs to $\mathcal{S}_p(\bR^{1,3},\bR)$. Moreover it fulfills points (i) and (ii) of Definition \ref{curoff4d}  and gives back the function $\mathfrak{g}$ if restricted to the mass shell.
\end{proof}

So far, we have not discussed the role played by the microscopic length scale $\varepsilon$ or, more precisely, how it enters the definition of a cutoff.
The leading idea is that a cutoff function should be concentrated around a ball of radius $\varepsilon^{-1}$. However, before making this  mathematically more precise, we need to make some considerations.

Even though from a physical point of view such a microscopic scale should be fixed a priori, from a purely mathematical perspective it is useful not to fix the value of the cutoff parameter $\varepsilon$, but instead to leave some freedom in the choice of the microscopic scale.  This allows us to consider the limit $\varepsilon\searrow 0$ of the various regularized quantities: this is the so-called \textit{continuum limit} (not analyzed here, see \cite{FF}). 
Therefore, in what follows we will merely assume that (compare with Assumption \ref{assumptionepsilon}):
$$
\mbox{\textit{The physical microscopic scale $\varepsilon$ ranges within} } (0,\varepsilon_{max}).
$$
Bearing this in mind, the above ideas can be implemented as follows.

Suppose we are given a four-dimensional cutoff function $\gG$ which, intuitively speaking, is concentrated around the unit ball. For example, we may consider a Gaussian function centered at the origin with unit variance, or a compactly supported function which equals the identity on the unit ball and rapidly decays to zero outside it. If we now define $\gG_\varepsilon(k):=\gG(\varepsilon\cdot k)$, we obtain a new cutoff function which is concentrated around the larger ball $B(0,\varepsilon^{-1})$. Led by this idea, we can give the following general definition.

\begin{definition}\label{definitionregularization}
	Let $\gG$ be a four-momentum cutoff. Then a regularization (cutoff) family is a collection $\{\gG_\varepsilon \}_{\varepsilon\in (0,\varepsilon_{max})}$ of four-momentum cutoffs defined by
	$$
	\gG_\varepsilon(k)=\gG(\varepsilon\cdot k)\quad\mbox{for all } k\in\bR^{1,3}.
	$$
	The corresponding restrictions to the mass shell (see (\ref{restriction})) are denoted by $\{\gg_\varepsilon \}_{\varepsilon\in (0,\varepsilon_{max})}$.
\end{definition}

It should be mentioned that there is an analogous but different way to implement a regularization cutoff family. One could, indeed, fix a \textit{three}-momentum cutoff $\gg$, define $\gg_\varepsilon(\V{k}):=\gg(\varepsilon\cdot \V{k})$ for every $\varepsilon\in (0,\varepsilon_{max})$ and  then take as $\gG_{\varepsilon}$ any extension of $\gg_\varepsilon$ as in Proposition \ref{extensiong}. 
This construction, however, is generally not equivalent to the one given in Definition \ref{definitionregularization}, because of the nonlinearity of  $\omega$. More precisely, a cutoff $\{\gG_{\varepsilon}\}_\varepsilon$ constructed in this way can generally not be written as $\{\hat{\gG}(\varepsilon\cdot k)\}_\varepsilon$  for some given $\hat{\gG}$ as in Definition \ref{definitionregularization}. In this paper we will focus mainly on the construction as in Definition \ref{definitionregularization}.

\begin{remark}\label{4momentumcutoff} A few remarks follow.
	\begin{itemize}[leftmargin=2.5em]
		\vspace{0.09cm}
		\item[\rm{(i)}]As already mentioned above, due to the factor $\delta(k^2-m^2)$, the four-momentum distribution of any solution of the Dirac equation is supported on the mass shell $k^2=m^2$. Therefore, when multiplying such distribution by a cutoff function, we see that the off-shell values of the latter do no contribute. With this in mind,
		 we can interchangeably refer to both $\gG_{\varepsilon}$ and $\gg_{\varepsilon}$ as regularization cutoff in momentum space.\\[-0.5em]
		\item[\rm{(ii)}] The intuitive picture is that $\gG_{\varepsilon}$ provides a smooth approximation of the characteristic function of the sphere $B(0,\varepsilon^{-1})$. 
		There is no need to make this assumption mathematically more precise, though, for we will mainly focus on the properties of the Fourier Transform of $\gG_\varepsilon$, as will be discussed shortly.
	\end{itemize}
	
\end{remark}

%
%

We are ready to \textit{regularize} the functions of $\scH_m$. 
As anticipated at the beginning of this section, the idea is to take out the large momenta in the momentum distributions of the solutions in $\scH_m$ by  multiplication against cutoff functions.

Take a regularization family $\{\gG_\varepsilon\}_\varepsilon$ and fix $\varepsilon\in (0,\varepsilon_{max})$. Every element $u$ of $\scH_m$ can be written as:
$$
u=u_\psi:=\mathrm{E}_0(\cF^{-1}(\psi))\mbox{ for some }\psi\in\mathcal{L}^2_p(\bR^3,\bC^4)
$$
The idea is then to replace the momentum distribution $\psi$ by the modified one $\gg_{\varepsilon}\psi$.

\begin{proposition}\label{defiregularization}
	Given a regularization family $\{\mathfrak{g}_\varepsilon\}_{\varepsilon}$,   \textbf{regularization operators} are defined by
	$$
	\gR_\varepsilon: \sol\ni \hat{\mathrm{E}}(\psi)\mapsto \hat{\mathrm{E}}(\mathfrak{g}_\varepsilon\psi)\in\sol\quad\mbox{for all }\varepsilon\in (0,\varepsilon_{max}).
	$$
	The following properties are fulfilled.
	\begin{itemize}[leftmargin=2.5em]
		\vspace{0.08cm}
		\item[\rm{(i)}] $\gR_{\varepsilon}$ is a bounded self-adjoint operator.\\[-0.5em]
		\item[\rm{(ii)}] $\gR_{\varepsilon}(\sol^\pm)\subset\sol^\pm$\\[-0.5em]
		\item[\rm{(iii)}] $u\in \ker \gR_\varepsilon\mbox{ if and only if }\mathfrak{g}_\varepsilon\cdot\hat{\mathrm{E}}^{-1}(u)=0 \mbox{ a.e.}.$
	\end{itemize}
\vspace{0.08cm}

	\noindent In particular, if the set of zeros of $\mathfrak{g}_\varepsilon$ has measure zero, then $\ker\gR_{\varepsilon}$ is trivial.
\end{proposition}
\begin{proof}
	Point (i) follows  from the continuity of $\psi\mapsto \mathfrak{g}_\varepsilon \psi$ on $\mathcal{L}_p^2(\bR^3,\bC^4)$, the fact that $\gg_\varepsilon$ is real valued  and the fact that $\hat{\mathrm{E}}$ is a unitary operator. Point (iii) is a direct consequence of the injectivity of $\hat{\mathrm{E}}$. Let us prove point (ii). Exploiting identity (\ref{fourierexpansion}), we see that $\gR_{\varepsilon}(\hat{\mathrm{E}} ( \hat{P}_\pm(\mathcal{S}_p(\bR^3,\bC^4))))\subset \sol^\pm$, for $\mathfrak{g}_\varepsilon$ is a scalar function and does not affect the spectral characteristics of the momentum distributions. The statement follows from the density of $\hat{\mathrm{E}} ( \hat{P}_\pm(\mathcal{S}_p(\bR^3,\bC^4)))$ within $\sol^\pm$ (see Lemma \ref{lemmaproizionieschwat} and the definition of $\hat{\mathrm{E}}$) and the continuity of $\gR_\varepsilon$.
\end{proof}
\begin{remark}
	Depending on the choice of the regularization family, the regularization operators can be injective or not. For the moment, we prefer not to make any assumption on the form of the cutoff functions. 
	The set $\ker\gR_\varepsilon$ is to be interpreted as the set of \textit{high-momenta} solutions of $\scH_m$, and as such they turn unphysical once the regularization is assumed to be physically meaningful. In the specific case of a mollification regularization that will be introduced soon (see Theorem \ref{teoremaformaregol} and Assumption \ref{assumpformarego}), the kernel is in fact trivial.
\end{remark}

As should be, this regularization procedure is equivalent to modifying the \textit{four}-dimensional momentum distribution by means of $\gG_{\varepsilon}$, when this is possible. 
More precisely, if we restrict our attention to the smooth elements $u_\varphi=P_c(\,\cdot\,,f)$ (see Proposition \ref{mostgeneralsolution}), it is not difficult to see that (compare with (\ref{abstractrestatement})):
$$
\varphi\mapsto \mathfrak{g}_{\varepsilon}\cdot \varphi\ \mbox{\textit{ is tantamount to} }\ \hat{P}_c\cdot\cF(f)\mapsto \gG_\varepsilon\cdot\big(\hat{P}_c\cdot\cF(f)\big).
$$
\begin{proposition}
	Referring to Proposition \ref{defiregularization} and Proposition \ref{mostgeneralsolution}, it holds that:
\begin{equation}\label{regsmooth}
\gR_\varepsilon(P_\pm(\,\cdot\,,f))(x)=\int\frac{d^4 k}{(2\pi)^2}\, \hat{P}_-(k)\,\gG_{\varepsilon}(k)\,\cF(f)(k)\,e^{-i\eta(k,x)},
\end{equation}
	for any $f\in\mathcal{S}_x(\bR^{1,3},\bC^4)$ and $x\in\bR^{1,3}$.
\end{proposition}

A well-known feature of the Fourier Transform is that it maps products into convolutions and vice versa\footnote{With our conventions: $\cF(f*g)=(2\pi)^{n/2}\, \cF(f)\cdot \cF(g)$ and $\cF^{-1}(u\cdot v)=(2\pi)^{-{n/2}}\, \cF^{-1}(u)*\cF^{-1}(v)$. This is true in the general case of a product or convolution between a Schwartz function and a tempered distribution.}.  This allows for a simple representation of the regularized operator in position space, at least for such smooth solutions. Indeed, given any $f\in\mathcal{S}_x(\bR^{1,3},\bC^4)$, we have:
\begin{equation}\label{regsmooth2}
\begin{split}
P_c(\,\cdot\,, f)\stackrel{\gR_\varepsilon}{\longmapsto} \cF^{-1}\big(\big(\hat{P}_c\cdot \cF(f)\big)\cdot \gG_\varepsilon\big)&= (2\pi)^{-2}\,\big(\cF^{-1}\big(\hat{P}_c\cdot \cF(f)\big)\big)* \cF^{-1}(\gG_\varepsilon)=\\
&= P_c(\,\cdot\,,f)*((2\pi)^{-2}\cF^{-1}(\gG_\varepsilon)).
\end{split}
\end{equation}
\begin{remark}
	It is important to notice that identities (\ref{regsmooth}) and (\ref{regsmooth2}) depend only on the values attained by $\gG_\varepsilon$ on the mass shell. In particular, if two regularization functions coincide on the mass shell, their Fourier Transform might be different, but their convolution with the solutions yield the same result.
\end{remark}

\begin{proposition}\label{convoltioncomesabout}
	The following statements are true.
	\begin{itemize}[leftmargin=2.5em]
		\vspace{0.1cm}
		\item[\rm{(i)}] Let $\{\gG_{\varepsilon}\}_\varepsilon$ be a regularization family and
		 $$
		 h_\varepsilon := (2\pi)^{-2}\,\cF^{-1}(\gG_\varepsilon)\in\mathcal{S}_x(\bR^{1,3},\bC).
		 $$
		 Then, for every $f\in\mathcal{S}_x(\bR^{1,3},\bC^4)$,
		\begin{equation}\label{regularizationP}
		\gR_\varepsilon(P_\pm(\,\cdot\,,f))= P_\pm(\,\cdot\,,f)*h_\varepsilon=P_\pm (\,\cdot\,,f*h_\varepsilon)\in \scH_m^\pm\cap\mathcal{C}^\infty(\bR^{1,3},\bC^4)
		\end{equation}\\[-2.1em]
		\item[\rm{(ii)}] There exists a regularization family $\{\gG_{\varepsilon}\}_\varepsilon$  such that $h_\varepsilon$ fulfills:
		\begin{itemize}[leftmargin=2.5em]
			\vspace{0.1cm}
			\item[\rm{(1)}] $h_\varepsilon\in \mathcal{C}^\infty_c(\bR^{1,3},\bR)$, $h\ge 0$ and $\|h_\varepsilon\|_{\mathcal{L}^1}=1$\\[-0.9em]
			\item[\rm{(2)}] $\supp h_\varepsilon \subset B(0,\varepsilon)$\\[-0.9em]
			\item[\rm{(3)}] $h_\varepsilon(-x_0,\mathrm{R}\V{x})=h_\varepsilon(x_0,\V{x})$ for every $(x_0,\V{x})\in\bR^{1,3}$ and $\mathrm{R}\in \mathbb{O}(3)$.\\[-0.8em]
		\end{itemize}
		In this case, the set of zeros of $\mathfrak{g}_\varepsilon$ is a Lebesgue null set and $\ker\gR_{\varepsilon}=\{0\}$.
	\end{itemize}
\end{proposition}
\begin{proof}
	Let us start with point (i). The first identity in (\ref{regularizationP}) was proven above, while the second equality comes from the fact that $\gG_\varepsilon\cdot\cF(f)=\cF(f*h_\varepsilon)$. From Proposition \ref{defiregularization} we know that $\gR_\varepsilon(P_\pm(\,\cdot\,,f))\in\scH_m^\pm$. Moreover, because the convolution of smooth functions is again a smooth function, we get $\gR_\varepsilon(P_\pm(\,\cdot\,,f))\in\mathcal{C}^\infty(\bR^3,\bC^4)$, concluding the proof of (1).
	
	Let us pass to the  proof of point (ii). Consider any function $h_1\in \mathcal{C}_0^\infty(\bR^{1,3},[0,\infty))$ which is supported within $B(0,1/2)$ and with the additional symmetry $h_1(-x_0,\mathrm{R}\V{x})=h_1(x_0,\V{x})$ for any $x\in\bR^{1,3}$ and $\mathrm{R}\in \mathbb{O}(3)$. It can be proven by direct inspection that the smooth function $h:=h_1*h_1$ is supported in the set $B(0,1)$, it takes values within the positive real line and has the same symmetries of $h_1$. Moreover, by means of a suitable renormalization, we can always suppose that $\|h\|_{\mathcal{L}^1}=1$.  Now, consider its Fourier Transform $\gG:=(2\pi)^2\cF(h)$. From the invariance of $h_1$ with respect to spacetime inversion $x\mapsto  -x$, it follows that $\cF(h_1)$ is real-valued and, therefore, since $\gG=(2\pi)^4\,\cF(h_1)\cdot\cF(h_1)$, we see that the function $\gG$ takes values within $[0,\infty)$.   Finally, it can be checked by direct inspection that the symmetries of $h$ are preserved in defining $\gG$. Therefore, $\gG$ defines a four-dimensional cutoff function and we can then consider the regularization family $\gG_\varepsilon(k):=\gG(\varepsilon k)$.  Now, it can be seen by direct inspection that the function $h_\varepsilon(x):=\varepsilon^{-4}h(x\varepsilon^{-1})$ fulfills point (i)-(iii) in the thesis and that $\cF(h_\varepsilon)=(2\pi)^{-2}\gG_{\varepsilon}$.

	To conclude we need to prove the last statement. Since $h_{\varepsilon}$ is compactly supported,  Paley-Wiener Theorem (see for example Theorem 4.9 of \cite{SW}) ensures that the Fourier Transform $\gG_{\varepsilon}$ is the restriction to $\bR^{1,3}$ of an entire function on $\bC^4$ and, as such, it must be real analytic. 
	Now, consider the function $G\in\mathcal{C}^\infty((0,\infty),\bR)$ defined by $G(r):=\gG_{\varepsilon}(\omega(r\V{e}_1),r\V{e}_1)=\mathfrak{g}_{\varepsilon}(r\V{e}_1)$. This is also real analytic, it being the composition of the real analytic function $\gG_{\varepsilon}$ and the real analytic function $(0,\infty)\ni r\mapsto (\sqrt{r^2+m^2},r,0,0)\in\bR^{1,3}$ (see Proposition 2.2.8 of \cite{KP}). As such, the function $G$ can be extended analytically to a complex analytic (homeomorphic) function on an open neighborhood of $(0,\infty)$ in the complex plane. It is well-known that such functions admit only isolated zeroes and therefore the same must hold true for the restriction $G$ on $(0,\infty)$, as well. In particular, the set $N$ of the zeroes of $G$ is at most countable. 
	At this point, bearing in mind the symmetries of $\mathfrak{g}_{\varepsilon}$ (or equivalently $\gG_{\varepsilon}$), it follows that such a function vanishes exactly on the spheres $\partial B(0,R)$ with $R\in N$ (and maybe at $\V{k}=0$). The union of all these spheres is a Lebesgue null measure set of $\bR^3$. Exploiting point (iii) of Proposition \ref{defiregularization}, we see that $\ker\gR_\varepsilon=\{0\}$
\end{proof}

The properties in (ii)-(1) of the previous proposition are exactly those characterizing \textit{mollifiers} in $\bR^{1,3}$ (see \cite[Section 1.6]{ziemer}). Let us denote the set of such functions by~$\mathscr{M}(\bR^{1,3})$.

\begin{lemma}\label{lemmamollifieroperator}
	Let $h\in \mathscr{M}(\bR^{1,3})$. Then the following statements are true.
	\begin{itemize}[leftmargin=2.5em]\vspace{0.08cm}
		\item[\rm{(i)}] Let $f\in\Sol$, then $f*h\in\Sol$\\[-0.8em]
		\item[\rm{(ii)}] Let $u\in\sol$, then $u*h\in \sol\cap\mathcal{C}^\infty(\bR^{1,3},\bC^4)$\\[-0.8em]
		\item[\rm{(iii)}] The linear operator $\sol\ni u\mapsto u*h\in\sol$ is continuous.
	\end{itemize} 
\end{lemma}
\begin{proof}
	See Appendix.
\end{proof}

At this point, if we stick to regularization families as  in (ii) of Proposition \ref{convoltioncomesabout}, it is possible to complete the discussion started with equation (\ref{regsmooth2}). Indeed, equation (\ref{regularizationP}) shows how the regularization operator acts in position space only for smooth solutions like $P(\,\cdot\,,f)$. What about the other elements of $\scH_m$?

In the special case when the Fourier Transform of the regularization cutoff is a mollifier, it is possible to extend the firt identity in (\ref{regularizationP}) to the whole space $\scH_m$.
Indeed, since we already know that $\gR_\varepsilon f = f*h_\varepsilon$ on $\Sol$ and both  $\gR_\varepsilon$ and  $\cdot*h_{\varepsilon}$ are linear continuous operators, they must coincide on the whole Hilbert space.
\begin{theorem}\label{teoremaformaregol}
	Let $\mathfrak{g}_\varepsilon$ be is as in point (ii) of Proposition \ref{convoltioncomesabout}. Then 
	$$
	\gR_\varepsilon u= u*h_{\varepsilon}\quad\mbox{for all}\quad u\in\sol
	$$  	
	In this case the regularization is said to be of \textbf{mollification type}.
\end{theorem}

\begin{assumption}\label{assumpformarego}
	In the remainder of the paper we will always assume that a regularization family is fixed. Moreover, if not stated otherwise, we will always assume that it is of mollification type and thus of the form presented in Theorem \ref{teoremaformaregol}. 
\end{assumption}

Of course, the regularization should return the original function in the limit $\varepsilon\searrow 0$. This is true, even though in general just in a distributional sense. 
Moreover, if we apply the regularization operators to smooth functions which ``do not vary too much" on the scale $\varepsilon$, we expect these functions to change ``just slightly". 
In order to make this mathematically more precise and quantitative, we define for any $u\in \mathcal{C}^\infty(\bR^{1,3},\mathbb{C}^4)$
\begin{equation}\label{jacobian}
\|\mathfrak{J} u\|_{x,\infty}:=\sup_{z\in B(x,\varepsilon)}\left(\sum_{\mu=0}^3|\nabla \Re u_\mu(z)|_{\bR^{1,3}}+\sum_{\mu=0}^3|\nabla \Im u_\mu(z)|_{\bR^{1,3}}\right).
\end{equation}
\begin{proposition}\label{approximationpropertiesreg}
	The following properties hold for $x\in\bR^{1,3}$ and $u\in\sol$\footnote{By $\|f\|_\infty$ we denote the uniform norm, defined by $\sup\{|f(x)|\,|\, x\in \mbox{dom}(f)\}$}.
	\begin{itemize}[leftmargin=2.5em]
		\vspace{0.1cm}
		\item[\rm{(i)}]  $\lim_{\varepsilon\searrow 0}\gR_\varepsilon u(x)= u(x)$ if $u\in\mathcal{C}^\infty(\bR^{1,3},\bC^4)$\\[-0.6em]
		\item[\rm{(ii)}] $\lim_{\varepsilon\searrow 0 }\gR_{\varepsilon} u=u$\\[-0.6em]
		\item[\rm{(iii)}] $\lim_{\varepsilon\searrow 0}(\eta\,|\,\gR_\epsilon u)_{\mathcal{L}^2}=(\eta|u)_{\mathcal{L}^2}$ for any $\eta\in \mathcal{C}_0^\infty(\bR^{1,3},\bC^4)$\\[-0.6em]
		\item[\rm{(iv)}] $|\gR_\varepsilon  u(x)|\le \pi\varepsilon^{5/2}\|h_\varepsilon\|_\infty\|u\|_m$ for every $x\in\bR^{1,3}$\\[-0.6em]
		\item[\rm{(v)}] $|\gR_\varepsilon  u(x)-u(x)|\le \varepsilon \|\mathfrak{J} u\|_{x,\infty}$ if $u\in \mathcal{C}^\infty(\bR^{1,3},\bC^4)$\\[-0.6em]
		\item[\rm{(vi)}] $\|\gR_\varepsilon\|\le 1$.
	\end{itemize}
\vspace{0.1em}
	Point (ii) shows that $\gR_{\varepsilon}\to \bI$ in the strong topology as $\varepsilon\searrow 0$. Point (iv) shows that $\gR_\varepsilon $ defines a continuous functional  at every point $x\in\bR^{1,3}$.
\end{proposition}
\begin{proof}
	Point (i) follows from (ii) of \cite[Theorem 1.6.1]{ziemer}. Let us prove point (ii). 
	Fix any $u\in\scH_m$ and $\delta>0$ and $T=1/2$. By Lemma \ref{lemmaconverging2} we know that $u\!\!\restriction_{R_T}\in\mathcal{L}^2(R_T,\bC^4)$. By density of $\mathcal{C}_0^\infty(R_T,\bC^4)$ within $\mathcal{L}^2(R_T,\bC^4)$ there exists $v\in\mathcal{C}_0^\infty(R_T,\bC^4)$ such that $\|u\!\!\restriction_{R_T}-v\|_{\mathcal{L}^2}<\delta$. Now, exploiting Lemma \ref{lemmaconverging2}  and Lemma \ref{lemmatecnico}, we get
	\begin{equation*}
	\begin{split}
	\|\gR_{\varepsilon}u-u\|_m&=\|(\gR_{\varepsilon}u-u)\!\restriction_{R_T}\|_{\mathcal{L}^2}=\\
	&\le \|(\gR_{\varepsilon}(u-v)\!\restriction_{R_T}\|_{\mathcal{L}^2}+\|(\gR_{\varepsilon}v-v)\!\restriction_{R_T}\|_{\mathcal{L}^2}+\|(u-v)\!\restriction_{R_T}\|_{\mathcal{L}^2}\le\\
	&\le\|(u-v)\!\restriction_{R_{T+\varepsilon}}\|_{\mathcal{L}^2}+\|\gR_{\varepsilon} v-v\|_{\mathcal{L}^2}+\|u\!\!\restriction_{R_T}-v\|_{\mathcal{L}^2}<\\
	&<\|(u-v)\!\restriction_{R_{T+\varepsilon}}\|_{\mathcal{L}^2}+\|\gR_{\varepsilon} v-v\|_{\mathcal{L}^2}+\delta
	\end{split}
	\end{equation*}
	where  we used that both $v$ and $\gR_{\varepsilon}v$ belong to $\mathcal{L}^2(\bR^{1,3},\bC^4)$ (see again \cite{ziemer}).
	Now, since $\supp v\subset R_T$, we have:
	\begin{equation*}
	\begin{split}
	\|(u-v)\!\restriction_{R_{T+\varepsilon}}\|_{\mathcal{L}^2}^2&=\int_{R_T}|u(x)-v(x)|^2\, d^4x+\int_{R_{T+\varepsilon}\setminus R_T}|u(x)-v(x)|^2\, d^4x=\\
	&= \|(u-v)\!\restriction_{R_T}\|_{\mathcal{L}^2}^2+\int_{R_{T+\varepsilon}\setminus R_T}|u(x)|^2\, d^4x=\\
	&= \|(u-v)\!\restriction_{R_T}\|_{\mathcal{L}^2}^2+2\varepsilon \|u\|_m^2=\\
	&\le \|u\!\restriction_{R_T}-v\|_{\mathcal{L}^2}^2+2\varepsilon \|u\|_m^2 < \delta^2 + 2\varepsilon \|u\|_m^2 ,
	\end{split}
	\end{equation*}
	where we made again use of Lemma 2.7. Putting all together we have just proven that
	$$
	\|\gR_{\varepsilon}u-u\|_m\le 2\delta + \|\gR_{\varepsilon} v-v\|_{\mathcal{L}^2} + \sqrt{2\varepsilon} \|u\|_m.
	$$
	The arbitrariness of $\delta$ and the fact that  $\|\gR_{\varepsilon} v-v\|_{\mathcal{L}^2} \to 0$ (see Theorem 1.6.1 \cite{ziemer}) concludes the proof.
	The proof of points (iii) and (iv) (with minor adjustments) can be found in Example 1.2.4 of \cite{FF}. So, let us prove point (v). 	
	\begin{equation*}
	\begin{split}
	|\gR_\varepsilon  u(x)-u(x)|&=\left|\int_{B_\varepsilon(x)}h_\varepsilon(y)(u(x\!-\!y)\!-\!u(x))\, d^4y\right|\le\\ &\le\int_{B_\varepsilon(x)}h_\varepsilon(y)|u(x\!-\!y)\!-\!u(x)|\, d^4y\le \\
	&\le \int_{B_\varepsilon(x)}h_\varepsilon(y)\|\mathfrak{J} u\|_{x,\infty} | y|\, d^4y\le \varepsilon\|\mathfrak{J} u\|_{x,\infty}
	\end{split}
	\end{equation*}
	where we applied a multi-variable version of the mean value theorem (see \cite{apostol}, Example 2 after Theorem 12.9). To conclude, let us prove point (vi). Choose $T>0$, then from Lemma \ref{lemmaconverging2} and Lemma \ref{lemmatecnico} we get for any $u\in\scH_m$:
	$$
	\|\gR_\varepsilon u\|_m=\sqrt{2T}\|\gR_{\varepsilon} u \restr_{R_T}\|_{\mathcal{L}^2}\le \sqrt{2T}\|u\restr_{R_{T+\varepsilon}}\|_{\mathcal{L}^2}=\frac{\sqrt{2T}}{\sqrt{2(T+\varepsilon)}}\|u\|_m\le \|u\|_m.
	$$
\end{proof}

\subsection{The Regularized Fermionic Projectors}

Referring to point (i) of Proposition \ref{convoltioncomesabout}, we can now regularize the fermionic projectors introduced in Definition \ref{deffermionicprojectors}.

\begin{definition}
	The \textbf{regularized fermionic projectors onto the positive and negative spectrum} are defined as the linear mappings
	\begin{equation}
	P_{\pm}^\varepsilon: \mathcal{S}_x(\bR^{1,3},\bC^4)\ni f\mapsto P_\pm(\,\cdot\,,f*h_\varepsilon)\in \scH_m^\pm\cap \mathcal{C}^\infty(\bR^{1,3},\bC^4).
	\end{equation}
	Similarly, we define the \textbf{regularized causal propagator} by $P_{c}^\varepsilon:=P_{-}^\varepsilon-P_{+}^\varepsilon$.
\end{definition}

As in the unregularized case, it is possible to represent the regularized fermionic projectors as kernel operators, 
$$
P_{\pm}^\varepsilon(x,f)=\int_{\bR^4}P_{\pm}^\varepsilon(x,y)  f(y)\, d^4y,
$$
where, in this case, the kernel does define a regular function on $\bR^{1,3}\times\bR^{1,3}$.

\begin{proposition}\label{rfp}
	The \textbf{kernel of the regularized fermionic projector} $P_{\pm}^\varepsilon$ is given by
\begin{equation*}
P_{\pm}^\varepsilon(x,y):=\int_{\bR^4}\frac{d^4k}{(2\pi)^4}\, \hat{P}_-(k)\,\gG_\varepsilon(k)\,e^{-i\eta(x-y,k)}.
\end{equation*}
	and has the following representation:
	\begin{equation}
	P_{\pm}^\varepsilon(x,y)= \pm\int_{\bR^3}\frac{d^3\V{k}}{(2\pi)^{4}}\, \mathfrak{g}_\varepsilon(\V{k})\, p_\pm(\V{k})\,\gamma^0\, e^{-i(\pm\omega(\V{k})(t_x-t_y)-\V{k}\cdot(\V{x-y}))}.
	\end{equation}
	Moreover, for every $y\in\bR^{1,3}$ and $a\in\bC^4$,
	\begin{equation}\label{pickedsolutions}
	P_{\pm}^\varepsilon(\,\cdot\,,y)a= P_\pm(\,\cdot\,,T_y(h_\varepsilon\, a))\in \sol^\pm\cap \mathcal{C}^\infty(\bR^{1,3},\bC^4),
	\end{equation}
	where $T_y(f)(x):=f(x-y)$ is the displacement operator.
	
\end{proposition}
\begin{proof}
	The fact that $P_{\pm}^\varepsilon(x,y)$ does define an integral kernel for $P_{\pm}^\varepsilon$ can be proved by direct inspection.  The smoothness of $P_{\pm}^\varepsilon(\,\cdot\,,y)$ and the fact that it solves the Dirac equation follow from the fact that $\gG_\varepsilon$ belongs to $\mathcal{S}_p(\bR^{1,3},\bC^3)$ together with Proposition \ref{generalsolutionsmooth}.  Finally, identity (\ref{pickedsolutions}) follows directly by plugging the identity $$\cF(T_y(h_\varepsilon\,a))(k)=(2\pi)^{-2}\,\gG_\varepsilon (k)a\,e^{i\eta(y,k)}$$ 
	into expression (\ref{fermionicprojdef}).
\end{proof}
\begin{remark}
	If we compare the proposition above with point (v) of Remark \ref{defPk}, we see that the ``maximally localized'' distributional solutions of (\ref{solutiondirac}) are now replaced by the regular  solutions determined by the initial data on $\Sigma_0$:
	\begin{equation}\label{loca2}
	\Omega_\varepsilon^\pm(a,\V{y};\V{x}):=\pm \int_{\bR^3}\frac{d^3\V{k}}{(2\pi)^4}\,\mathfrak{g}_\varepsilon(\V{k})\,p_\pm(\V{k})(\gamma^0 a)\, e^{i\V{k}\cdot(\V{x}-\V{y})}.
	\end{equation}
	Note that $\Omega_\varepsilon^\pm(a,\V{y};\,\cdot\,)\in\mathcal{S}_p(\bR^3,\bC^4)$. The corresponding solutions $P_{\pm}^\varepsilon(\,\cdot\,,y)a$ are  peaked around the light-cone centered at $(0,\V{y})$, but are regular and do not diverge anywhere. 
\end{remark}

%

\subsection{The Doubly-Regularized Kernel of the Fermionic Projector}

Given its importance in the rest of the paper, we introduce also the  double regularization of the fermionic projectors.
\begin{definition}\label{rfp2}
	The \textbf{doubly-regularized kernel of the fermionic projector} $P_\pm$ is defined by\footnote{In order to avoid confusion, we stress that the notation $2\varepsilon$ indicates that the regularization is performed two times on the kernel ($\gG_\varepsilon\to \gG_\varepsilon^2$) and it should not be confused with  $\gG_\varepsilon\to \gG_{2\varepsilon}$.}
	\begin{equation}\label{doubleregularizeddistribu}
	P_{\pm}^{2\varepsilon}(x,y):=\int_{\bR^4}\frac{d^4k}{(2\pi)^4}\, \hat{P}_-(k)\,\gG_\varepsilon(k)^2\,e^{-i\eta(x-y,k)}
	\end{equation}
	for any couple $x,y\in\bR^{1,3}$. As always, we define 
	$$
	P_{c}^{2\varepsilon}(x,y):=P_{-}^{2\varepsilon}(x,y)-P_{+}^{2\varepsilon}(x,y).
	$$ 
	The \textbf{doubly-regularized fermionic projectors} are defined for any $f\in\mathcal{S}_x(\bR^{1,3},\bC^4)$ by
	$$
	P_{\pm}^{2\varepsilon}(x,f):= \int_{\bR^4}P_{\pm}^{2\varepsilon}(x,y)f(y)\, d^4y,\quad P_{\pm}^{2\varepsilon}(\,\cdot\,,f)\in\scH_m^\pm\cap\mathcal{C}^\infty(\bR^{1,3},\bC^4).
	$$
\end{definition}

Let us analyze some features of these kernels.
\begin{proposition}\label{double regularazione}
	Referring to Definition \ref{rfp2} the following properties hold.
	\vspace{0.1cm}
	\begin{itemize}[leftmargin=2.5em]
		\item[\rm{(i)}] For any $a\in\bC^4$, 
		$$
		P_{\pm}^{2\varepsilon}(\,\cdot\,,y)a=\gR_{\varepsilon} (P_{\pm}^{2\varepsilon}(\,\cdot\,,y)a).
		$$\\[-2em]
		\item[\rm{(ii)}] The linear function $\bC^4\ni a\mapsto P_{\pm}^{2\varepsilon}(\,\cdot\,,y)a$ is injective.
	\end{itemize}
\end{proposition}
\begin{proof}
	Point (i) can be proved by direct inspection exploiting the definitions given so far.  Let us pass to point (ii). Assume that $P_{\pm}^{2\varepsilon}(\,\cdot\,,y)a=0$ for some $a\in\bC^4$. If we evaluate this function at $y$, i.e. $P_{\pm}^{2\varepsilon}(y,y)a=0$, we get rid of the exponential in (\ref{doubleregularizeddistribu}). At this point, making use of the three-dimensional representation of the distributional integrals  (see (\ref{abstractrestatement})) and noticing that $(\gamma^0 a)^\dagger\, p_\pm(\V{k})\,\gamma^0a\ge 0$ as well as $\gG_{\varepsilon}(k)^2\ge 0$, we conclude that:
	$$
	\mathfrak{g}_\varepsilon(\V{k})^2\  (\gamma^0 a)^\dagger\, p_\pm(\V{k})(\gamma^0 a) =0 \mbox{ for any } \V{k}\in\bR^3.
	$$
 In the proof of Proposition \ref{convoltioncomesabout} it was shown that the zeros of $\mathfrak{g}_\varepsilon$ lie in the set $\cup_{n}\partial B(0,R_n)\cup\{0\}$, with $\{R_n\}_n$ isolated positive real numbers. Without loss of generality we can suppose that $\mathfrak{g}_\varepsilon\neq 0$ in an open neighborhood of the origin $B(0,r)$. As a consequence, we get $(\gamma^0 a)^\dagger p_\pm(\V{k})(\gamma^0 a)=0$ for any $\V{k}\in B(0,r)$. Since $p_\pm$ is a projector, it follows that $p_\pm(\V{k})(\gamma^0 a)=0$, or, equivalently, that $\gamma^0a\in W^\mp_{\V{k}}$ for any $\V{k}\in B(0,r)$. Let us now focus on the $-$-case, the other one being equivalent. Exploiting (\ref{fundamentalspinors}) we see that there exist scalars $\lambda_{\uparrow,\downarrow}$ such that:
	\begin{equation}
	\begin{split}
	\lambda_\uparrow(0)\chi_\uparrow^+(0)+\lambda_\downarrow(0)\chi_\downarrow^+(0)\stackrel{(1)}{=}\gamma^0a\stackrel{(2)}{=}\lambda_\uparrow(\V{k})\chi_\uparrow^+(\V{k})+\lambda_\downarrow(\V{k})\chi_\downarrow^+(\V{k})\quad\mbox{for all }\V{k}\in\bR^3
	\end{split}
	\end{equation}
	Comparing the 0th and 1st components, we get $\lambda_{\uparrow,\downarrow}(\V{k})=\lambda_{\uparrow,\downarrow}(0)=:\lambda_{\uparrow,\downarrow}$, which shows that the scalars do not depend on the point $\V{k}$. Thus, choosing for example $\V{k}=(r/2,0,0)$ and comparing the 2nd and 3rd components, we get $\lambda_{\uparrow,\downarrow}=0$. This gives $\gamma^0 a=0$ and therefore $a=0$, concluding the proof.
\end{proof}

To conclude this section, in view of what follows, it is useful to compute explicitly the form of the regularized kernel of the fermionic projector on the diagonal $x=y$.  Exploiting (\ref{abstractrestatement}), it follows that
\begin{equation*}
\begin{split}
P_{\pm}^{2\varepsilon}(x,x)&=\pm\int_{\bR^3}{\frac{d^3\V{k}}{(2\pi)^{4}}}\,\mathfrak{g}_\varepsilon(\V{k})^2\, p_\pm(\V{k})\gamma^0=\\ &=\frac{1}{2}\int_{\bR^3}{\frac{d^3\V{k}}{(2\pi)^4}}\,\mathfrak{g}_\varepsilon(\V{k})^2 \left(\pm\gamma^0-\frac{\V{k}\cdot \boldsymbol{\gamma}}{\omega(\V{k})}+\frac{m}{\omega(\V{k})}\right)
\end{split}
\end{equation*}
Exploiting the rotational invariance of the function $\mathfrak{g}_\varepsilon$, we see that the term involving the matrices $\gamma^i$ vanishes.
\begin{proposition}\label{propositionkerneldiagonal}
	For any $x\in\bR^3$, the doubly-regularized  kernels read
	\begin{equation}\label{diagonalformp}
	P_{\pm}^{2\varepsilon}(x,x)=\frac{1}{2(2\pi)^{4}}\left( m\left\|\frac{\mathfrak{g}_\varepsilon^2}{\omega}\right\|_{\mathcal{L}^1}\!\bI_4\pm\|\mathfrak{g}^2_\varepsilon\|_{\mathcal{L}^1}\gamma^0\right).
	\end{equation}
	The spectrum\footnote{In order to avoid confusion, we stress that here we are not referring to the spectrum of the integral operator $P_\pm^{2\varepsilon}$, but to the spectrum (set of eigenvalues) of the $4\times 4$ matrix $P^{2\varepsilon}_\pm(x,x)$ \textit{at fixed} $x\in\bR^{1,3}$.}  of the matrix $P_\pm^{2\varepsilon}(x,x)$ consists of two elements, both with multiplicity two:
	\vspace{-0.35cm}
	\begin{equation}\label{spectrumP}
	\begin{split}
	\nu^{\pm}_+(\varepsilon)&:= \frac{1}{2(2\pi)^{4}}\left(m\left\|\frac{\mathfrak{g}_\varepsilon^2}{\omega}\right\|_{\mathcal{L}^1}+\|\mathfrak{g}^2_\varepsilon\|_{\mathcal{L}^1} \right)>0\\
	\nu^{\pm}_-(\varepsilon)&:=\frac{1}{2(2\pi)^{4}}\left(m\left\|\frac{\mathfrak{g}_\varepsilon^2}{\omega}\right\|_{\mathcal{L}^1}- \|\mathfrak{g}^2_\varepsilon\|_{\mathcal{L}^1} \right)<0.
	\end{split}
	\end{equation}
In particular, the spectrum does not depend on $x$.
\end{proposition}
\begin{proof}
	The explicit form of the eigenvalues follow directly from the diagonal expression of the kernel. In order to check the sign of $\nu^{-}$, notice first that $1>m\omega^{-1}>0$ on $\bR^3\setminus\{0\}$ and therefore $\|\mathfrak{g}^2\|_{\mathcal{L}^1}\ge \|m\omega^{-1}\mathfrak{g}^2\|_{\mathcal{L}^1}$. We now claim that this inequality is in fact strict. Indeed, from $\|\mathfrak{g}^2\|_{\mathcal{L}^1}= \|m\omega^{-1}\mathfrak{g}^2\|_{\mathcal{L}^1}$ it follows that $\|(1-m\omega^{-1})\mathfrak{g}^2\|_{\mathcal{L}^1}=0$, and therefore $(1-m\omega^{-1})\mathfrak{g}^2=0$ almost everywhere. Now, since $(1-m\omega^{-1})>0$ almost everywhere, this implies that $\mathfrak{g}$ vanishes almost everywhere and therefore $\mathfrak{g}=0$, by continuity. This is not possible.
\end{proof}
\begin{remark} A few remarks follow.\\[-1em]
	\begin{itemize}[leftmargin=2.5em]
		\item[\rm{(i)}] Note that the limit $\varepsilon\searrow 0$ is generally ill-defined because the integrals diverge.\\[-0.5em]
		\item[\rm{(ii)}] In the intuitive picture of $\mathfrak{g}_\varepsilon$ as the characteristic function of $B(0,\varepsilon^{-1})$, the kernel (\ref{diagonalformp}) and the eigenvalues (\ref{spectrumP}) can be calculated explicitly using:
		\begin{equation*}\label{approssimazionecutof}
		\begin{split}
		m\left\|\frac{\mathfrak{g}_\varepsilon^2}{\omega}\right\|_{\mathcal{L}^1}&=2\pi\,m^3\left(\sqrt{\frac{1}{(m\varepsilon)^4}+\frac{1}{(m\varepsilon)^2}}-\sinh^{-1}\left[\frac{1}{m\varepsilon}\right]\right),\\ \|\mathfrak{g}_\varepsilon^2\|_{\mathcal{L}^1}&=\frac{4\pi\,m^3}{3}\frac{1}{(m\varepsilon)^3}.
		\end{split}
		\end{equation*}
		Expanding the first term of in powers of $m\varepsilon$ (compare with Assumption \ref{assumptionepsilon}), the leading-order terms of the kernel and its eigenvalues are, respectively,
		\begin{equation}\label{stime}
		\begin{split}
		&P_{\pm}^{2\varepsilon}(x,x)\cong \frac{1}{2(2\pi)^3}\left(\frac{m}{\varepsilon^2}\,\bI_4\pm \frac{2}{3\,\varepsilon^{3}}\,\gamma^0 \right),\\
		\nu^{\pm}_+(\varepsilon)\cong \frac{1}{2(2\pi)^3}&\left(\frac{m}{\varepsilon^2}+ \frac{2}{3\,\varepsilon^{3}}\right),\quad
		\nu^{\pm}_-(\varepsilon)\cong \frac{1}{2(2\pi)^3}\left(\frac{m}{\varepsilon^2}- \frac{2}{3\,\varepsilon^{3}} \right).
		\end{split}
		\end{equation}
		These approximate identities give some rough intuition on the dependence of the doubly-regularized kernel on the microscopic regularization length. It should be kept in mind, though, that in this work the regularization function $\mathfrak{g}_\varepsilon$ is chosen to be a smooth rapidly decaying function, and not as a brute discontinuous cutoff.
	\end{itemize}
	
\end{remark}

As a conclusion of this section, we point out another result, which will prove important later on. The leading question is:
\begin{center}
	\textit{How does the spectrum of $P_{-}^{2\varepsilon}(x,x)$ deviate from (\ref{spectrumP}) when a set of positive- (negative-)energy physical solutions is added to (removed from) the system?}
\end{center}
The full meaning of this question will become clear later on, once the concept of a vacuum causal fermion system has been introduced. For the moment, let us
state  the following result, whose proof can be found in the appendix.

We denote the usual non-degenerate
indefinite inner product\footnote{An \textit{indefinite inner product} on a complex vector space is a sesquilinear form $\phi$ s.t.$\phi(a,b)=\overline{\phi(b,a)}$.} on spinors by
\begin{equation}\label{spinscalarproduct}
\Sl \psi|\varphi\Sr:= \psi^\dagger \,\gamma^0\,\varphi.\footnote{In physics textbooks the spin inner product is often denoted by $\overline{\psi}\phi$.}
\end{equation}
and refer to it as the \textbf{spin scalar product}.
\begin{proposition}\label{representationP}
	Let $\{u^\pm_n\}_{n\in\bN}\subset \hat{\mathrm{E}}(\hat{P}_\pm(\mathcal{S}_p(\bR^3,\bC^4)))$ be a Hilbert basis of $\sol^\pm$. Then
	\begin{equation}\label{sumrepresentation}
	P_{\pm}^{2\varepsilon}(x,y)=\pm\frac{1}{2\pi}\sum_{n\in\bN}\gR_\varepsilon u^\pm_n(x)\,\Sl\gR_\varepsilon u^\pm_n (y)\,|\,\cdot\Sr
	\end{equation}
	for every couple of spacetime points $x,y\in\bR^{1,3}$. In particular, we have:
	\begin{equation}
	P_{c}^{2\varepsilon}(x,y)=-\frac{1}{2\pi}\sum_{n\in\bN}\gR_\varepsilon u^-_n(x)\,\Sl\gR_\varepsilon u^-_n (y)\,|\,\cdot\Sr-\frac{1}{2\pi}\sum_{n\in\bN}\gR_\varepsilon u^+_n(x)\,\Sl\gR_\varepsilon u^+_n (y)\,|\,\cdot\Sr
	\end{equation}
\end{proposition}

Let us go back to the fermionic projector onto the negative spectrum. The addition of positive-energy (removal of negative-energy) physical solutions consists in the addition (subtraction) of terms of the form 
$$
-(2\pi)^{-1}\,\gR_\varepsilon u (x)\,\Sl\gR_\varepsilon u(x)\,|\,\cdot\Sr
$$ 
in (\ref{sumrepresentation}) (See Section \ref{sectioninterpretation}, in particular Proposition \ref{representationgenerickernel}).

\begin{proposition}\label{propositionperturb}
	Let $\{e_{i}^\pm\}_{i=1,\dots, N_\pm}$ be finite orthonormal sets  in $\sol^\pm$. Then the eigenvalues ${\nu}$ of the matrix 
	$$
	P_{-}^{2\varepsilon}(x,x) + \Delta P(x,x),\quad  
	$$
	with the perturbation
	$$
	\Delta P(x,x):= -\frac{1}{2\pi}\sum_{i=1}^{N_+} \gR_\varepsilon e_i^+(x)\,\Sl\gR_\varepsilon e_i^+(x)\,|\,\cdot\Sr+\frac{1}{2\pi}\sum_{i=1}^{N_-} \gR_\varepsilon e_i^-(x)\,\Sl\gR_\varepsilon e_i^-(x)\,|\,\cdot\Sr,
	$$
	fulfill the constraint\footnote{Given a matrix $A\in \sM(n,\bC)$ we define the norm $\|A\|_2:=\sup_{x\neq 0} |Ax|/|x|$.}
	$$
	\min\{|{\nu}-\nu_{-}^+(\varepsilon)|,|{\nu}-\nu_{-}^-(\varepsilon)|\}\le \|\Delta P(x,x)\|_2.
	$$
\end{proposition}
\begin{proof}
	Exploiting identity (\ref{diagonalformp}) we see that the matrix $P_{-}^{2\varepsilon}(x,x)$ is symmetric. More precisely it is real and diagonal with two non-vanishing eigenvalues $\nu_{-}^\pm(\varepsilon)$ both with multiplicity two. The thesis follows then from Bauer-Fike Theorem (see Theorem IIIa in \cite{BF}).
\end{proof}

\begin{notation}
	For simplicity of notation, we will drop the subscript $-$ from  $P_{-}$ from now on, as we will focus exclusively on the negative spectrum.
\end{notation}

%
%
%
%
%

\section{The Emergence of Causal Fermion Systems}\label{sectionCFS}

Before entering the construction of causal fermion systems in Minkowski space, we need to introduce some preliminary definitions and results.

\subsection{The General Mathematical Set-Up of the Theory}\label{sectiondefcfs}

In this section we introduce the basic set-up of the theory.

\begin{definition}
	Let $(\mathscr{H},\langle\cdot|\cdot\rangle)$ be a Hilbert space. We denote by $\mathscr{F}(\mathscr{H})$ the family of bounded self-adjoint operators on $\mathscr{H}$ which have - counting  multiplicities - at most two positive and two negative eigenvalues. 
\end{definition}
Note that the set $\mathscr{F}(\mathscr{H})$ does not inherit any linear structure from $\scB(\scH)$: the only allowed operation is multiplication by real numbers.



\begin{theorem}\label{closednessF}
	The family $\mathscr{F}(\mathscr{H})$ is a closed double-cone\footnote{A double-cone of a linear space $V$ is a subset which is closed under multiplication by real scalars.} of $\mathcal{B}(\mathscr{H})$.
\end{theorem}
\begin{proof}\footnote{The proof of this theorem is based on  \cite[Chapter V Section 4.3]{kato}. I would like to thank Christoph Langer for pointing this out.  For a deeper analysis on this topic and its connections to theory of causal fermion systems the interested reader is referred to his forthcoming PhD thesis \cite{christoph}}
	The fact that $\scF(\mathscr{H})$ is a double-cone is obvious, as $\sigma(\lambda F)=\lambda\,\sigma(F)$ for any $\lambda\in\bR$ and $F\in \scF(\mathscr{H})$. So, let us pass to the proof of closedness.
	First, notice that the set of operators with rank at most four is closed in the weak operator topology and therefore, if we have any sequence $F_n\in \scF(\mathscr{H})$ which converges to $F\in\scB(\mathscr{H})$ in the uniform topology, then $\ran F\le 4$. Moreover, $F$ is self-adjoint, it being uniform limit of self-adjoint operators.
	Now, let $0<\epsilon<1$ be arbitrarily small and take $N\in\bN$ such that $\|F-F_n\|<\epsilon$ for any $n\ge N$. If we define $T:=F$, $A:=F_n-F$ and $S=T+A=F_n$, then we can follow the discussion in Chapter V, Section 4.3 of \cite{kato}. First of all, notice that the operator $A$ is $T$-bounded with vanishing $T$-bound:
	$$
	\|Au\|=\|(F_n-F)u\|\le \|F_n-F\|\|u\|<\epsilon\|u\|+0\|Tu\|\quad \mbox{for all }u\in\mathscr{H}.
	$$
	In order to match with the notation in \cite{kato}, we prefer not to repeat the eigenvalues according to their multiplicity. Therefore, we assume by contradiction that there exists $0<k\le 4$ strictly positive eigenvalues $\lambda_1,\dots,\lambda_k$ of $T=F$ with multiplicities $m_i$ such that $m_1+\cdots+m_k\ge 3$. Let $d_i>0$ denote the isolation distance (defined in \cite{kato}) for the eigenvalue $\lambda_i$. Choosing $\epsilon$ sufficiently small, we see that inequality (4.11) in \cite{kato} is trivially satisfied in our case (notice that the distance $d_i$ depends only on $F$ and not on $\epsilon$). The discussion therein shows that the total multiplicity of the eigenvalues of $S=F_n$ which lies in the interval $(\lambda_i-d_i/2,\lambda_i+d_i/2)$ is exactly $m_i$. Moreover, by definition of isolation distance, the intervals $(\lambda_i-d_i/2,\lambda_i+d_i/2)$ are disjoint from each other and $0$ does not belong to any of them, it being an eigenvalue of $T=F$, too. Putting all together we see that the total multiplicity of the eigenvalues of $S=F_n$ in the positive axis $(0,\infty)$ is at least three, which is impossible by definition of $\scF$. The case of negative eigenvalues is analogous.
\end{proof}

We are ready to give the fundamental definition. 

\begin{definition}\label{definitioncfs}
	A \textbf{causal fermion system} is a couple $(\mathscr{H},\varrho)$ where $\mathscr{H}$ is a Hilbert space and $\varrho$ is a Borel measure on $\scF(\mathscr{H})$.
\end{definition}

\begin{remark}
	The support of the measure $\varrho$ plays an important role in the theory of causal fermion systems, in that it is believed to describe  physical spacetime when the measure arises as a minimizer of a specific action (for more information see the \textit{causal action principle} in \cite{FF}). In this paper we will not enter the details of this, we will simply show that there exists a canonical example of measure whose support realizes a smooth manifold diffeomorphic to Minkowski space.
\end{remark}

In the following sections we will deal with causal fermion systems constructed out of closed subspaces of $\sol$. As will become clear later, it is useful to study the relations which occur among such objects. More precisely: given two Hilbert spaces $\mathscr{H}_0\subset\mathscr{H}_1$, how are the corresponding spaces $\scF$ related to each other? Notice that we can always write $\mathscr{H}_1$ as $\mathscr{H}_0\oplus \mathscr{H}_0^\perp$. Bearing this in mind, we can state and prove the following result.

\begin{lemma}\label{lemma1}
	Let $\mathscr{H}_0$ be a closed subspace of a Hilbert space $\mathscr{H}_1$. Then the function
	\begin{equation}\label{phi}
	\iota:\scF(\mathscr{H}_0)\ni x\mapsto x\oplus 0\in\scF(\mathscr{H}_1)
	\end{equation}
	is a well-defined, one-to-one, norm-preserving (thus continuous) closed map. 
	In particular, by means of this identification, the set $\scF(\mathscr{H}_0)$ defines a closed subset of $\scF(\mathscr{H}_1)$ and the map $\iota$ is a homeomorphism onto its image. 
\end{lemma}
\begin{proof}
	It is not difficult to prove that the map is one-to-one and norm-preserving. In particular the map is continuous with respect to the uniform topology.
	Now we prove that the map is also closed. 
	Suppose first that $C$ is a closed subset of $\scF(\mathscr{H}_0)$ in the uniform topology and let  $\{\iota(x_n)\}\subset \iota(C)$ be any sequence converging to $T\in\scF(\mathscr{H}_1)$ in the uniform topology. It being convergent, the sequence $\{\iota(x_n) \}$ is of Cauchy type within $\scF(\mathscr{H}_1)$ and therefore the same holds true also for $\{x_n \}$, the map $\iota$ preserving the norm. Now, note that $\scF(\mathscr{H}_0)$ is complete in the uniform topology, it being a closed subset of $\scB(\mathscr{H}_0)$ (see Theorem \ref{closednessF}). Therefore, also $C$ is complete in the uniform topology, it being closed in $\scF(\mathscr{H}_0)$. As a consequence, there must exist some $x\in C$ such that $x_n\to x$. Since the map $\iota$ is continuous we get $T=\lim_{n\to\infty}\iota(x_n)=\iota(x)$ and so $T\in \iota(C)$, i.e. $\iota(C)$ is closed. The last point follows immediately from the fact that $\iota$ is injective, continuous and closed.
\end{proof}


Lemma \ref{lemma1} suggests that one can actually focus to the larger space $\scF(\mathscr{H}_1)$. This is indeed the case even in presence of a Borel measure, as we are going to explain. 
As already mentioned, not every operator in $\scF$ is of physical interest: what is relevant for the theory is what is contained in the support of some given Borel measure and this information is not loss when the identifications of Lemma \ref{lemma1} are taken into account. 

More precisely, consider a causal fermion system on the Hilbert spaces $\mathscr{H}_0$ with a measure $\varrho_0$ defined on the Borel $\sigma$-algebra of $\scF(\mathscr{H}_0)$. Suppose that $\mathscr{H}_0$ identifies itself as a Hilbert subspace of $\mathscr{H}_1$. We know by Lemma \ref{lemma1} that $\scF(\mathscr{H}_0)$ can be embedded within $\scF(\mathscr{H}_1)$, but what about the measure? How can we read $\varrho_0$ as a measure on $\scF(\mathscr{H}_1)$? The most natural thing to do is to push-forward it by means of the map $\iota$:
\begin{equation}\label{inclusion}
\iota_*\varrho_0: \mathfrak{Bor}(\scF(\mathscr{H}_1))\ni \Omega\mapsto \varrho_0(\iota^{-1}(\Omega))\in [0,\infty)
\end{equation}
Notice that this is well-defined,  as the map $\iota$ is continuous and therefore measurable.

\begin{lemma}\label{lemmasupp}
	$\iota(\supp\varrho_0)=\supp(\iota_*\varrho_0)$. 
\end{lemma}
\begin{proof}
	Thanks to Lemma \ref{lemma1}, the set $\iota(\supp\varrho_0)$ is closed within $\scF(\mathscr{H})$. Equivalently, $\scF(\mathscr{H})\setminus\iota(\supp\varrho_0)$ is open and
	\begin{equation*}
	\begin{split}
	\iota_*\varrho_0(\scF(\mathscr{H})\setminus\iota(\supp\varrho_0))&=\varrho_0(\iota^{-1}(\scF(\mathscr{H})\setminus\iota(\supp\varrho_0)))=\\
	&=\varrho_0(\iota^{-1}(\iota(\scF(\mathscr{H}_0))\setminus\iota(\supp\varrho_0)))=\\
	&=\varrho_0(\scF(\mathscr{H}_0)\setminus\supp\varrho_0))=0,
	\end{split}
	\end{equation*}
	where we used the fact that $\iota$ is injective.
	By definition of support of a measure, this means that all the points belonging to $\scF(\mathscr{H})\setminus \iota(\supp\varrho_0)$ cannot belong to the support of $\iota_*\varrho_0$, for the given set is open and has vanishing measure. More precisely, we have
	$
	\scF(\mathscr{H})\setminus\iota(\supp\varrho_0)\subset \scF(\mathscr{H})\setminus \supp\iota_*\varrho_0
	$
	or, equivalently, $\supp\iota_*\varrho_0\subset \iota(\supp\varrho_0)$. 
	On the contrary, take any $\iota(x)\in \iota(\supp\varrho_0)$ and let $A\subset \scF(\mathscr{H})$ be any open set such that $\iota(x)\in A$. Since the map $\iota$ is continuous, the set $\iota^{-1}(A)$ must be an open neighborhood of $x$ in $\scF(\mathscr{H}_0)$. Thus
	$
	\iota_*\varrho_0(A)=\varrho_0(\iota^{-1}(A))\neq 0.
	$
	Since the set $A$ is arbitrary, it must hold that $\iota(x)\in \supp\iota_*\varrho_0$.	
\end{proof}

This results makes the identification complete: also the measure can be lifted to a measure on the larger space and no information is lost. 
\begin{center}
\textit{In what follows, we will always make use of these identifications, if not stated otherwise.}
\end{center}

\subsection{The Emergence of Causal Fermion Systems}\label{emergencecfs}

We are now ready to see  how these structures arise in relativistic quantum mechanics, when a regularization is introduced.  In words, the fundamental result reads:
\begin{center}
	\textit{Every ensemble of physical solutions of (\ref{solutiondirac}) gives rise to a representation of  spacetime in terms of finite-rank self-adjoint operators.}
\end{center}
By \textit{ensemble of physical solutions} we mean a Hilbert spaces $(\scH,\langle \cdot|\cdot\rangle )$, where
$$
\scH\ \text{ a closed subspace of }\ \scH_m\quad\text{and}\quad \langle\cdot|\cdot\rangle:=(\cdot|\cdot)_m\!\!\restriction\!_{\scH\times \scH}.
$$
A \textit{representation} of spacetime points in terms of operators is then realized as follows.

\begin{theorem}\label{teoremaesistenzaF}
	Let $(\mathscr{H},\langle\cdot|\cdot\rangle)$ be an Hilbert subspace of $\sol$. Then, for every $x\in\bR^{1,3}$ there exists a unique operator
	$\rF^\varepsilon(x)\in\scF(\mathscr{H})$ such that
	\begin{equation}\label{defiF}
	\langle u|\rF^\varepsilon(x)v\rangle =-\Sl\gR_\varepsilon u(x)|\gR_\varepsilon v(x)\Sr \quad \mbox{for all }u,v\in\mathscr{H}.
	\end{equation}
	Moreover, the function $\rF^\varepsilon:\bR^{1,3}\ni x\mapsto \rF^\varepsilon(x)\in\scF(\mathscr{H})$ is continuous.
\end{theorem}
\begin{proof}
	The proof of  existence and  continuity can be found in \cite{FF}, respectively in Sections 1.2.2 and 1.2.3. The uniqueness follows directly from the arbitrariness of $u,v$ and the non-degeneracy of the scalar product.
\end{proof}

\begin{remark}
	The operator $\rF^\varepsilon(x)$ gives information on the densities and correlations of the physical solutions at the spacetime point $x\in\bR^{1,3}$. For this reason, the operator $\rF^\varepsilon(x)$ is called the \textbf{local correlation operator}, while the function $\rF^\varepsilon$ is referred to as the \textbf{local correlation function.}
\end{remark}

The function $\rF^\varepsilon$ provides a representation of spacetime in terms of operators. As will become clear later, each $\rF^\varepsilon(x)$ collects in its image a few distinguished physical solutions which are relevant at the spacetime point $x$.  We will see that, for specific closed subspaces $\mathscr{H}$, the corresponding local correlation functions are even homeomorphisms onto their image $\rF^\varepsilon(\bR^{1,3})$. \textit{In such a way,  spacetime is realized within $\scF$ not only as a set, but also as a differentiable manifold.}

In order to read everything in terms of causal fermion systems, it is necessary to single out suitable measures $\varrho$ on $\scF$ which realize $\rF^\varepsilon(\bR^{1,3})$ as their support.
The most natural choice for this measure is the push-forward of the Lebesgue-Borel\footnote{By Lebesgue-Borel measure we mean the restriction of the Lebesgue measure to $\bR^{1,3}$ to the Borel sets.} measure~$\mu$ on $\bR^{1,3}$ to $\scF$ through $\rF^\varepsilon$:
$$
(\rF^\varepsilon)_*\mu:  \mathfrak{Bor}(\scF(\mathscr{H}))\ni \Omega\mapsto \mu(({\rF^\varepsilon})^{-1}(\Omega))\in[0,\infty)
$$ 
Notice that the function $\rF^\varepsilon$ is continuous and therefore the push-forward is well-defined. 
\begin{definition}\label{defCFS}
	The pair $(\mathscr{H},(\rF^\varepsilon)_*\mu)$ is called the \textbf{(regularized) causal fermion system associated with} $\mathscr{H}$.
\end{definition}

As will be discussed at the end of this paper (see Proposition \ref{remarkunicitymeasure}), the choice of $(\rF^\varepsilon)_*\mu$ as the physical measure is a sensible choice in the case of $\mathscr{H}=\sol^-$ (which describes the Minkowski vacuum). In presence of particles or antiparticles, this choice might not be optimal, though. Nevertheless, these questions go beyond the scope of this paper and will not be investigated here.

\begin{remark} A few remarks follow.
	\vspace{0.1cm}
	\begin{itemize}[leftmargin=2.5em]
		\item[\rm{(i)}] From the continuity of $\rF^\varepsilon$ and the definition of $(\rF^\varepsilon)_*\mu$, it is easy to see that
		\begin{equation}\label{closureImF}
		\supp(\rF^\varepsilon)_*\mu =\overline{\rF^\varepsilon(\bR^{1,3})},
		\end{equation} 
		where the closure can be taken indifferently in $\scF(\mathscr{H})$ or $\scB(\mathscr{H})$, the former being a closed subset of the latter (see Theorem \ref{closednessF}). In the special cases we will study later on (e.g. the vacuum), the function $\rF^\varepsilon$ is injective and closed and, therefore, 
		$$
		\rF^\varepsilon(\bR^{1,3})=\overline{\rF^\varepsilon(\bR^{1,3})},
		$$ 
		showing that the support  of the measure faithfully realizes $\bR^{1,3}$.\\[-0.2cm]
		\item[\rm{(ii)}] The construction of $\rF^\varepsilon$ depends heavily on the chosen regularization. This is clearly an issue that needs to be clarified if the regularization is introduced as a mere tool which has to be removed afterwards. However, the situation is different if the regularization is believed to play a physical role itself, as can be in modeling the microscopic structure of  spacetime.\\[-0.2em]
		\item[\rm{(iii)}] The function $\rF^\varepsilon$ depends on the chosen closed subspace $\mathscr{H}\subset\sol$ and as such the realization of Minkowski space through $\rF^\varepsilon$ depends on the specific choice of solutions of the Dirac equation. Such solutions carry information about  spacetime itself and when specific and sufficiently large ensembles of them are gathered together, the said representation becomes indeed faithful, as we we will see in Section \ref{subsectininjective}.\\[-0.2cm]
		\item[\rm{(iv)}] Referring to Section \ref{sectiondefcfs}, everything can be understood in the larger set $\scF(\sol)$, and this will always be the case from now on. For simplicity of notation, the set $\scF(\sol)$ will be denoted simply by $\scF$.\\[-0.2cm]
		\item[\rm{(v)}] For the sake of compactness, we denote the (regularized) causal fermion system  simply by $(\mathscr{H},\rF^\varepsilon)$, for the measure is always understood to be $(\rF^\varepsilon)_*\mu$. Also, notice that the function $\rF^\varepsilon$ depends uniquely on the choice of the subspace $\mathscr{H}$, and therefore the causal fermion system depends uniquely on $\mathscr{H}$, too.
	\end{itemize}
	
\end{remark}


At this point, it is useful to analyze the relations between causal fermion systems which arise from different ensembles of physical solutions.
Consider two different Hilbert subspaces $\mathscr{H}_0\subset \mathscr{H}_1$ of $\sol$. As established with Theorem \ref{teoremaesistenzaF}, we can construct two causal fermion systems, by means of the local correlation functions:
\begin{equation}\label{restrictionF}
\mathrm{F}^\varepsilon_0:\bR^{1,3}\rightarrow \scF,\quad \mathrm{F}^\varepsilon_1:\bR^{1,3}\rightarrow \scF.
\end{equation}

These two functions are strictly connected to each other. To see this, consider the orthogonal projector $\Pi_0$ onto the subspace $\mathscr{H}_0$. Then, for every $u,v\in\mathscr{H}_0$ we have:
\begin{equation*}
\begin{split}
\langle u|\Pi_0\,\mathrm{F}^\varepsilon_1(x)\,\Pi_0v\rangle =\langle u|\mathrm{F}^\varepsilon_1(x)v\rangle=-\Sl\gR_\varepsilon u(x)|\gR_\varepsilon v(x)\Sr = \langle u|\rF^\varepsilon_0(x)v\rangle.
\end{split}
\end{equation*}
Since the restriction of $\langle\cdot|\cdot\rangle$ to any closed subspace defines again a Hermitian inner product and $u,v\in\mathscr{H}_0$ are arbitrary, we have just proven the following result.
%
\begin{proposition}\label{propprojcetion}
	Let $\mathscr{H}_0\subset \mathscr{H}_1$ be two Hilbert subspaces of $\sol$. Then the corresponding local correlation functions (\ref{restrictionF}) satisfy
	$$
	\mathrm{F}^\varepsilon_0(x)=\Pi_0\, \mathrm{F}^\varepsilon_1(x)\, \Pi_0
	$$
	for any $x\in\bR^{1,3}$, where $\Pi_0$ is the orthogonal projector on $\mathscr{H}_0$. 
\end{proposition}

To summarize, to every spacetime point $x\in\bR^{1,3}$ we have associated a bounded operator which encodes  information on the densities and correlations of the physical solutions generating $\mathscr{H}$ \textit{at the space point $x\in\bR^{1,3}$}.

A related and fundamental concept is given by the \textit{kernel of the fermionic operator}, which carries information on the correlations among wave functions \textit{at different spacetime points $x,y\in\bR^{1,3}$}. Of course, the two concepts have to agree when $x=y$.
The most natural way to do this is to project the action of the local correlation operator at one point onto the image of the other. 

\begin{definition}\label{projdef}
	Let $(\mathscr{H},\rF^\varepsilon)$ be a causal fermion system and $x,y\in\bR^{1,3}$. The associated \textbf{kernel of the fermionic operator} is defined as the operator
	$$
	\mathrm{P}^\varepsilon(x,y):=\pi_x \rF^\varepsilon(y):\scH_m\to \scH_m,\footnote{Everything is embedded within $\scH_m$. Note that $\mathrm{P}^\varepsilon(x,y)$ vanishes on the orthogonal of $\ran\rF^\varepsilon(x)$.}
	$$
	with $\pi_x$ the orthogonal projector on $\ran\rF^\varepsilon(x)$.\hspace{-0,05cm} In particular $\mathrm{P}^\varepsilon(x,x)=\rF^\varepsilon(x)$.
\end{definition} 

As the name suggests, this operators is strictly connected to the kernel of the fermionic projector of Definition \ref{kernelP} (more precisely, to the doubly-regularized kernel of Proposition \ref{rfp2}). This will be analyzed in detail in Section \ref{sectioninterpretation}, where the reasons behind the above definition of $\mathrm{P}^\varepsilon$ will become more evident.

\subsection{Correspondence to the Four-Dimensional Spinor Space}\label{correpondences}

In the previous sections we saw how spacetime can be realized in terms of operators starting from given ensembles of physical solutions of the Dirac equation. 
At this point, it is interesting to understand whether or not (and how, in case) the spinorial structure of the physical solutions can be retrieved by studying the operators $\rF^\varepsilon(x)$.  This is indeed possible to some extent, as we will show in this section.

Let us choose a closed subspace $\mathscr{H}$ of $\sol$, together with its local correlation function $\rF^\varepsilon$ (bear in mind that everything is embedded in the larger space $\sol$, also $\rF^\varepsilon$). 

Let us start our analysis with the spinor space $\bC^4$, equipped with the spin scalar product \eqref{spinscalarproduct}.
The closest structure to $\bC^4$ that we have at our disposal is the set $\ran\rF^\varepsilon(x)$, whose dimension is indeed not larger than four,
\begin{equation}
S_x:=\ran \mathrm{F}^\varepsilon(x)\subset\mathscr{H},\quad N_x:=\ker\rF^\varepsilon(x)\subset\sol.
\end{equation}
Exploiting the self-adjointness of $\mathrm{F}(x)$, we have the orthogonal decomposition
\begin{equation}\label{decomposS}
\sol=S_x\oplus N_x.
\end{equation}
Moreover, it holds that $\mathrm{F}^\varepsilon(x)( S_x)\subset  S_x$.

The set $S_x$ is a good candidate to represent $\bC^4$. However, in order to reach a full identification, we need to equip $S_x$ with an inner product of signature $(2,2)$  and find a canonical unitary mapping connecting the two spaces. The first claim is indeed always true, as shown in the next proposition. The second statement is true only under suitable assumptions.

\begin{proposition}\label{innerproductdef}
	Let $x\in\bR^{1,3}$. Then the linear space $ S_x$ can be equipped with a non-degenerate indefinite inner product of signature $(2,2)$ by defining:  
	\begin{equation}\label{spinscalarproductx}
	\Sl u| v\Sr_x:=-\langle u|\mathrm{F}^\varepsilon(x)v\rangle\quad \mbox{for all }u,v\in S_x
	\end{equation}
	The couple $(S_x,\Sl \cdot| \cdot\Sr_x)$ is called the \textbf{spin space at $x$}.
\end{proposition}
\begin{proof}
	It is clear that $\Sl \cdot| \cdot\Sr_x$ is a sesquilinear form. Moreover, 
	$$
	\overline{\Sl u|v\Sr_x}=-\overline{\langle u|\rF^\varepsilon(x)v\rangle}=-\langle\rF^\varepsilon(x)v|u\rangle=-\langle v|\rF^\varepsilon(x)u\rangle =\Sl v|u\Sr_x.
	$$ 
 Let us now check the non-degeneracy. Fix any $u\in S_x$ and suppose that $\Sl v| u\Sr_x=0$ for any $v\in S_x$. Consider any $\omega\in \mathcal{\mathscr{H}}$, then $\omega=s+t$, with $s\in S_x$ and $t\in N_x$, and thus
	$
	\langle\omega|\mathrm{F}(x)u\rangle=\langle s| \mathrm{F}(x)u\rangle=0.
	$
	Since $\omega$ is arbitrary, this implies $\rF^\varepsilon(x)u=0$. This calculation shows that $u\in \ran \mathrm{F}(x)\cap \ker \mathrm{F}(x)=\{0\}$. The signature of the inner product follows from the assumptions on  $\scF$.
\end{proof}

\begin{notation}
	The adjoint with respect to any $\Sl \cdot| \cdot\Sr_x$ will be denoted by~$^*$.
\end{notation}

The realization of a unitary mapping between $S_x$ and $\bC^4$ is unfortunately not always possible, as the dimension of $S_x$ is not always four. This may happen, for example, when the causal fermion systems is generated by an ensemble of physical solutions which is not large enough or, intuitively speaking, is not carrying enough information (more on this later on). 
On the other hand, when the dimension equals four, the causal fermion system is rich of interesting features. In the remainder of the paper we will focus mainly on such cases.

\begin{definition}\label{definitionregular}
	A point $x\in\bR^{1,3}$ is said to be \textbf{regular} for a causal fermion system (or the causal fermion system is regular at $x$) if the dimension is maximal, i.e. $\dim S_x=4$. The causal fermion system is said to be regular if it is regular at every point.
\end{definition}
\noindent Given this definition, we can now prove the following theorem: 
\begin{center}
	\textit{It is always possible to embed (canonically and) isometrically the space $S_x$\\ into $\bC^4$ (as non-degenerate indefinite inner-product spaces),\\ but unitarily only for regular points.}
\end{center}
\begin{theorem}\label{isometry}
	For any $x\in\bR^{1,3}$, referring to the corresponding indefinite inner products \eqref{spinscalarproduct} and \eqref{spinscalarproductx}, the function
	\begin{equation}\label{functionPhi}
	\Phi_x:S_x\ni u\mapsto \gR_\varepsilon u(x)\in\bC^4
	\end{equation}
	is a linear isometry, i.e.
	$$
	\Sl \Phi_x(u)|\Phi_x(v)\Sr=\Sl u|v\Sr_x\quad\mbox{for all }u,v\in S_x.
	$$ 
	The point is regular if and only if \eqref{functionPhi} is surjective.
\end{theorem}
\begin{proof}
	Let $u,v\in S_x$, then $$\Sl u|v\Sr_x =-(u|\mathrm{F}^\varepsilon(x)v)= \Sl\gR_\varepsilon u(x)|\,\gR_\varepsilon v(x)\Sr = \Sl\Phi_x u|\,\Phi_x v\Sr.$$ The non-degeneracy of $\Sl\cdot|\cdot\Sr_x$ guarantees that $\Phi_x$ is injective. The last statement follows from the rank-nullity theorem.
\end{proof}


\begin{lemma}\label{lemmanucleoregolare}
	Let $(\mathscr{H},\rF^\varepsilon)$ be a causal fermion system and $x\in\bR^{1,3}$. Then the following statements hold.
	\vspace{0.1cm}
	\begin{itemize}[leftmargin=2.5em]
		\item[\rm{(i)}] $\{u\in\mathscr{H}\:|\: \gR_\varepsilon u(x)=0 \}\subset N_x\cap\mathscr{H}$.\\[-0.7em]
		\item[\rm{(ii)}] If $x$ is regular, then the inclusion in point (i) can be replaced by an equality.\\[-0.7em]
		\item[\rm{(iii)}] If $x$ is regular, then $\gR_\varepsilon u(x)=\gR_\varepsilon (\pi_xu) (x)$ for every $u\in\mathscr{H}$.\\[-0.7em]
		\item[\rm{(iv)}] The point $x$ is regular if and only if there exist $\{ u_\mu\}_{\mu=0,1,2,3}\subset\mathscr{H}$ such that the vectors 
		$$
		\{\gR_\varepsilon u_\mu(x)\:|\: \mu=0,1,2,3\}
		$$
		are linearly independent.\\[-0.7em]
	\item[\rm{(v)}] The point $x$ is regular for the causal fermion system associated with any 
	  $\scH_1\supset\scH$.
	\end{itemize}
\end{lemma}
\begin{proof}
	Let us start with point (i). Consider any $u\in \mathscr{H}$ such that $\gR_\varepsilon u(x)=0$, then we have 
	$
	\langle v|\mathrm{F}(x)u\rangle=-\Sl\gR_\varepsilon v(x)|\gR_\varepsilon u(x) \Sr=0 
	$
	for any $v\in\mathscr{H}$.  Since the vector $v$ is arbitrary, this implies that $\mathrm{F}(x)u=0$, i.e. $u\in N_x$.  Let us now prove point (ii). Suppose the spacetime point is regular and consider $u\in N_x\cap\mathscr{H}$. Since the point is regular, we can always find vectors $v_\mu\in S_x\subset\mathscr{H}$ such that $\gR_\varepsilon v_\mu (x)=e_\mu$. Thus,
	 $$
	\Sl\gR_\varepsilon u(x)|e_\mu\Sr=\Sl\gR_\varepsilon u(x)|\gR_\varepsilon v_\mu(x)\Sr = -\langle u|\mathrm{F}(x)v_\mu\rangle=-\langle\rF(x)u|v_\mu\rangle=0,
	 $$ 
	 which gives $\gR_\varepsilon u(x)=0$. Point (iii) is obvious by linearity of $\gR_\varepsilon$ and points (i),(ii).  
	Let us now prove point (iv).	If the point is regular, then the thesis follows immediately, by simply taking $u_\mu:=\Phi_x^{-1}(e_\mu)$. On the contrary, suppose there are functions $u_\mu\in\mathscr{H}$ as in the assumption. Without loss of generality, we can always assume that $\gR_\varepsilon u_\mu(x)=\gamma^0 e_\mu$.
	Take any $v\in N_x\cap\mathscr{H}$, then 
	$$
	(e_\mu)^\dagger\,\gR_{\varepsilon} v(x)=\Sl\gR_{\varepsilon} u_\mu(x)|\gR_{\varepsilon} v(x)\Sr=-\langle u_\mu|\rF^\varepsilon(x)v\rangle=0,
	$$ 
	which implies $\gR_{\varepsilon} v(x)=0$. 
	At this point, notice that $u_\mu = \pi_x u_\mu + n_\mu$ for some $n_\mu\in N_x\cap\mathscr{H}$ (see (\ref{decomposS})) and therefore,
	$$
	\gamma^0 e_\mu=\gR_{\varepsilon} u_\mu(x) = \gR_{\varepsilon}  (\pi_x u_\mu) (x) +\gR_{\varepsilon} n_\mu (x)= \gR_{\varepsilon} (\pi_x u_\mu) (x).	
	$$ 
	 Since $\pi_x u_\mu\in S_x$, the thesis follows by  definition of the isometry $\Phi_x$. To conclude, notice that point (v) follows directly from point (iv).
\end{proof}

%

Given a causal fermion system $(\mathscr{H},\rF^\varepsilon)$, consider the closed subspace  generated by all the spin spaces $S_x$. 
It is not difficult to see that
\begin{equation}\label{SperpS}
S:=\overline{\mathrm{span} \bigcup_{x\in\bR^{1,3}}S_x}\subset\mathscr{H}\subset\sol,\quad S^\perp=\bigcap_{x\in\bR^{1,3}}N_x\subset\sol.
\end{equation}
If a vector $v\in\mathscr{H}$ belongs to the orthogonal of $S$, then $\rF^\varepsilon(x)v=0$ for every $x\in\bR^{1,3}$. In this case, the vector $v$ has no physical relevance in our settings, because its (regularized) local density and its correlation with other states are zero everywhere in spacetime. Therefore, the space $S$ should be interpreted as the (effective) physical Hilbert space of the system. 
\begin{definition}
	Given a causal fermion system $(\mathscr{H},\rF^\varepsilon)$, the corresponding subspace $S$ is called the \textbf{physical Hilbert space}.
\end{definition}
For regular systems, every physical solution in the  ensemble $\mathscr{H}$ is physically relevant and the spin spaces $S_x$ generate back the whole subspace $\scH$.

\begin{lemma}\label{lemmakerR}
	If the causal fermion system is regular, then $S=\mathscr{H}$.
\end{lemma}
\begin{proof}
	Let $v\in\mathscr{H}$ be orthogonal to every subspace $S_x$, then for any $x\in\bR^{1,3}$:
	$$
	\Sl\gR_\varepsilon  v(x)|\,\gR_\varepsilon  u(x)\Sr=-\langle v|\mathrm{F}^\varepsilon(x)u\rangle =0\quad\mbox{for all }u\in\mathscr{H}.
	$$
	Since the causal fermion system is regular, there always exist $u_\mu\in\mathscr{H}$ such that $\gR_\varepsilon  u_\mu(x)=e_\mu$ for any $\mu=0,1,2,3$. Replacing $u$ by $u_\mu$ in the identity above, we see that $\gR_\varepsilon  v(x)=0$. As  $x$  is arbitrary, we have $v\in\ker\gR_\varepsilon=\{0\} $, i.e. $v=0$.
\end{proof}

\begin{remark}
	It should be mentioned that in the general case where $\ker\gR_\varepsilon$ is not trivial, the physical space $S$ does not coincide with the original space $\scH$. Not even for regular systems.
\end{remark}

\subsection{Some Further Correspondences to Spinors for Regular Systems}

Let a causal fermion system $(\mathscr{H},\rF^\varepsilon)$ and a regular point $x\in\bR^{1,3}$ be given. In the previous section we saw that the function $\Phi_x$ defined in (\ref{functionPhi}) determines an isometry between $S_x$ and $\bC^4$, as indefinite inner product spaces. What can be said if we equip $\bC^4$ with the Euclidean positive definite inner product $(\lambda,\sigma)\mapsto \lambda^\dagger\sigma$? Remember that the two inner products are related by
$
\Sl \lambda|\sigma\Sr = \lambda^\dagger\, \gamma^0 \sigma,
$
where $\gamma^0$ is the zeroth Dirac matrix, which satisfies
$$
\gamma^0\gamma^0 = \bI_4,\quad (\gamma^0)^\dagger=\gamma^0,\quad(\gamma^0)^*=\gamma^0.
$$
We have at our disposal also other three Dirac matrices $\gamma^i$, which fulfill 
$$
\gamma^i\gamma^i=-\bI_4,\quad (\gamma^i)^\dagger=-\gamma^i,\quad \gamma^0\gamma^i=-\gamma^i\gamma^0,\quad (\gamma^i)^*=\gamma^i.
$$
The matrices $\gamma^\mu$ can be lifted to operators on $S_x$ by defining
$$
\Gamma^\mu(x):=\Phi_x^{-1}\, \gamma^\mu\,\Phi_x.
$$
It follows directly from the properties of the Dirac matrices that
\begin{equation}
\Gamma^0(x)\, \Gamma^0(x)=\bI_{S_x}=-\Gamma^i(x)\, \Gamma^i(x),\quad \Gamma^\mu(x)^*=\Gamma^\mu(x).
\end{equation}
This linear operator can always be understood as an operator on the whole space $\mathscr{H}$, simply by rewriting it as
$
\hat{\Gamma}^\mu(x):=\Gamma^\mu(x)\, \pi_x.
$
At this point, we can consider a new set of operators, namely
$$
\mathrm{F}^{\varepsilon,\mu}(x):=\mathrm{F}^\varepsilon(x)\hat{\Gamma}^\mu(x).
$$

\begin{proposition}\label{propertiesgamma}
	Let $x\in\bR^{1,3}$ be regular for $(\mathscr{H},\rF^\varepsilon)$. Then the maps $\hat{\Gamma}^\mu(x)$  satisfy
	\begin{itemize}[leftmargin=2.5em]
		\item[\rm{(i)}] $\pi_x=\hat{\Gamma}^0(x)^2=-\hat{\Gamma}^i(x)^2$\\[-0.7em]
		\item[\rm{(ii)}] $\mathrm{F}^{\varepsilon,\mu}(x)$ is self-adjoint \\[-0.7em]
		\item[\rm{(iii)}] $\mathrm{F}^{\varepsilon}(x)=\mathrm{F}^{\varepsilon,0}(x)\,\hat{\Gamma}^{0}(x)=-\mathrm{F}^{\varepsilon,i}(x)\,\hat{\Gamma}^{i}(x)$\\[-0.7em]
		\item[\rm{(iv)}] $\ran \mathrm{F}^{\varepsilon,\mu}(x)=S_x\ $ and $\ \ker\mathrm{F}^{\varepsilon,\mu}(x)= (S_x)^\perp$. 
	\end{itemize}
\end{proposition}
\begin{proof}
	Point (i) follows directly from the definition. To prove point (ii) choose (dropping the index $\varepsilon$ for simplicity)
	any $u,v\in\mathscr{H}$, then
	\begin{equation*}
	\begin{split}
	\langle u|\rF^\mu(x)v\rangle&=\langle u|\mathrm{F}(x)\hat{\Gamma}^\mu(x) v\rangle=-\Sl \pi_x u|\Gamma^\mu(x) \pi_x v\Sr_x=\\
	&=-\Sl \Gamma^\mu(x)\pi_x u| \pi_x v\Sr_x=\langle \hat{\Gamma}^\mu(x) u|\mathrm{F}(x)v\rangle=\\
	&=\langle u|\hat{\Gamma}^\mu(x)^\dagger\,\mathrm{F}(x)v\rangle=\langle u|\rF^\mu(x)^\dagger v\rangle.
	\end{split}
	\end{equation*}
	Now, note that $(\mathrm{F}(x)\,\hat{\Gamma}^\mu(x))\, \hat{\Gamma}^\mu(x)=\eta^{\mu\mu}\,\mathrm{F}(x)\pi_x=\eta^{\mu\mu}\,\mathrm{F}(x)$, from which point (iii) follows. Point (iv) follows by noticing that $\Gamma^\mu(x)$ and $\rF^\varepsilon(x)$ are bijections on $S_x$  and that $\mathrm{F}^\mu(x)$ is self-adjoint.
\end{proof}

We are ready to state and prove the following important representation result, which can be seen as the Euclidean counterpart of Theorem \ref{teoremaesistenzaF}, Proposition  \ref{innerproductdef} and Theorem \ref{isometry}.
\begin{proposition}
	Let $x\in\bR^{1,3}$ be regular. Then, for every   $u,v\in \mathscr{H}$,
	\begin{equation}\label{positiveprod}
	\langle u|\mathrm{F}^{\varepsilon,\mu}(x) v\rangle =-\Sl\gR_\varepsilon u(x)|\, \gamma^\mu\,\gR_\varepsilon v(x)\Sr.
	\end{equation}
	In particular, $S_x$ can be endowed with a  positive-definite inner product
	\begin{equation}\label{positiveinnerproduct}
	S_x\times S_x\ni (u,v)\mapsto -\langle u|\rF^{\varepsilon,0}(x)v\rangle\in\bC\quad \mbox{for all } u,v\in S_x,
	\end{equation}
	which makes  $\Phi_x$ a unitary mapping onto the spinor space $\bC^4$, if equipped with the Euclidean inner product.
\end{proposition}
\begin{proof}
	Let us start with the first identity. Tke any $u,v\in\mathscr{H}$, then
	\begin{equation*}
	\begin{split}
	\Sl\gR_\varepsilon u(x)|\, \gamma^\mu\,\gR_\varepsilon v(x)\Sr&=\Sl \Phi_x^{-1}(\gR_\varepsilon u(x))|\Phi_x^{-1}(\gamma^\mu\,\gR_\varepsilon v(x))\Sr_x=\Sl \pi_x u|\Gamma^\mu(x) \pi_x v\Sr_x=\\
	&=-\langle\pi_x u|\mathrm{F}^{\varepsilon}(x)\Gamma^\mu(x)\pi_x v\rangle=-\langle u|\mathrm{F}^{\varepsilon,\mu}(x) v\rangle
	\end{split}
	\end{equation*}
	Now, if we manage to show that \eqref{positiveinnerproduct} is a positive definite inner product on $S_x$, then the last statement will follow directly from (\ref{positiveprod}). The sesquilinearity and Hermiticity come from the analogous properties of $\langle\cdot|\cdot\rangle$ and the self-adjointness of $\mathrm{F}^{\varepsilon,0}(x)$. The positivity follows from (\ref{positiveprod}) and the positivity of the Euclidean inner product.
\end{proof}
\begin{remark}
	Notice that the positive inner product \eqref{positiveinnerproduct} can be defined on $S_x$ only for regular points, while the indefinite inner product defined in Proposition \ref{innerproductdef} always exists. Nevertheless, if we focus on regular causal fermion systems, this results shows that all the algebraic properties of the spin spaces are encoded in the local correlation function.
\end{remark}

\subsection{Some Physical Interpretations}\label{sectioninterpretation}

Consider a finite-dimensional subspace $\sU\subset\sol$ and let $\{u_i\}_{i=1,\dots,n},\{v_i\}_{i=1,\dots,n}$ be any two orthonormal bases of it. The corresponding wedge products (remember that the particles are fermions) are connected~by
$$
u_1\wedge \dots \wedge u_n = (\det\sM)\ v_1\wedge \dots \wedge v_n,
$$
where $\sM$ is the unitary matrix transforming one basis into the other, in particular $\det\sM\in \mathbb{U}(1)$. Therefore,  being equal up to a phase, the multi-particle states are physically equivalent. This shows that there exists a one-to-one identification between  $n$-particle (pure wedge products) states and $n$-dimensional subspaces of $\sol$. 

Bearing this in mind, 
one can imagine to describe an infinitely large ensemble of fermions by the infinite dimensional subspace of $\scH_m$ generated by the corresponding orthonormal physical solutions.\footnote{A rigorous construction of wedge products of infinitely many orthonormal states is developed in \cite[Section 2]{DDMS}. I am grateful to an anonymous referee for drawing my attention to this work.}
Following the original idea of Dirac, one can then interpret the whole subset $\sol^-$ as the ``multi-particle state" formed by all  negative-energy particle states and interpret it as the vacuum. The addition of positive-energy states and the removal of negative-energy states correspond instead to the presence of particles or antiparticles. 
This interpretation is revived in the context of causal fermion systems. 

The subspace $\sol^-$ is the smallest (non-trivial) ensemble of physical solutions whose causal fermion system  carries a unitary representation of the translation group, see Proposition  \ref{propositioninvariancetransl} and Proposition \ref{spectrumF}. 
This is no longer true when states are added or deleted from the system, i.e. when the ensemble is modified in a non trivial way.
With this in mind, the causal fermion system arising from $\sol^-$ is a sensible choice for describing a vacuum \footnote{The same argument holds for $\scH_m^+$, but the two constructions are unitarily equivalent.}.

We now introduce the following notation.
Given any finite-dimensional subspaces $\sU_+\subset\sol^+$  and $\sU_-\subset\sol^- $, we define
\begin{equation}\label{aggiuntotoglio}
a_+(\sU_+):=\sol^- \oplus\sU_+\subset \sol,\quad a_-(\sU_-):=\sol^- \,\cap\sU_-^\perp\subset\sol.
\end{equation}
Similarly, given a combined system with some states in the negative spectrum and some states in the positive spectrum, we write:
\begin{equation}\label{entramne}
a(\sU_-,\sU_+):= \left(\sol^-  \cap \sU_-^\perp\right) \oplus\sU_+\subset\sol.
\end{equation}

\begin{remark}\label{remarkentang}
	It should be stressed that the theory here presented does not encompass the description of entangled states, in that the identifications discussed above hold only for pure products of single-particle states.
\end{remark}
The following result follows directly from Proposition \ref{propertiesgamma}, bearing in mind that the definition of trace is independent from the Hilbert basis and that $\rF^\varepsilon(x)$ vanishes on the orthogonal of $S_x$. We anticipate (it will be proved in Section \ref{subsectionregularityvacuum}) that the causal fermion system arising from $\sol^-$ is indeed regular at any point $x\in\bR^{1,3}$.
\begin{proposition}\label{density}
	Let $(\scH,\rF^\varepsilon)$ be a causal fermion system, $x\in\bR^{1,3}$ a regular point and  $\sN_x,\sN$ Hilbert bases of $S_x$ and $\mathscr{H}$, respectively. Then
	\begin{equation}\label{infinitedistrib}
	\sum_{u\in\sN}\Sl\gR_\varepsilon u(x)|\, \gamma^\mu\,\gR_\varepsilon u(x)\Sr= 
	-\mathsf{tr}\, \mathrm{F}^{\varepsilon,\mu}(x)=\sum_{w\in\sN_x}\Sl\gR_\varepsilon w(x)|\, \gamma^\mu\,\gR_\varepsilon w(x)\Sr.
	\end{equation}
	In particular, the left-hand series converges and does not depend on the Hilbert~basis. 
\end{proposition} 

The result above has a direct physical interpretation. Sticking to the standard interpretation of quantum mechanics, given a normalized smooth physical solution $\psi$, the value 
$$
J_\psi(t,\V{x}):=\Sl\psi(t,\V{x})\,|\,\gamma^\mu\, \psi(t,\V{x})\Sr
$$
defines the \textbf{current density} at time $t\in\bR^3$. A regularization consists in the replacement
$$
J_\psi(t,\V{x})\stackrel{\gR_{\varepsilon}}{\longmapsto} J_{\gR_{\varepsilon}\psi}(t,\V{x}):=\Sl\gR_{\varepsilon}\psi (t,\V{x})\,|\,\gamma^\mu\, \gR_{\varepsilon}\psi (t,\V{x})\Sr,
$$
which is to be interpreted as the (regularized) current density at time $t$. If we are given a system with more then one particle, the total current density consists in the sum of the individual densities. In particular, the series on the left-hand side of (\ref{infinitedistrib}) can be understood as \textit{the current density of the system where all the states (a Hilbert basis to more precise) in $\mathscr{H}$ are occupied.} Notice that this is finite also for infinite dimensional subspaces like $\sol^- $. Nevertheless, in the limit $\varepsilon\searrow 0$, this sum generally diverges. In order to get a finite quantity, one proceeds as follows.

\begin{theorem}
	Let $\sU_\pm\subset \sol^\pm$ be finite-dimensional subspaces with Hilbert bases $\sN_\pm$. Then
	\begin{equation}
	\mathsf{tr}(\mathrm{F}^{\varepsilon,\mu}_{vac}(x)-\mathrm{F}^{\varepsilon,\mu}(x))=\sum_{z\in\sN_+}J^\mu_{\gR_\varepsilon z}(x) -\sum_{w\in\sN_-}J^\mu_{\gR_\varepsilon w}(x).
	\end{equation}
	where  $\mathrm{F}$ is the causal fermion system relative to $a(\sU_-,\sU_+)$ and $\rF^\varepsilon_{vac}$ to $\sol^-$.
\end{theorem}
\begin{proof}
	Let us start with the pure negative-energy case first. Let $\sM_-$ be any Hilbert basis of $a_-(\sU_-)$, then $\sM_+:=\sM_-\cup \sN_+$ is a Hilbert basis of $a(\sU_+,\sU_-)$ and $\sM_0:=\sM_-\cup \sN_-$ is a Hilbert basis of $\sol^- $. Applying Proposition \ref{density} we get 
	\begin{equation*}
	\begin{split}
	-\mathsf{tr}\,\mathrm{F}^\mu(x)+\mathsf{tr}\,\mathrm{F}^\mu_{vac}(x) =&\sum_{u\in\sM_-}J^\mu_{\gR_\varepsilon u}(x)+\sum_{z_+\in\sN_+}J^\mu_{\gR_\varepsilon z+}(x)-\\
	&-\sum_{u\in\sM_-}J^\mu_{\gR_\varepsilon u}(x)-\sum_{z_-\in\sN_-}J^\mu_{\gR_\varepsilon z_-}(x),
	\end{split}
	\end{equation*}
	and the thesis follows.
\end{proof}
\begin{remark}
	After spacetime $\bR^{1,3}$ has been realized in terms of the local correlation operators $\rF^\varepsilon(x)$, the particle current density is  addressed as a property of spacetime alone, no longer of the physical solutions.
	
\end{remark}

To conclude this section, we compute the explicit form of the kernel of the fermionic operator when lifted to $\bC^4$ through the action of the isometries $\Phi_x$.

Consider an arbitrary Hilbert basis $\{u_n\}_{n\in\bN}$ of $\mathscr{H}$, then the completeness relation gives
\begin{equation}\label{expansionP}
\mathrm{P}^\varepsilon(x,y)=\pi_{x}\rF^\varepsilon(y)=\sum_{n\in\bN} \pi_x u_n \langle u_n|\rF^\varepsilon(y)\,\cdot\,\rangle=-\sum_{n\in\bN} \pi_x u_n \Sl \pi_y u_n|\pi_y\,\cdot\,\Sr_y
\end{equation}
From this equation we see that the correlation between the points $x,y$ is the result of adding up all the individual contributions coming from the basis elements at the given points. 
More precisely, exploiting the definition of the isometry (\ref{functionPhi}) we have the following result.
\begin{proposition}\label{representationgenerickernel}
	Let $\mathscr{H}\subset\sol$ and let $\{u_n\}_{n\in\bN}$ be any Hilbert basis of $\scH$. Then
	\begin{equation}\label{expansionP2}
	\Phi_x\,\mathrm{P}^\varepsilon(x,y)\,\Phi_y^{-1}\restr_{\Phi_z(S_y)}=-\sum_{n=0}^\infty\gR_{\varepsilon} u_n (x)\,\Sl\gR_{\varepsilon}u_n(y)|\cdot\Sr\,\restr_{\Phi_z(S_y)},
	\end{equation}
	for any spacetime points $x,y\in\bR^{1,3}$.
\end{proposition}
This result clarifies the statement above: all the physical solutions generating $\mathscr{H}$ contribute to the correlation between $x,y\in\bR^{1,3}$ described by $\mathrm{P}^\varepsilon(x,y)$. We will go back to this representation in Section \ref{sectionvacuumprojector}.

\begin{remark}
	Equation (\ref{expansionP2})  shows that the choice for $\mathrm{P}^\varepsilon$ made in Definition \ref{projdef} was indeed sensible, as it does realize a correlation function between two points, once represented in $\bC^4$. 
\end{remark}

\section{On the Regularity of Causal Fermion Systems}

In section \ref{correpondences} we introduced the concept of a regular causal fermion system (see Definition \ref{definitionregular}). It is interesting at this point to study in which situations a given causal fermion system is indeed regular or, more precisely, how regularity is related to the choice of the wave functions forming the ensemble.
As one could expect, the vacuum realizes a regular causal fermion system. Moreover, this feature is not lost when additional positive-energy physical solutions are added to the system. In line with the concept of a Dirac sea, one may interpret this process as an enrichment of the system and, therefore, it is reasonable that regularity is preserved. On the other hand, the removal of negative-energy solutions can break regularity, as we will see in a straightforward example. Only in very special cases regularity is preserved, namely when the physical solutions do not vary too much on the microscopic scale $\varepsilon$. This shows an \textit{asymmetry between particles and antiparticles}, at least within the limits of validity of this interpretation.

\subsection{The Regularity of the Vacuum}\label{subsectionregularityvacuum}

The aim of this section is to show that the causal fermion system associated with $\sol^-$ (that is, the vacuum) is regular. 

Let us choose any arbitrary $x_0=(t_0,\V{x}_0)\in\bR^{1,3}$. Then, as stated in Lemma \ref{lemmanucleoregolare}, we need to find four elements  $u_\mu\in\sol^-$ such that $\gR_\varepsilon u_\mu(x_0)=e_\mu$ for every $\mu=0,1,2,3$.
\begin{itemize}[leftmargin=2.5em]
	\item[\rm{(a)}] Let us start by considering the smooth function $G^{(a)}\in\mathcal{S}_p(\bR^3,\bC)$ defined by:
	\begin{equation}
	G^{(a)}(\V{k}):=A\,e^{-\frac{\V{k}^2}{4\sigma^2}}e^{ -i(\omega(\V{k})t_0+\V{k}\cdot\V{x}_0)},\quad A\in\mathbb{R}^0_+.
	\end{equation}
	The $\mathcal{L}^2$ and $\mathcal{L}^1$ norms can be easily computed:
	\begin{equation}\label{equalityA}
	\|G^{(a)}\|_{\mathcal{L}^2}=A(\sqrt{2\pi}\sigma)^{3/2},\,\quad \|G^{(a)}\|_{\mathcal{L}^1}=A(2\sqrt{\pi}\sigma)^3=:\beta(\sigma).
	\end{equation}
	Let us now replace $\lambda_{\uparrow,\downarrow}:=(2\pi)^{3/2}G^{(a)}$ into (\ref{basicsolution}). Then, taking into account (\ref{fundamentalspinors}) and the fact that the Gaussian and the energy are even functions,  the corresponding physical solutions $u_{\uparrow}^{(a)}$ and $u_{\downarrow}^{(a)}$  fulfill
	\begin{equation}\label{valorein0}
	u_\uparrow^{(a)}(x_0)=\beta(\sigma) e_2,\quad u_\downarrow^{(a)}(x_0)=\beta(\sigma) e_3.
	\end{equation}
	
	\item[\rm{(b)}] Now, consider the smooth function $G^{(b)}\in\mathcal{S}_p(\bR^3,\bC)$ defined by
	$$
	G^{(b)}(\V{k}):=B e^{-\frac{\V{k}^2}{4\sigma^2}}e^{ -i(\omega(\V{k})t_0+\V{k}\cdot\V{x}_0)}(\omega(\V{k})+m)k_3,\quad B>0.
	$$
	Let us discuss $\|G^{(b)}\|_{\mathcal{L}^2}$ first. Exploiting the inequalities 
	$$
	(k_3)^2\le |\V{k}|^2\quad\mbox{and}\quad m\le \omega(\V{k})+m\le 2\omega(\V{k}),
	$$ 
	it can be shown by direct computation that
	\begin{equation}\label{inequalityB}
	mB(2\pi)^{3/4}\sigma^{5/2}\le \|G^{(b)}\|_{\mathcal{L}^2}\le 2(2\pi)^{3/4} B \sqrt{15\sigma^7+3m^2 \sigma^5}.
	\end{equation}
	Now, as done in point (a), we replace $\lambda_{\uparrow,\downarrow}:= (2\pi)^{3/2} G^{(b)}$ into (\ref{basicsolution}) and evaluate the corresponding physical solutions  $u_{\uparrow}^{(b)}$ and $u_{\downarrow}^{(b)}$   at $x=x_0$. Exploiting the form of the fundamental solutions (\ref{fundamentalspinors}), it is not difficult to see that the first, second and third components of $u^{(b)}_{\uparrow}$ as well as the zeroth, second and third components of $u^{(b)}_\downarrow$ vanish, for the integrands are odd with respect to $k_3$. We are left with the integral $$\gamma(\sigma):=\int_{\bR^3}  e^{-\frac{\V{k}^2}{4\sigma^2}}(k_3)^2\,d^3\V{k}=2^4 \pi^{3/2} B\sigma^5.$$ More precisely,
	\begin{equation}\label{valorein02}
	u^{(b)}_\uparrow(x_0)=-\gamma(\sigma) e_{0},\quad u^{(b)}_\downarrow(x_0)=+\gamma(\sigma) e_1.
	\end{equation}
\end{itemize}

Concluding, we see that the vectors $
u_{\downarrow,\uparrow}^{(a),(b)}(x_0)$ are orthogonal to each other and, therefore, linearly independent.
%
We now show that this feature is preserved if we regularize the corresponding solutions, that is we prove that the vectors $\gR_\varepsilon u_{\downarrow,\uparrow}(x_0)$ are linearly independent.

In order to do this, we need  to carry out some estimates on the derivatives of the physical solutions. Exploiting Lemma \ref{estimate}-(ii) and  inequalities 
$$
|k_\mu|\le |\V{k}|+m\quad\mbox{and}\quad(\omega(\V{k})+m)(|\V{k}|+m)|k_3|\le (|\V{k}|+2m)^2|\V{k}|,
$$ 
it can be shown by direct computation that, for any $x\in\bR^{1,3}$,
\vspace{0.1cm}
\begin{align}
	\mbox{(a)}\qquad |\partial_\mu u^{(a)}_{\uparrow,\downarrow}(x)|&\le \sqrt{2}\pi A(32\sigma^4+8m\sqrt{\pi}\sigma^3),\label{estimatea} \\[0.3em]
	\mbox{(b)}\qquad	|\partial_\mu u^{(b)}_{\uparrow,\downarrow}(x)|&\le 2^6\sqrt{2} \pi B(4\sigma^6+2m^2\sigma^4+3m\sqrt{\pi} \sigma^5) \label{estimateb}.
\end{align}
\vspace{0.1cm}
Exploiting these estimates, we can finally prove the regularity of the vacuum.

\begin{proposition}\label{regularlinearlyindep1}
	Let $u^{(a)}, u^{(b)}$ be the special solutions defined above with $\sigma=m$. Then
	$$
	\gR_\varepsilon u_{\uparrow}^{(a)}(x_0),\quad \gR_\varepsilon u_{\downarrow}^{(a)}(x_0),\quad \gR_\varepsilon u_{\uparrow}^{(b)}(x_0),\quad 
	\gR_\varepsilon u_{\downarrow}^{(b)}(x_0)
	$$
	are linearly independent vectors of $\bC^4$.
\end{proposition}
In order to prove this, we need the following elementary result whose proof can be found in the appendix.
\begin{lemma}\label{lemmalinearlyinde}
	Let $\mathscr{H}$ be a Hilbert space, $\{e_1,\dots,e_n \}$ an orthonormal set and $0<\epsilon<\frac{1}{n}$. Then any set $\{v_1,\dots,v_n \}\subset\mathscr{H}$ which fulfills $\|e_i-v_i\|<\epsilon$ for every $i=1,\dots,n$ is linearly independent.
\end{lemma}
\begin{proof}[Proof of Proposition \ref{regularlinearlyindep1}]
	Take $A:=(2\sqrt{\pi}\sigma)^{-3}$, then $\beta(\sigma)=1$ and $u_{\uparrow,\downarrow}^{(a)}(x_0)=e_{2,3}$. Take $B=2^{-4}\pi^{-3/2}\sigma^{-5}$, then $\gamma(\sigma)=1$ and $u_{\uparrow,\downarrow}^{(b)}(x_0)=\pm e_{0,1}$. This choice of $A$ and $B$, together with estimates (\ref{estimatea}) and (\ref{estimateb}) above, gives (see Equation (\ref{jacobian}))
	\begin{equation}\label{estimatesJu}
	\begin{split}
	\|\mathfrak{J} u_{\uparrow,\downarrow}^{(a)}\|_{x_0,\infty}&\le\! \sup_{B(x_0,\varepsilon)}\!  \left(2^3\cdot \sum_{\mu=0}^3|\partial_\mu u_{\uparrow,\downarrow}^{(a)}(x)|\right) \le 2^2\sqrt{\frac{2}{\pi}}(32\sigma + 8m\sqrt{\pi})\le 10^3\, m,\\
	\|\mathfrak{J} u_{\uparrow,\downarrow}^{(b)} \|_{x_0,\infty}&\le\! \sup_{B(x_0,\varepsilon)}\!\left( 2^3\cdot\sum_{\mu=0}^3|\partial_\mu u_{\uparrow,\downarrow}^{(b)}(x)| \right) \le  2^7\sqrt{\frac{2}{\pi}}(4\sigma+2m^2\sigma^{-1}+3\sqrt{\pi}m)\le 10^4\, m,
	\end{split}
	\end{equation} 
	where the last inequalities follows by setting $\sigma=m$.
	At this point, exploiting Assumption \ref{assumptionepsilon} we get
	$$
	10^3\, m\varepsilon< 10^{4}\,m\varepsilon  <\frac{1}{4}.
	$$
	This, together with Lemma \ref{lemmalinearlyinde}, Lemma \ref{approximationpropertiesreg}-(v) and identities (\ref{valorein0}), (\ref{valorein02}) ends the proof.
\end{proof}

Since the point $x_0\in\bR^{1,3}$ was chosen arbitrarily, the following theorem follows.

\begin{theorem}
	The causal fermion system associated with $\sol^- $ (the vacuum) is regular.
\end{theorem}

To conclude this section, we show that the local correlation operators of the vacuum do not only have full-rank, but are also unitary equivalent. This is a manifestation of translation invariance. 
\begin{proposition}\label{propositioninvariancetransl}
	For any $a\in\bR^{1,3}$ the operator defined on the solutions space $\sol$ by
	$$
	\mathrm{U}_a:\sol \ni u\mapsto u_a:=u(\,\cdot+a)\in\sol
	$$
	is well-defined and unitary. Moreover, \\[-0.9em]
	\begin{itemize}
		\item[(i)]	$\mathrm{U}_a(\sol^-)\subset\sol^-$ and $\mathrm{U}_a(\sol^+)\subset\sol^+$\\[-0.8em]
		\item[(ii)] $\gR_\varepsilon\,\mathrm{U}_a=\mathrm{U}_a\,\gR_\varepsilon$.
	\end{itemize}
\end{proposition}
\begin{proof}
	The proof of (i) can be carried out working on the dense subspace $\Sol$ and then taking the unique continuous extension. Concerning point (ii), pick any $u\in\sol^- $ and $a\in\bR^{1,3}$, then by definition we have:
	\begin{equation}\label{commutation}
		\begin{split}
			(\gR_\varepsilon u)(x+a)&=\int_{\bR^4} h_\varepsilon(x+a-z)u(z)\,d^4z\stackrel{w=z-a}{=}\int_{\bR^4}h_\varepsilon(x-w)u(w+a)\,d^4w=\\
			&=\int_{\bR^4} h_\varepsilon(x-w)u_a(w)\, d^4w=\gR_\varepsilon u_a (x).
		\end{split}
	\end{equation}
The claim follows by definition of $\mathrm{U}$.
\end{proof}
\begin{proposition}\label{spectrumF}
	Let $\rF_{vac}^\varepsilon$ be the local correlation function associated with $\sol^-$. Then, for every $x\mbox{ and }a\in\bR^{1,3}$,
	$$
	\rF_{vac}^\varepsilon(x+a)=(\mathrm{U}_a)^\dagger\,\rF_{vac}^\varepsilon(x)\,\mathrm{U}_a.
	$$
	In particular, 
	$$
	\|\rF_{vac}^\varepsilon(x)\|=\|\rF_{vac}^\varepsilon(y)\|\quad\text{and}\quad
	\sigma(\mathrm{F}_{vac}^\varepsilon(x))=\sigma(\mathrm{F}^\varepsilon_{vac}(y))=:\sigma_{vac}.
	$$
\end{proposition}
\begin{proof}
	First, notice that $\mathrm{U}_a(\sol^-)\subset\sol^-$, as proved in Proposition \ref{propositioninvariancetransl}-(i). 
	Exploiting \eqref{commutation}, it follows that
	\begin{equation*}
	\begin{split}
	\langle u|\rF_{vac}^\varepsilon(x+a)v\rangle &=-\Sl\gR_\varepsilon u(x+a)|\gR_\varepsilon v(x+a)\Sr= -\Sl\gR_\varepsilon u_a (x)|\gR_\varepsilon v_a(x)\Sr=\\
	&=\langle u_a|\rF_{vac}^\varepsilon(x)v_a\rangle=\langle u|(\mathrm{U}_a)^\dagger\,\rF_{vac}^\varepsilon(x)\,\mathrm{U}_a v\rangle,
	\end{split}
	\end{equation*}
	from which $\rF_{vac}^\varepsilon(x+a)=(\mathrm{U}_a)^\dagger\,\rF_{vac}^\varepsilon(x)\,\mathrm{U}_a$.
\end{proof}

\subsection{The Presence of Holes in the Dirac Sea}\label{sectionholes}

Before studying this case in detail, we show that, in the general case, the removal of negative-energy solutions may cause a loss of regularity.
Consider the causal fermion system generated by the vacuum $(\sol^- ,\rF_{vac}^\varepsilon)$ and choose any spacetime point $x\in \bR^{1,3}$. We already know that the linear subspace $S_x$ has dimension four, because of regularity. At this point, we can think of removing exactly the physical solutions corresponding to the finite dimensional subspace $S_x$ itself:
$$
\sol^- \longmapsto a_-(S_x):=\sol^-\cap S_x^\perp,
$$
and construct its causal fermion system $(a_-(S_x),\mathrm{F}_0^\varepsilon)$. From Proposition \ref{propprojcetion} we know that
$$
\mathrm{F}^\varepsilon_0(x)=\Pi_0\, \mathrm{F}^\varepsilon_{vac}(x)\, \Pi_0,
$$
where $\Pi_0$ is the orthogonal projector onto $a_-(S_x)$. Since the operator $\mathrm{F}^\varepsilon_{vac}(x)$ is self-adjoint, it satisfies
$$
\ker \mathrm{F}^\varepsilon_{vac}(x)=\big(\ran \mathrm{F}_{vac}(x)\big)^\perp = S_x^\perp\supset a_-(S_x),
$$
which implies $\mathrm{F}^\varepsilon_0(x)=0$. Indeed, let $u\in\sol$ be any arbitrary vector, then $\Pi_0 u\in a_-(S_x)\subset \ker \mathrm{F}^\varepsilon_{vac}(x)$ and therefore,
$$
\mathrm{F}^\varepsilon_0(x)u=\Pi_0\, \mathrm{F}^\varepsilon_{vac}(x)\, \Pi_0u=0.
$$
In particular, the causal fermion system $(a_-(\sU_0),\mathrm{F}^\varepsilon_0)$ is \textit{not regular} at $x$.
Removing the physical solutions which are relevant at the given spacetime point generates a critical defect in the local correlation function.

\begin{remark}
	Note that this argument does not apply to the positive-energy case. Indeed, consider the causal fermion system constructed out of $a_+(\sU)$ for some $\sU\subset\sol^+ $ and choose any $x\in\bR^{1,3}$. Then, the spin-space $S_x$ \textit{is not} a subset of $\sol^+ $ in general, but of the whole space $\sol^- \oplus\sU$, because its vectors have a component also in the negative-energy spectrum. Therefore, their removal from $a_+(\sU)$ does not correspond to the annihilation of positive-energy solutions.
\end{remark}

This simple example shows that the removal of even a finite number of particles from the Dirac sea can perturb the system in a way that regularity is loss.  The physical solutions belonging to $S_x$ are nevertheless quite special: they are very peaked around the light-cone centered at $x$ and diverge on it in the limit $\varepsilon\searrow 0$. One may then wonder whether the removal of solutions which vary very slowly on the scale $\varepsilon$ everywhere in spacetime preserves regularity. The discussion in the following section shows that this is indeed the case and provides sufficient conditions under which the creation of holes in the Dirac sea does not affect regularity.

Mathematically speaking, we are interested in the causal fermion system constructed out of a finite dimensional closed subspace $\sU\subset\sol^-$,
$$
a_-(\sU)=\sol^- \cap\sU^\perp.
$$
The strategy is the same as above: we need to find four elements of $a_-(\sU)$ whose values at $x_0$ are linearly independent. This is more complicated than the vacuum case, since the constraint $u\in \sU^\perp$ is tricky to control while carrying out an argument like the one above. On the other hand, given their convenience, we do not want to give up Gaussian-like solutions, although they generally do not belong to $a_-(\sU)$. In what follows, we show how all this can be put together. The proof of the following three results can be found in the appendix.
\vspace{0.2cm}

\begin{itemize}[leftmargin=2.5em]	
	\item[\rm{(1)}] \textit{In this first part we exploit the density of the Schwartz functions to construct a set of orthonormal smooth vectors which span a subspace which is ``as close as possible" to $\sU$.}
	\begin{proposition}\label{part1}
		Let $\{u_1,\dots,u_n \}$ be an orthonormal basis of $\sU$ and $0<\epsilon<n^{-1/2}$. Then there exists an orthonormal set $\{\psi_1,\dots,\psi_n\}\subset \hat{\mathrm{E}}(\hat{P}_-(\mathcal{S}_p(\bR^3,\bC^4))$ such that:
		\vspace{0.2cm}
		\begin{itemize}[leftmargin=2.5em]
			\item[\rm{(i)}] $\|u_i-\psi_i\|<\epsilon$ and $|(u_i|\psi_j)-\delta_{ij}|\le \epsilon$ for every $i,j=1,\dots,n$,\\[-0.5em]
			\item [\rm{(ii)}] $\mathrm{span}\{\psi_1,\dots,\psi_n\}\cap \sU^\perp=\{0\}$.
		\end{itemize}
	\end{proposition}

	\begin{remark} A few remarks follow.
		\vspace{0.2cm}
		\begin{itemize}[leftmargin=2.5em]
			\item[\rm{(i)}]  Point (a) ensures that  $$\|u_1\wedge\dots\wedge u_n-\psi_1\wedge\dots\wedge \psi_n\|<n^2\epsilon.$$ Therefore, the multi-particle states are close to each other.\\[-0.7em]
			\item[\rm{(ii)}] 	For the sake of simplicity, let us denote the vector of functions $\psi_i$ by $\bpsi$. In particular, we can define
			$$
			\bpsi(x):=(\psi_1(x),\dots,\psi_n(x)).
			$$
			With the symbol $|\bpsi(x)|$ we denote the $\bC^{4n}$-norm of the vector $\bpsi(x)$. More precisely,
			\begin{equation}
			|\bpsi(x)|^2=\sum_{i=1}^n|\psi_i(x)|^2=\sum_{i=1}^n\varrho_{\psi_i}(x)=:\varrho_{\bpsi}(x),
			\end{equation}
			where $\varrho_{\psi_i}$ is the \textit{probability density} of  $\psi_i$. The quantity $\varrho_{\bpsi}$ is to be understood as the \textit{particle-number density} associated with the multi-particle state~$\bpsi$.
		\end{itemize}
	\end{remark}
\vspace{0.3cm}
	\item[\rm{(2)}] \textit{In this second part we construct smooth solutions which are orthogonal to $\sU$. More precisely, we define a specific way of projecting the space $\hat{\mathrm{E}}(\hat{P}_-(\mathcal{S}_p(\bR^3,\bC^4))$ onto the space $\hat{\mathrm{E}}(\hat{P}_-(\mathcal{S}_p(\bR^3,\bC^4))\cap a_-(\sU)$.}
	\begin{proposition}\label{part2}
		Let $\{\psi_1,\dots,\psi_n \}$ be as in Proposition \ref{part1}. Then, for every solution $\varphi\in\hat{\mathrm{E}}(\hat{P}_-(\mathcal{S}_p(\bR^3,\bC^4))$ there exist unique scalars $\lambda_i(\varphi)\in\bC$, with $i=1,\dots,n$, such that
		\begin{equation}\label{projectionoverU}
		\Psi[\varphi]:=\varphi-\sum_{i=1}^n\lambda_i(\varphi)\psi_i\in a_-(\sU).
		\end{equation}
		In particular, $\lambda_i(\varphi)=0$ for every $i=1,\dots,n$ if and only if $\varphi\perp\sU$.
	\end{proposition}
\vspace{0.3cm}

	\item[\rm{(3)}] \textit{In this third part we estimate the norm of the vector $\lambda(\varphi)\in\bC^n$ in (\ref{projectionoverU}). 
	}
	\begin{proposition}\label{part3}
		There exists $\{\psi_1,\dots,\psi_n \}$ as in Proposition \ref{part1} such that the scalars $\lambda_i(\varphi)$  in (\ref{projectionoverU}) satisfy
		\begin{equation}\label{stimaambda}
		|\lambda(\varphi)|\le 2\|\varphi\|
		\end{equation}
		for every function $\varphi\in\hat{\mathrm{E}}(P_-(\mathcal{S}_p(\bR^3,\bC^4))$.
	\end{proposition}

\end{itemize}

Given a finite-dimensional subspace $\sU$ of $\sol^-$, the states $\psi_i$  give good smooth approximations of the states within $\sU$. We then give a new definition.
\begin{definition}
	Let $\sU$ be a finite-dimensional subspace of $\sol^-$. Then a family of physical solutions $\{\psi_1,\dots,\psi\}\subset\hat{\mathrm{E}}(\hat{P}_-(\mathcal{S}_p(\bR^3,\bC^4)))$ as in Proposition \ref{part3} is called an \textbf{approximating set of smooth states} for $\sU$.
\end{definition}

We are ready to discuss the regularity of the causal fermion system in presence of holes. To this aim, choose an approximating set of smooth states for $\sU$. The properties of the operator $\Psi$ and the boundedness of the scalars $\lambda_i$ allow us to repeat the discussion carried out for the vacuum in the previous section. The idea is to consider the functions $\Psi[u]$, with $u$ the solutions defined in (a) and (b) of Section \ref{subsectionregularityvacuum}, apply the regularization operator and find some suitable $\sigma$ such that their values at $x_0$ are linearly independent.

\vspace{0.3cm}

For simplicity of notation, let $u_\alpha$ with $\alpha=0,1,2,3$ denote the solutions $u_{\uparrow,\downarrow}^{(a),(b)}$ defined in Section \ref{subsectionregularityvacuum} (which depend on a parameter $\sigma$), $C_\alpha(\sigma)$ the associated coefficient $\beta(\sigma)$ or $\pm\gamma(\sigma)$ and $\chi_\alpha$ the corresponding correction term $-\sum_{i=1}^n\lambda_{i}(u_\alpha)\psi_i$, as given in Proposition \ref{part2}.
We need to check whether the estimate
$$
A_\alpha:=\left|\frac{\gR\Psi[u_\alpha](x_0)}{C_\alpha(\sigma)}-e_\alpha\right|\stackrel{!}{<}\frac{1}{4}
$$
is fulfilled for any $\alpha$. If this is true, then Lemma \ref{lemmalinearlyinde} would conclude the argument.

\begin{equation}\label{estimateregulfinale}
\begin{split}
A_\alpha:\!&=\!\left|\frac{\gR_\varepsilon u_\alpha(x_0)+\gR\chi_\alpha(x_0)}{C_\alpha(\sigma)}\!-\!\frac{u_\alpha(x_0)+\chi_\alpha(x_0)}{C_\alpha(\sigma)}\!+\!\frac{u_\alpha(x_0)+\chi_\alpha(x_0)}{C_\alpha(\sigma)}\!-\! \frac{u_\alpha(x_0)}{C_\alpha(\sigma)}\right|=\\
&=\frac{1}{|C_\alpha(\sigma)|}\left(\left|\gR_\varepsilon u_\alpha(x_0)-u_\alpha(x_0)\right|+|\gR\chi_\alpha(x_0)-\chi_\alpha(x_0)|\right)+\frac{|\chi_\alpha(x_0)|}{|C_\alpha(\sigma)|}\le\\
&\le \frac{1}{|C_\alpha(\sigma)|}\,\varepsilon\left(\|\mathfrak{J} u_\alpha\|_{x_0,\infty}\!+\!\sum_{i=1}^n|\lambda_i(u_\alpha)|\|\mathfrak{J} \psi_i\|_{x_0,\infty}\right)\!+\!\frac{1}{|C_\alpha(\sigma)|}|\lambda(u_\alpha)||\bpsi(x_0)|\le\\
&\le \frac{\|\mathfrak{J}u_\alpha\|_{x_0,\infty}}{|C_\alpha(\sigma)|}\,\varepsilon\!+\!\frac{2^3\pi^{3/2}\|G_\alpha\|_{\mathcal{L}^2}}{|C_\alpha(\sigma)|}\left(\varepsilon\sum_{i=1}^n\|\mathfrak{J} \psi_i\|_{x_0,\infty} + \!|\bpsi(x_0)|\right)\stackrel{!}{<}\frac{1}{4},
\end{split}
\end{equation}
where we used $|\lambda_i(u_\alpha)|\le |\lambda(u_\alpha)|$, Proposition \ref{approximationpropertiesreg}-(v), Lemma \ref{estimate} and (\ref{stimaambda}).

The quantity between brackets will play an important role in the following discussion and deserves its own notation.
\vspace{0.2cm}

\begin{definition}
	Let $\{\psi_1,\dots,\psi_n\}\subset\sol\cap \mathcal{C}^\infty(\mathbb{R}^3,\mathbb{C}^4)$ be an orthonormal set and $x\in\bR^{1,3}$. Then the quantity
	$$
	\cE(\bpsi,\varepsilon,x):=|\bpsi(x)|+\varepsilon\sum_{i=1}^n\|\mathfrak{J} \psi_i\|_{x,\infty}.
	$$
	is called the \textbf{microscopic-behavior distribution} of the multi-particle state.
	
\end{definition}
\begin{remark}
	This quantity provides information on the local behavior of the particle-number distribution and about its variation on the scale $\varepsilon$. 
\end{remark}
At this point, reasoning as in the derivation of the inequalities (\ref{estimatesJu}), setting $\sigma=\lambda m$ (with $\lambda$ arbitrary dimensionless parameter) and making use of (\ref{equalityA}) and (\ref{inequalityB}), it can be shown that
\begin{equation*}
\begin{split}
A_{\uparrow,\downarrow}^{(a)}
&\le 2^7(1+\lambda)m\varepsilon+\dfrac{(2\pi)^{3/4}}{\lambda^{3/2}m^{3/2}}\ \cE(\bpsi,\varepsilon,x_0),\\[0.2cm]
A_{\uparrow,\downarrow}^{(b)}&
\le 2^{11}\left(\lambda+\frac{1}{\lambda}\!+\!1\right)m\varepsilon+\dfrac{4(2\pi)^{3/4}}{m^{3/2}}\left(\frac{1}{\lambda^{3/2}}+\dfrac{1}{\lambda^{5/2}}\right)\cE(\bpsi,\varepsilon,x_0).
\end{split}
\end{equation*}
If we choose $\lambda=(\varepsilon_0/\varepsilon)^{1/2}$, for some arbitrary length parameter $\varepsilon_0$, the right-hand sides above become:
\begin{equation}\label{inequalitiesA}
\begin{split}
A_{\downarrow,\uparrow}^{(a)}&\le 2^7m\left(\varepsilon+\sqrt{\varepsilon_0\varepsilon}\right), +\frac{(2\pi)^{3/4}}{m^{3/2}}\left(\frac{\varepsilon}{\varepsilon_0}\right)^{\frac{3}{4}}\cE(\bpsi,\varepsilon,x_0).\\
A_{\downarrow,\uparrow}^{(b)}&\le 2^{11}\left(\sqrt{\varepsilon_0\varepsilon}+\frac{\varepsilon^{3/2}}{\sqrt{\varepsilon_0}}+\varepsilon\right)+\frac{4(2\pi)^{3/4}}{m^{3/2}}\left(\frac{\varepsilon}{\varepsilon_0}\right)^{\frac{3}{4}}\left(1+\sqrt{\frac{\varepsilon}{\varepsilon_0}}\right)\cE(\bpsi,\varepsilon,x_0)
\end{split}
\end{equation}
We see that the desired condition $A<1/4$ can be obtained by taking $\varepsilon$ small enough.
Nevertheless, this argument is not legitimate, because we want the length parameter $\varepsilon$ to freely vary within the interval $(0,\varepsilon_{max})$. Only at the very end a specific choice for $\varepsilon$ should be made, which is based upon physical arguments and cannot depend on the specific choice of the physical solutions. Troubles could arise from the addend involving the microscopic-behaviour function, which can be very large. The remaining terms involve only positive powers of the microscopic scale $\varepsilon$ and are generally very small. Therefore, we need to make some assumptions on the quantity $\cE(\bpsi,\varepsilon,x_0)$.

Let us start with the limiting case $\cE(\bpsi,\varepsilon,x_0)=0$. If we choose $\varepsilon_0=\varepsilon$, then the above inequalities can be simplified to
$$
A_{\downarrow,\uparrow}^{(a)}\le 2^{8}\,m\varepsilon<10^3\,m\varepsilon,\quad A_{\downarrow,\uparrow}^{(b)}\le 3\cdot 2^{11}\,m\varepsilon<3\cdot 10^4\,m\varepsilon.
$$ 	
If we now take $m\varepsilon< 10^{-15}$ as in Assumption $\ref{assumptionepsilon}$, we get indeed $A<1/4$.

However, the condition $\cE(\bpsi,\varepsilon,x_0)=0$ is too strong in general. Indeed, as discussed briefly in Section \ref{decayproperteis}, it holds that
$$
\supp u(t,\cdot)=\bR^3\quad \mbox{for all }u\in \hat{\mathrm{E}}(\hat{P}_-(\mathcal{S}_p(\bR^3,\bC^4)) \setminus\{0\}\ \mbox{ and for every~$t\in\bR$.}
$$
In particular, the strict inequality $|\bpsi(t,\cdot)|>0$ holds on a dense open subset of $\bR^{3}$.

Small - but non vanishing - distributions $\cE(\bpsi,\varepsilon,x_0)$ arise under the heuristic assumptions
$$
\varepsilon\sum_{i=1}^n\|\mathfrak{J}\psi_i\|_{x_0,\infty}\ll 1\quad\mbox{and}\quad |\bpsi(x_0)|\ll 1.
$$ 
\begin{center}
	\textit{Physically speaking, we require a low particle-number density around the point $x_0$ and a slow variation of the physical solutions on the microscopic scale $\varepsilon$}.
\end{center}
More concretely, take for example $\varepsilon_0=10^{16}\varepsilon$. Then inequalities (\ref{inequalitiesA}) simplify to
\begin{equation*}
\begin{split}
A_{\downarrow,\uparrow}^{(a)}&<  10^{11}\,m\varepsilon+\frac{(2\pi)^{3/4}}{m^{3/2}}10^{-12}\,\cE(\bpsi,\varepsilon,x_0),\\[0.1cm]
A_{\downarrow,\uparrow}^{(b)}&<3\cdot 10^{12}\, m\varepsilon+\frac{8(2\pi)^{3/4}}{m^{3/2}}10^{-12}\,\cE(\bpsi,\varepsilon,x_0).
\end{split}
\end{equation*}
Exploiting again Assumption \ref{assumptionepsilon} and choosing $\sigma=m(\varepsilon_0/\varepsilon)^{1/2}$, we then have the following result.
\begin{proposition}\label{proporegularityholes}
	Let $u^{(a)}, u^{(b)}$ be the special solutions defined in  Section \ref{subsectionregularityvacuum} with $\sigma=10^{8}m$ and suppose that $\sU$ admits an approximating smooth set of states $\{\psi_1,\dots,\psi_n\}$ that fulfills
	\begin{equation}\label{macroscopicond}
	\cE(\bpsi,\varepsilon,x_0)<10^{9}\cdot m^{3/2}.
	\end{equation}
	Then the following vectors are linearly independent:
	$$
	\gR_\varepsilon \Psi\big[u_{\uparrow}^{(a)}\big](x_0),\quad \gR_\varepsilon \Psi\big[u_{\downarrow}^{(a)}\big](x_0),\quad \gR_\varepsilon \Psi\big[u_{\uparrow}^{(b)}\big](x_0),\quad 
	\gR_\varepsilon \Psi\big[u_{\downarrow}^{(b)}\big](x_0).
	$$
\end{proposition}

\begin{remark}
	Condition (\ref{macroscopicond})  can be understood as a ``macroscopic" behavior of the physical solutions $\psi_i$ at the spacetime point $x_0$ with respect to the microscopic scale~$\varepsilon$. 
\end{remark}
This suggest the following definition.
\begin{definition}
	A finite dimensional subspace $\sU\subset\sol^- $ is said to be \textbf{$\boldsymbol{\varepsilon}$-macroscopic} at $x\in\bR^{1,3}$ if it admits an approximating smooth set of states which fulfills the macroscopic condition (\ref{macroscopicond}).
\end{definition}

Summarizing, we have proven the following result.
\begin{theorem}\label{emacroscopicregular}
	Let $\sU\subset\sol^-$ be a finite-dimensional subspace which is $\varepsilon$-macroscopic at $x\in\bR^{1,3}$. Then the associated causal fermion system $(a_-(\sU),\mathrm{F}^\varepsilon)$ is regular at $x$. 
\end{theorem}

\subsection{The General Case of Particles and Antiparticles}

The addition of positive-energy particles to the vacuum enriches the system and, as such, it should not affect regularity. This is indeed the case, as we now prove in full generality.

Suppose we are given a finite-dimensional subspace $\sU_-\subset\scH_-$ such that $a_-(\sU_-)$ is regular, such as an $\varepsilon$-macroscopic system or the vacuum itself. If we now add some positive-energy physical solutions (described by a finite-dimensional subspace $\sU_+\subset\sol^+ $) the relevant Hilbert space becomes
$$
a(\sU_+,\sU_-)=a_-(\sU_-) \oplus\sU_+.
$$ 
Since this Hilbert space contains $a_-(\sU_-)$, Lemma \ref{lemmanucleoregolare}-(v) implies the following.

\begin{corollary}\label{regularitygeneralcase}
	Let $\sU_\pm\subset\sol^\pm$ be finite-dimensional subspaces and suppose that $a_-(\sU_-)$ is regular at $x\in\bR^{1,3}$, then $a(\sU_-,\sU_+)$ is also regular at $x$.
\end{corollary}

In particular, the above result applies to the special case $\sU_-=\{0\}$, i.e. when only particles (and not antiparticle) are added to the system. Corollary \ref{regularitygeneralcase} shows that no assumptions are needed to ensure regularity, differently from the case when antiparticles are present. This shows an asymmetry between the concepts of particle and antiparticle, at least in the present settings.


\subsection{Characterization of the Vacuum Local Correlation Operators}\label{sectionvacuumprojector}

In this section we focus on the vacuum, i.e. to the causal fermion system generated by $\sol^-$, and conclude the discussion started at the end of Section \ref{sectioninterpretation}. In the previous section we showed that the vacuum  defines a regular system and, therefore, every spin space $S_x$ is isometric to $\bC^4$ by means of the function $\Phi_x$ (see Theorem \ref{isometry}).  We now show that, under these identifications, the kernel of the fermionic operator and the doubly-regularized kernel of the fermionic projector do in fact coincide.
\begin{theorem}\label{diagram}
	For any couple $x,y\in\bR^{1,3}$ the following diagram commutes\footnote{\textit{Abuse of language}:  $\mathrm{P}^\varepsilon(x,y)$ actually acts on the whole space $\scH_m$ but it vanishes on~$(S_y)^\perp$.}:
	\begin{displaymath}
	\begin{tikzcd}
	S_y  \arrow[rr, "{\mathrm{P}^\varepsilon(x,y)}"]\arrow[d, "\Phi_y"'] & &S_x\arrow[d, "\Phi_x"]\\	
	\bC^4   \arrow[rr, "{2\pi P^{2\varepsilon}(x,y)}"] 	& & \bC^4
	\end{tikzcd}
	\end{displaymath}
	where $P^{2\varepsilon}(x,y)$ is introduced in Definition \ref{rfp2}.
\end{theorem}
\begin{proof}
	The proof comes from Proposition \ref{representationP} and Equation (\ref{expansionP2}).
\end{proof}

We want to make use of this identifications and study the spectral properties of the operators $\rF^\varepsilon(x)$. From Proposition \ref{spectrumF} we already know that the local correlation operators $\rF^\varepsilon(x)$ are unitarily equivalent and therefore have the same spectrum $\sigma_{vac}$. 
Thanks to the definition of $\mathrm{P}^\varepsilon(x,y)$ and to Theorem \ref{diagram}, it is clear that the spectrum of $\rF^\varepsilon(x)|_{S_x}$ coincides with the spectrum of $2\pi P^{2\varepsilon}(x,x)$. Using this, together with Proposition \ref{propositionkerneldiagonal}, it is possible to prove the following proposition. 
\begin{proposition}
	Referring to Proposition \ref{propositionkerneldiagonal}, for every $x\in\bR^{1,3}$,
	$$
	\sigma_{vac}=\sigma(2\pi\,P^{2\varepsilon}(x,x))\cup\{0\}=\left\{0,2\pi\nu^+(\varepsilon),2\pi\nu^-(\varepsilon) \right\}.
	$$
\end{proposition}
We  now make use of the representation in Theorem \ref{diagram} to investigate the properties of the physical solutions which form the spin spaces $S_x$.
In particular, we focus on the eigenvector bases of the local correlation operators $\rF^\varepsilon(x)$.

\begin{proposition}\label{elemetnsS}
	For any spacetime point $x\in\bR^{1,3}$ the following statements hold.
	\vspace{0.1cm}
	\begin{itemize}[leftmargin=2.5em]
		\item[\rm{(i)}] A vector $u\in\sol^-$ belongs to  $S_x$ if and only if there exists $a\in\bC^4$ such that
		\begin{equation}\label{equazionedaR}
		u=P^{\varepsilon}(\,\cdot,x)a.
		\end{equation}
		More precisely, for every $u\in S_x$ there is a unique $a\in\bC^4$ such that (\ref{equazionedaR}) holds.\\[-0.15cm]
		\item[\rm{(ii)}] The action of the local correlation operator on $u\in\sol^- $ is given by
		$$
		\rF_{vac}^\varepsilon(x)u = 2\pi\,P^{\varepsilon}(\,\cdot\,,x)\gR_\varepsilon u(x).
		$$
		\item[\rm{(iii)}] In particular, a linear basis for $S_x$ is given by the four linearly independent functions
		\begin{equation}\label{eigenvectors}
		e_{x,\mu}:=P^{\varepsilon}(\,\cdot\,,x)e_\mu.
		\end{equation}
		Moreover, for every $a\in\bR^{1,3}$,
		$$
		\mathrm{U}_a (e_{x,\mu})=e_{x+a,\mu}.
		$$
		\item[\rm{(iv)}] The vectors $e_{x,\mu}$ are eigenvectors of $\rF^\varepsilon(x)$ with eigenvalues $2\pi\nu^\pm(\varepsilon)$. More precisely,
		\begin{equation}
		\begin{cases}
		\rF_{vac}^\varepsilon(x)\,e_{x,0}=(2\pi\,\nu^-(\varepsilon))\, e_{x, 0}\\[0.3em]
		\rF_{vac}^\varepsilon(x)\,e_{x,1}=(2\pi\,\nu^-(\varepsilon)) \,e_{x, 1}
		\end{cases}
		\quad
		\begin{cases}
		\rF_{vac}^\varepsilon(x)e_{x,2}=(2\pi\,\nu^+(\varepsilon))\, e_{x, 2}\\[0.3em]
		\rF_{vac}^\varepsilon(x)e_{x,3}=(2\pi\,\nu^+(\varepsilon))\, e_{x, 3}
		\end{cases}
		\end{equation}
		\vspace{0.1cm}
	\end{itemize}
\end{proposition}
\begin{proof}
	Point (i). First, notice that $P^{\varepsilon}(\,\cdot\,,x)a\in\sol^-$ for any $a\in\bC^4$, as proved in Proposition \ref{rfp}.	Now, by definition a vector $u\in\sol^-$ belongs to $S_x$ if it can be written as $u=\mathrm{F}(x)w$ for some element $w$ in $\sol^- $. 
	Therefore, at every other spacetime point $y\in\bR^{1,3}$ we have
	\begin{equation*}\label{equationeutile}
	\begin{split}
	\gR_\varepsilon u(y)&=\Phi_y(\pi_y\mathrm{F}(x)w)=\Phi_y(P(y,x)w)=2\pi\,P^{2\varepsilon}(y,x)\, \gR_\varepsilon w(x)=\\
	&=\gR_\varepsilon (2\pi\,P^{\varepsilon}(\cdot,x)\gR_\varepsilon w(x))(y),
	\end{split}
	\end{equation*}
	where we used Proposition \ref{double regularazione}.
	Exploiting Lemma \ref{lemmanucleoregolare} and Lemma \ref{lemmakerR}, we get
	$$
	u-2\pi\,P^{\varepsilon}(\,\cdot,x)\gR_\varepsilon  w(x)\in\sol^-\cap\bigcap_{y\in\bR^{1,3}}\sN(y)=\sol^-\cap S^\perp=\{0\},
	$$
	which implies $u=P^{\varepsilon}(\,\cdot,x)(2\pi\,\gR_\varepsilon  w(x))$.
	On the other hand, take $a\in\bC^4$
	and consider the solution $P^{\varepsilon}(\,\cdot,x)a\in\sol^-$. Then, by regularity there always exists some $w\in\sol^- $ such that $\gR_\varepsilon w(x)=a$. Therefore, we have
	\begin{equation*}
	\begin{split}
	\gR_\varepsilon  (2\pi\,P^{\varepsilon}(\,\cdot,x)a)(y)&=2\pi\,P_{2\varepsilon}(y,x)\,\gR_\varepsilon w(x)=\Phi_y(\mathrm{P}^\varepsilon(y,x)w)=\\
	&=\Phi_y (\pi_y\mathrm{F}^\varepsilon(x)w)=\gR_\varepsilon (\mathrm{F}^\varepsilon(x)w)(y).
	\end{split}
	\end{equation*}
	Reasoning as above, the arbitrariness of $y$ ensures that
	$
	2\pi\,P_{\varepsilon}(\,\cdot,x)a=\mathrm{F}^\varepsilon(x)w\in S_x.
	$
	To conclude point (i), notice that the uniqueness of $a$ in (\ref{equazionedaR}) follows  by first applying $\gR_\varepsilon$ to both sides of (\ref{equazionedaR}) and then making use of Proposition \ref{double regularazione}.
	
	Let us now prove point (ii). Take $u\in\scH_m^-$. Thanks to point (i) there exists a unique  $a\in\bC^4$ such that $\rF^\varepsilon(x)u=P^{\varepsilon}(\,\cdot\,,x)a$. Exploiting the isometry $\Phi_x$ and the identity $\rF^\varepsilon(x)=\mathrm{P}^\varepsilon(x,x)$, we get
	$$
	2\pi\,P^{2\varepsilon}(x,x)\gR_\varepsilon u(x)=\Phi_x(\rF_{vac}^\varepsilon(x)u)=\Phi_x(P_{\varepsilon}(\,\cdot\,,x)a)=P^{2\varepsilon}(x,x)a.
	$$
	At this point, arguing as in the proof of point (ii) of Lemma \ref{double regularazione}, we see that $a=2\pi\,\gR_\varepsilon u(x)$.  
	The proof of point (iii) can be checked easily, by applying the regularization operator to \eqref{eigenvectors}, exploiting Lemma \ref{double regularazione} and noticing that  $S_x$ has dimension four.
	To conclude, let us prove point (iv).  Exploiting points (ii) and Proposition \ref{propositionkerneldiagonal}, we have
	\begin{equation}\label{1}
	\begin{split}
	\rF_{vac}^\varepsilon(x)e_{x,\mu}=2\pi\,P^{\varepsilon}(\,\cdot\,,x)P^{2\varepsilon}(x,x)e_\mu=2\pi\,\nu_\mu\, P^{\varepsilon}(\,\cdot\,,x)e_\mu= (2\pi\nu_\mu)\, e_{x,\mu}.
	\end{split}
	\end{equation}
	The proof is complete.
\end{proof}

\begin{remark}
	At every spacetime point $x\in\bR^{1,3}$, the spin space $S_x$ of the vacuum is formed by the ``maximally localized" physical solutions $P^{\varepsilon}(\,\cdot\,,x)a$. 
\end{remark}

To conclude this section, we briefly discuss how the eigenvalues $\nu^\pm(\varepsilon)$ change when  particles or antiparticles are added to the system. Consider two finite-dimensional subspaces $\sU_+\subset \sol^+$ and $\sU_-\subset\sol^-$. The local correlation operators of $a_\pm(\sU_\pm)$ are self-adjoint on the corresponding spin spaces and as such they can be diagonalized with real eigenvalues. For simplicity, we consider the case of smooth physical solutions.
First, notice that, in the case of regular systems, the kernel of the fermionic operator can be represented in $\bC^4$ as the  matrix (see identity (\ref{expansionP2})):
\begin{equation}\label{realizationpa}
\begin{split}
\Phi_x \mathrm{P}^\varepsilon(x,y)\Phi_x^{-1}=2\pi\,P^{2\varepsilon}(x,y)&+\sum_{i=1}^{n_-} \gR_\varepsilon e_i^-(x)\Sl\gR_\varepsilon e_i^-(y)|\cdot\Sr-\\
&-\sum_{i=1}^{n_+} \gR_\varepsilon e_i^+(x)\Sl\gR_\varepsilon e_i^+(y)|\cdot\Sr
\end{split}
\end{equation}
where $\{e_i^\pm\}_{i=1,\dots,n_\pm}$ are arbitary Hilbert bases of $\sU_\pm$. 
\begin{theorem}\label{approximationeigenvalues}
	Let $\sU_\pm\subset \hat{\mathrm{E}}(\hat{P}_\pm(\mathcal{S}_p(\bR^3,\bC^4)))$ be finite-dimensional and assume that the causal fermion system $a(\sU_-,\sU_+)$ is regular at $x\in\bR^{1,3}$. Then, the corresponding kernel $\mathrm{P}^{\varepsilon}(x,x)$ has four non-vanishing real eigenvalues $\{\nu_i\}_{i=1,2,3,4}$  (counting multiplicities) which satisfy
	$$
	\min\left\{\left|\nu_{i}-2\pi\,\nu^+(\varepsilon)\right|,\left|\nu_i-2\pi\,\nu^-(\varepsilon)\right|\right\}
	\le 2\,\cE(\boldsymbol{e},\varepsilon,x)^2,\quad i=1,2,3,4
	$$
	where $\boldsymbol{e}=\boldsymbol{e}_+\cup \boldsymbol{e}_-$, with $\boldsymbol{e}_\pm$ any Hilbert basis of $\sU_\pm$.
\end{theorem}
\begin{proof}
	The thesis follows from identity \eqref{realizationpa}, Proposition \ref{propositionperturb} and the following:
	\begin{equation*}
	\begin{split}
	& 2\pi\|\Delta P(x,x)\|_2\le \sum_{i=1}^{n}\|\gR_\varepsilon e_i(x)\Sl\gR_\varepsilon e_i(x)|\cdot\Sr \|_2\le\\
	&\le \sum_{i=1}^{n}\|\gR_\varepsilon e_i(x)\gR_\varepsilon e_i(x)^\dagger \|_2=\sum_{i=1}^{n}|\gR_\varepsilon e_i(x)|^2\le\\
	&\le \sum_{i=1}^n||e_i(x)|+\varepsilon\|\gJ e_i\|_{x,\infty}|^2\le 2\sum_{i=1}^n\big||e_i(x)|^2+\varepsilon^2\big(\|\gJ e_i\|_{x,\infty}\big)^2\big|=\\
	&=2\left(|\boldsymbol{e}(x)|^2+\varepsilon^2\sum_{i=1}^n\big(\|\gJ e_i\|_{x,\infty}\big)^2\right)\le 2\,\cE(\V{e},\varepsilon,x)^2,
	\end{split}
	\end{equation*}
	where we used $\|AB\|_2\le \|A\|_2\|B\|_2$ and $\|\gamma^0\|_2=1$.
	
\end{proof}

\begin{remark}
	The above result shows that, in presence of particles or antiparticles, the variation of the eigenvalues of the kernel of fermionic projector is controlled by the local behavior of the wave functions around the point.
\end{remark}

\section{On the Smooth Manifold Structure of Regular Systems}\label{sectionsmooth}

In this final section we prove how and under which assumptions regular causal fermion systems realize differentiable manifolds within $\scF$. More precisely, the goal is to show in which situations the function $\rF^\varepsilon:\bR^{1,3}\rightarrow \scF$ is closed and homeomorphic onto its image. When this happens, $\rF^\varepsilon(\bR^{1,3})=\supp\varrho$ and $\rF^\varepsilon$ itself realizes a global chart which turns $\supp\rho$ into a smooth manifold. The proof is divided into two steps.

\subsection{Step 1: Injectivity of the Local Correlation Function}\label{subsectininjective}

The local correlation function 
$$
\mathrm{F}^\varepsilon:\bR^{1,3}\ni x\mapsto\rF^\varepsilon(x)\in \scF
$$ 
provides a representation of  spacetime in terms of operators. The question is: are we still able to distinguish two spacetime points by looking at the corresponding local correlation operators? Do we loose information in representing spacetime points as operators? In other words,

\begin{center}
	\textit{Is the local correlation function faithful?}
\end{center}
The general answer is negative. As we already mentioned before, choosing too poor ensembles of physical solutions might cause a critical loss of information and some features, like regularity, can be lost. 
The same goes for injectivity and here we show a simple example of this. 

Consider the vacuum causal fermion system $(\sol^- ,\mathrm{F}^\varepsilon_{vac})$ and take any couple of different spacetime points $x\neq y$. 
If we now remove the negative-energy solutions corresponding to the subspace $\sU_0:=S_x+ S_y$ (that is, if we construct the causal fermion system of the subspace $a_-(\sU_0)$), the corresponding function $\mathrm{F}^\varepsilon_0$ satisfies (see Lemma \ref{propprojcetion})
\begin{equation}\label{lossinjectivity}
\mathrm{F}^\varepsilon_0(x)=\Pi_0\, \mathrm{F}_{vac}^\varepsilon(x)\,\Pi_0 =0 = \Pi_0\, \mathrm{F}_{vac}^\varepsilon(y)\Pi_0 = \mathrm{F}^\varepsilon_0(y).
\end{equation}
This identity can be proved with an argument similar to the proof of loss of regularity. Identity \eqref{lossinjectivity} shows that  injectivity is lost. 

Nevertheless, for special subspaces of $\sol$ injectivity turns out to be preserved. The first example is of course the vacuum. Then we show that the addition of positive-energy solutions does not affect injectivity (as one could expect). On the other hand, in creating holes we have to choose the solutions very carefully in order to preserve injectivity. This is the content of the current section. 

\subsubsection{The Vacuum}
Let us start with the vacuum.
In order to prove injectivity we again make use of special solutions: Gaussian wave-packets. Let us introduce
$$
G_{\V{p}}^{(\sigma)}(\V{k}):=\frac{1}{(\sqrt{2\pi}\sigma)^{3}}e^{-\frac{(\V{k}-\V{p})^2}{2\sigma^2}}.
$$
It is well-known that these functions converge as distributions to the Dirac delta in the limit $\sigma\searrow 0$:
$$
\lim_{\sigma\to 0}G_{\V{p}}^{(\sigma)}=\delta^{(3)}_{\V{p}}.
$$
Now, consider the negative-energy solution $u_{\V{p}}^{(\sigma)}\in\sol^- $ defined as in (\ref{basicsolution}) by
\begin{equation}
u_\V{p}^{(\sigma)}(x):=\int_{\bR^3}\frac{d^3\V{k}}{(2\pi)^{3/2}}\,G_{\V{p}}^{(\sigma)}(\V{k})\chi_\uparrow^-(\V{k})e^{i(\omega(\V{k})t+\V{k}\cdot\V{x})}.
\end{equation}
Applying the regularization operator to this function, we obtain
\begin{equation*}
\begin{split}
\gR_\varepsilon u_{\V{p}}^{(\sigma)}(x)&=\int_{\bR^3} \frac{d^3\V{k}}{(2\pi)^{3/2}}\, \mathfrak{g}_\varepsilon(\V{k})\, G_{\V{p}}^{(\sigma)}(\V{k})\chi_\uparrow^-(\V{k}) e^{i (\omega(\V{k})t+\V{k}\cdot\V{x})}.
\end{split}
\end{equation*}
Exploiting the properties of $G_\V{p}^{(\sigma)}$, we finally have
\begin{equation}\label{limit}
\lim_{\sigma\to 0}\gR_\varepsilon u_{\V{p}}^{(\sigma)}(x)=(2\pi)^{-{3/2}}\mathfrak{g}_\varepsilon(\V{p})\chi_\uparrow^-(\V{p})e^{i(\omega(\V{p})t+\V{p}\cdot\V{x})}.
\end{equation}
%
%
%
After these preliminaries, we are ready to prove the most important result of this section.

\begin{theorem}\label{injetivity}
	The function $\mathrm{F}^\varepsilon_{vac}:\bR^{1,3}\rightarrow \scF$ is injective.
\end{theorem}
\begin{proof}
	Assume there exist two points $x,y$ such that $\mathrm{F}^\varepsilon_{vac}(x)=\mathrm{F}^\varepsilon_{vac}(y)$. By definition,
	$$
	\Sl\gR_\varepsilon u(x)|\gR_\varepsilon v(x)\Sr=\Sl\gR_\varepsilon u(y)|\gR_\varepsilon v(y)\Sr\quad\mbox{for all } u,v\in\sol^- .
	$$
	Now, consider the special case of $u=u_{\V{0}}^{(\sigma)}$ and $v=u_{\V{p}}^{(\sigma)}$ for arbitrary $\V{p}\in\bR^3$. Then, exploiting (\ref{limit}), we get
	\begin{equation*}
	\begin{split}
	&\Sl\mathfrak{g}_\varepsilon(\V{0})\chi_\uparrow^-(\V{0})|\,\mathfrak{g}_\varepsilon(\V{p})\chi_\uparrow^-(\V{p})\Sr\, e^{i((\omega(\V{p})-m)t_x+\V{p}\cdot\V{x})}=\\
	&=\Sl\mathfrak{g}_\varepsilon(\V{0})\chi_\uparrow^-(\V{0})|\,\mathfrak{g}_\varepsilon(\V{p})\chi_\uparrow^-(\V{p})\Sr\, e^{i((\omega(\V{p})-m)t_y+\V{p}\cdot\V{y})}.
	\end{split}
	\end{equation*}
	Without loss of generality we can assume that $\mathfrak{g}_\varepsilon\neq 0$ in a neighborhood $B_\delta(\V{0})\subset\bR^3$.
	Therefore, if we take $\V{p}$ in this set we can divide the above equations by $\mathfrak{g}_\varepsilon(\V{0})$ and $\mathfrak{g}_\varepsilon(\V{p})$ and get:
	\begin{equation}\label{equazione da semplicifcare}
	\Sl\chi_\uparrow^-(\V{0})| \chi_\uparrow^-(\V{p})\Sr\,e^{i((\omega(\V{p})-m))t_x+\V{p}\cdot\V{x})}= \Sl\chi_\uparrow^-(\V{0})| \chi_\uparrow^-(\V{p})\Sr\,e^{i((\omega(\V{p})-m)t_y+\V{p}\cdot\V{y})}.
	\end{equation}
	At this point, exploiting the identity $\Sl\chi(\V{0})|\,\chi(\V{0})\Sr=-1$ and the continuity of the function $\V{p}\mapsto \Sl\chi(\V{0})|\, \chi(\V{p})\Sr$, shrinking  the neighborhood if necessary, we get
	$$
	\Sl\chi(\V{0})|\, \chi(\V{p})\Sr\neq 0\ \mbox{ for all }\ \V{p}\in B_\delta(\V{0}).
	$$
	This allows to further simplify the above equations and get the following identity:
	\begin{equation}\label{equazioneesponenziali}
	e^{i((\omega(\V{p})-m)t_x+\V{p}\cdot\V{x})}=e^{i((\omega(\V{p})-m)t_y+\V{p}\cdot\V{y})}\ \mbox{ for all }\ \V{p}\in B_\delta(\V{0}).
	\end{equation}
	%
	The norm of the vector $\V{p}$ is bounded by $\delta$, but its direction is unconstrained. Therefore, we can choose a unit vector $\V{n}$ which is orthogonal to both $\V{x}$ and $\V{y}$ and consider the vectors $\V{p}_t:=t\V{n}\in B_\delta(\V{0})$ for $t\in (0,\delta)$. With this choice, identity \eqref{equazioneesponenziali} reduces to 
	$$
	e^{i(\omega(\V{p}_t)-m)t_x}=e^{i(\omega(\V{p}_t)-m)t_y},\ \mbox{ or equivalently }\  	e^{i(\omega(\V{p}_t)-m)(t_x-t_y)}=1,\mbox{ for all } t\in (0,\delta).
	$$
	Now, notice that the function
	$
	(0,\delta)\ni t\mapsto	\omega(\V{p}_t)-m=\sqrt{t^2+m^2}-m
	$
	is continuous and strictly monotonically increasing. From this 
	it is not difficult to see that $t_x$ and $t_y$ must coincide. 
	
	At this point, going back to equation (\ref{equazioneesponenziali}), the terms involving $t_x$ and $t_y$ factor out and the identity reduces to
	$
	e^{i\V{p}\cdot(\V{x}-\V{y})}=1
	$
	for every $\V{p}\in B_\delta(\V{0})$.
	Taking $\V{p}=t\V{e}_i$ with $t\in (0,\delta)$, we obtain
	$$
	e^{it(x_i-y_i)}=1\quad \mbox{for all } t\in (0,\delta),\ i=1,2,3.
	$$
	Reasoning as before we infer that $x_i=y_i$ for any $i=1,2,3$ and the proof is complete.
\end{proof}

\subsubsection{The General Case of Particles and Antiparticles}

We now address the following question: what happens if we add particles or antiparticles? As in the case of regularity, the injectivity of $\mathrm{F}^\varepsilon$ is preserved if additional positive-energy physical solutions are added to the system, while the creation of holes is a much more delicate matter. Again, we need to find suitable assumptions in order not to lose too much information. 

In the case of regular systems, a sufficient condition is given in the following theorem. Remember that a sufficient condition for a system with holes to be regular is the assumption of being \textit{$\varepsilon$-macroscopic} (see Theorem \ref{emacroscopicregular}). Again, for simplicity, we focus our attention to smooth solutions. In particular, we assume that the solutions describing the particles and antiparticles have compactly supported three-momentum distributions. 

\begin{theorem}\label{theoreminjectivityparticle}
	Let $\sU_\pm\subset\hat{\mathrm{E}}(\hat{P}_-(\mathcal{C}_{0,p}^\infty(\bR^3,\bC^4)))$ be finite-dimensional subspaces and assume that the causal fermion system $a_-(\sU_-)$ is regular. Then, the local correlation function $\rF^\varepsilon:\bR^{1,3}\rightarrow\scF$ of the causal fermion system $a(\sU_-,\sU_+)$ is injective.
\end{theorem}
\begin{proof}
	It is sufficient to focus on  the case $\sU_+=\{0\}$. Indeed, once we have proven that $a_-(\sU_-)$ admits an injective local correlation function, then the injectivity for $a(\sU_-,\sU_+)$ will follow directly from Proposition \ref{propprojcetion}.  Therefore, let us consider the pure negative case.  By $\rF^\varepsilon_{vac}$ we denote the local correlation function associated with $\sol^-$, i.e. the vacuum. Fix any two different points $x,y\in\bR^{1,3}$ and suppose by contradiction that $\rF^\varepsilon(x)=\rF^\varepsilon(y)$. This can be rewritten as $\Pi(\rF^\varepsilon_{vac}(x)-\rF^\varepsilon_{vac}(y))\Pi=0$, where $\Pi$ is the orthogonal projector on $a_-(\sU_-)\subset\scH_m^-$. Since the causal fermion system is assumed to be regular, there must exist at least one element $u\in a_-(\sU_-)$ such that $\gR_{\varepsilon} u(x)= e_0$ (this will be used later). By assumption, $(\rF^\varepsilon_{vac}(x)-\rF^\varepsilon_{vac}(y))u\in \sU_-$ and therefore, exploiting Proposition \ref{elemetnsS}:
	$$
	P^{\varepsilon}(\,\cdot\,,x)\gR_{\varepsilon} u(x)-P^{\varepsilon}(\,\cdot\,,y)\gR_\varepsilon u(y) = \omega,
	$$
	for some $\omega\in\sU_-$. 
	More explicitly, we have for  $t=0$ and arbitrary $\V{z}\in\bR^3$ that
	\begin{equation*}
	\begin{split}
	&-\!\int_{\bR^3}\frac{d^3\V{k}}{(2\pi)^4}\,\mathfrak{g}_\varepsilon (\V{k})\,p_-(\V{k})\,\gamma^0\left(\gR_\varepsilon u(x)e^{i\eta(k,x)}-\gR_\varepsilon u(y)e^{i\eta(k,y)}\right)\,e^{i \V{k}\cdot\V{z}}=\\
	&=\int_{\bR^3}\frac{d^3\V{k}}{(2\pi)^{4}}\,\varphi({\V{k}})\,e^{i\V{k}\cdot\V{z}}
	\end{split}
	\end{equation*}
	for some $\varphi\in \hat{P}_-(\mathcal{C}_{0,p}^\infty(\bR^3,\bC^4))\subset \mathcal{C}_{0,p}^\infty(\bR^3,\bC^4)$, where we used $k^0=-\omega(\V{k})$. Since the Fourier Transform is a bijection on the Schwartz space, the above identity implies
	\begin{equation}
	\begin{split}
	-\mathfrak{g}_\varepsilon (\V{k})\,p_-(\V{k})\,\gamma^0\left(\gR_\varepsilon u(x)e^{i\eta(k,x)}-\gR_\varepsilon u(y)e^{i\eta(k,y)}\right)=\varphi(\V{k})\quad\mbox{for all }\V{k}\in\bR^3.
	\end{split}
	\end{equation}
	At this point, notice that the function $\mathfrak{g}_\varepsilon$ cannot be of compact support, as we know that it vanishes at most on a countable family of spheres which are isolated from each other (see the proof of Proposition \ref{convoltioncomesabout}). Therefore, bearing in mind that $\mathfrak{g}_\varepsilon$ is rotationally symmetric, there must exists some annulus $B(r,R)\subset\bR^3$ ($0<r<R$) on which $\mathfrak{g}_\varepsilon>0$ and $\varphi\equiv 0$ and therefore:
	$$
	p_-(\V{k})\,\gamma^0\left(\gR_\varepsilon u(x)e^{i\eta(k,x-y)}-\gR_\varepsilon u(y)\right)=0\quad\mbox{for all } \V{k}\in B(r,R).
	$$
	This implies  (see Section \ref{sectionhamiltoniana})
	$$
	\gamma^0\left(\gR_\varepsilon u(x)e^{i\eta(k,x-y)}-\gR_\varepsilon u(y)\right)\in W_\V{k}^+\quad\mbox{for all $\V{k}\in B(r,R)$},
	$$ 
	or, equivalently (exploiting the form of the gamma matrix $\gamma^0$),
	$$
	\gR_\varepsilon u(x)e^{i\eta(k,x-y)}-\gR_\varepsilon u(y)=\lambda_\uparrow (\V{k})\left(
	\begin{matrix}
	e_{\uparrow}\\[2ex]
	\dfrac{\V{\sigma}\cdot\V{k}}{\omega(\V{k})+m}\, e_{\uparrow}
	\end{matrix}
	\right)+
	\lambda_{\downarrow}(\V{k})\left(
	\begin{matrix}
	e_{\downarrow}\\[2ex]
	\dfrac{\V{\sigma}\cdot\V{k}}{\omega(\V{k})+m}\, e_{\downarrow}
	\end{matrix}
	\right),
	$$
	for some continuous functions $\lambda_{\uparrow,\downarrow}$. 
	At this point, we finally exploit that $\gR_\varepsilon u(x)=e_0$, which gives
	\begin{equation}\label{setequations}
	\begin{split}
	&\lambda_\uparrow(\V{k})= \left(\gR_\varepsilon u(x)e^{i\eta(k,x-y)}-\gR_\varepsilon u(y)\right)_0=e^{i\eta(k,x-y)}-[\gR_{\varepsilon} u(y)]_0, \\[0.3em]
	&\lambda_{\downarrow}(\V{k})=\left(\gR_\varepsilon u(x)e^{i\eta(k,x-y)}-\gR_\varepsilon u(y)\right)_1=-[\gR_\varepsilon u(y)]_1=:\lambda_\downarrow=\mbox{const},\\[0.3em]
	& \dfrac{\lambda_\uparrow(\V{k})}{\omega(\V{k})+m}\left(\begin{matrix} k_3 \\ k_1+ik_2 \end{matrix}\right)+\frac{\lambda_{\downarrow}}{\omega(\V{k})+m}\left(\begin{matrix} k_1-ik_2 \\ -k_3 \end{matrix}\right)= \left(\begin{matrix}
	[\gR_\varepsilon u(y)]_2\\
	[\gR_\varepsilon u(y)]_3
	\end{matrix}\right)=\mbox{const}.
	\end{split}
	\end{equation}
	\vspace{0.05cm}
	
	\noindent From the third identity, taking $\V{k}=s\V{e}_1$ and $\V{k}=s\V{e}_2$ with arbitrary $r<s<R$, we see that $s\lambda_{\downarrow}=-is\lambda_{\downarrow}$ which implies $\lambda_{\downarrow}=0$.
	The last identity, then, reduces to
	\begin{equation}
	\frac{\lambda_\uparrow(\V{k})}{\omega(\V{k})+m}\left(\begin{matrix}  k_3   \\ k_1 + ik_2\end{matrix}\right)= \left(\begin{matrix}
	[\gR_\varepsilon u(y)]_2\\
	[\gR_\varepsilon u(y)]_3
	\end{matrix}\right).
	\end{equation}
	Choosing $\V{k}= s\V{e}_1$
	we get $[\gR_\varepsilon u(y)]_2=0$ and therefore $\lambda_{\uparrow}(\V{k})k_3=0$ for any $\V{k}\in B(r,R)$. The continuity of $\lambda_\uparrow$ yields $\lambda_{\uparrow}(\V{k})=0$ on $B(r,R)$. In particular, the first identity in (\ref{setequations}) becomes
	$$
	e^{-i(\omega(\V{k})(t_x-t_y)+\V{k}\cdot(\V{x-y}))}=[\gR_\varepsilon u(y)]_0=\mbox{constant}\quad\mbox{for all }\V{k}\in B(r,R).
	$$
	At this point, reasoning as in the proof of Theorem \ref{injetivity} and get $t_x=t_y$ and $\V{x}=\V{y}$.
\end{proof}
\begin{remark}
	It should be remarked that the assumption of compactly supported momentum distributions  was crucial only for negative-energy solutions and it was not used for positive-energy solutions. Therefore, the condition of compact support can be dropped for $\sU_+$ in Theorem \ref{theoreminjectivityparticle}.  Again, this reveals an asymmetry between particles and antiparticles.
\end{remark}

Again, we see that the addition of positive-energy solutions to $\sol^-$ enriches the system and also injectivity (as regularity) is preserved. On the other hand, removing negative-energy solutions may perturb the system in a way that injectivity (as regularity) is lost. This is prevented if the removed solutions do not play a crucial role in the construction of the causal fermion system: in other word if their distribution is sufficiently regular in spacetime.

\begin{remark}
	Notice that the regularity of $a_-(\sU_-)$ in  Theorem \ref{theoreminjectivityparticle} implies that the causal fermion system $a(\sU_-,\sU_+)$ is regular as well (see Corollary \ref{regularitygeneralcase}). In this and in the next section we will always focus on regular systems.
\end{remark}

\subsection{Step 2: Closedness of the Local Correlation Function}

So far we have showed that the local correlation function is continuous and, under suitable assumptions, also injective. However, we have no knowledge, yet, on the continuity of its inverse. If this could be proved, then $\rF^\varepsilon$ would define a global homeomorphism from $\bR^{1,3}$ to the image $\rF^\varepsilon(\bR^{1,3})$, giving the latter a structure of smooth manifold.

Another critical property that we would like to prove is the closedness of the local correlation function. In such a way, the image $\rF^\varepsilon(\bR^{1,3})$ would be a closed subset of $\scF$ and therefore coincide with $\overline{\rF^\varepsilon(\bR^{1,3})}=\supp (\rF^\varepsilon)_*\mu$, showing in this way that the causal fermion system selects exactly $\bR^{1,3}$ as the support of its measure.
In order to prove all this we need some technical results.

A function $f:X\rightarrow Y$ between two topological spaces is called \textit{proper} if the inverse image of any compact set is compact. In \cite[Section 3.7]{engelk} (see in particular Theorem 3.7.18)  it is proved that every such a function is also \textit{closed} if $Y$ is a $k$-space. First-countable Hausdorff spaces are always of this kind (see Theorem 3.3.20): in particular this applies to $\scF$.

\subsubsection{The Vacuum}
Let us start with the vacuum.
\begin{theorem}\label{closedF}
	The local correlation function $\rF^\varepsilon_{vac}:\bR^{1,3}\rightarrow\scF$ of the vacuum is closed.
\end{theorem}
\begin{proof}
	Let $K\subset\scF$ be compact (hence closed) and consider any sequence $\{x_n \}_n\subset H:=(\rF^\varepsilon_{vac})^{-1}(K)\subset\bR^{1,3}$. We want to prove that there exists some subsequence which converges to some element in $H$. This gives the compactness of $H$. 
	
	Since $K$ is compact and $\{\rF^\varepsilon_{vac}(x_n) \}_n\subset K$, there must exist some subsequence $\{y_{n} \}_n$ and some $A\in K$ such that $\rF^\varepsilon_{vac}(y_{n})\to A$.
	Notice that $A^*=A$. 
	Now, we want to prove that $A\neq 0$. Assume by contradiction that $A=0$, then $\|\rF^\varepsilon_{vac}(y_{n})\|\to 0$.
	Let $\lambda\in\sigma_{vac}\setminus\{0\}$ and take $u\in\sol^- $ such that $\rF^\varepsilon_{vac}(y_{1})u=\lambda u$, then $u_n:=\mathrm{U}_{y_1-y_n}u$ satisfies $\rF^\varepsilon_{vac}(y_n)u_n=\lambda u_n$, as follows from Proposition \ref{spectrumF}.
	Thus, we have
	$$
	|\lambda|\|u\|=|\lambda|\|u_n\|=\|\rF^\varepsilon_{vac}(y_n)u_n\|\le \|\rF^\varepsilon_{vac}(y_n)\|\|u\|\to 0,
	$$
	which is impossible, it being $u\neq 0$.
	
	Now, notice that $\hat{\mathrm{E}}(\hat{P}_-(\mathcal{S}_p(\bR^3,\bC^4)))$ is a dense subspace of $\sol^- $, as proved in Lemma \ref{lemmaproizionieschwat}, and its elements are exactly the functions of the form
	$$
	u_\varphi(x):=\int_{\bR^3}\varphi(\V{k})e^{i(\omega(\V{k})t+\V{k}\cdot\V{x})}\,d^3{\V{k}}\quad \left(\gR_\varepsilon u_\varphi(x)=\int_{\bR^3}\varphi(\V{k})\mathfrak{g}_\varepsilon(\V{k})e^{i(\omega(\V{k})t+\V{k}\cdot\V{x})}\,d^3{\V{k}}\right).
	$$	
	with $\varphi\in \hat{P}_-(\mathcal{S}_p(\bR^3,\bC^4))$ (see Proposition \ref{generalsolutionsmooth}).
	Therefore, as $A$ is bounded, self-adjoint and different from zero, there must exist at least one $\varphi\in \hat{P}_-(\mathcal{S}_p(\bR^3,\bC^4))$ such that $(u_\varphi|Au_\varphi)\neq 0$.
	
	At this point, we have two possibilities: either $\{y_n \}_n$ is bounded or it is not. If $\{y_n\}_n$ is unbounded, then we can always extract some subsequence $\{z_{n} \}_n$ which diverges to infinity: $|z_n|\to \infty$. Since $\varphi \mathfrak{g}_\varepsilon\in\mathcal{S}_p(\bR^3,\bC^4)$, Theorem \ref{hormander} implies that
	$
	\gR_\varepsilon u_\varphi(z_n)\to 0
	$
	in the limit  $n\to\infty$ and, therefore,
	$$
	\langle u_\varphi|Au_\varphi\rangle=\lim_{n\to\infty}\langle u_\varphi|\rF^\varepsilon_{vac}(z_{n})u_\varphi\rangle=-\lim_{n\to\infty}\Sl\gR_\varepsilon u_\varphi(z_{n})\,|\, \,\gR_\varepsilon u_\varphi(z_{n})\Sr=0
	$$
	for all $\varphi\in\mathcal{S}_p(\bR^3,\bC^4)$, which is a contradiction. Therefore, $\{y_{n} \}_n$ must be bounded and so there exists some subsequence $\{w_{n} \}_n$ which converges to some $x\in\bR^{1,3}$.
	To conclude, notice that $H$ is closed, it being the continuous inverse of $K$, and therefore
	$\{w_n \}_n\subset H$ and $w_n\to x$ imply $x\in H$. The proof is complete.
\end{proof}

\subsubsection{The General Case of Particles and Antiparticles}
It is possible to show that closedness is preserved in presence of particles or antiparticles, under suitable assumptions in the latter case, as for injectivity. Again, we stick to the case of smooth solutions, though we do not need to put any restriction on their support this time.

\begin{theorem}\label{teoremaclosedparticle}
	Let $\sU_\pm\subset\hat{\mathrm{E}}(\hat{P}_\pm(\mathcal{S}_p(\bR^3,\bC^4)))$ be finite-dimensional subspaces and assume that $a_-(\sU_-)$ is regular. Moreover, assume that there exists a Hilbert basis $\boldsymbol{e}_-$ of $\sU_-$ such that
	\begin{equation}\label{stima}
	\cE(\boldsymbol{e}_-,\varepsilon,x)< \sqrt{|\pi\,\nu(\varepsilon)^-|}\quad\mbox{for every }x\in\bR^{1,3}.
	\end{equation}
	Then the local correlation function of the causal fermion system $a(\sU_-,\sU_+)$ is  closed.
\end{theorem}
\begin{proof}
	Let us start with the pure negative-energy case first, i.e. $\sU_+=\{0\}$. We proceed similarly as in the proof of Theorem \ref{closedF}.  
	Let $K\subset\scF$ be compact (hence closed) and consider any sequence $\{x_n \}_n\subset H:=(\rF^\varepsilon)^{-1}(K)\subset\bR^{1,3}$. We want to prove that there exists a subsequence which converges to some element in $H$. This gives the compactness of $H$. 
	Since $K$ is compact and $\{\rF^\varepsilon(x_n) \}_n\subset K$, there must exist some subsequence $\{y_{n} \}_n$ and some $A\in K$ such that $\rF^\varepsilon(y_{n})\to A$.
	Notice that $A^*=A$. Moreover $A\restr_{a_-(\sU_-)^\perp}=0$, as the same holds for any $\rF^\varepsilon(y_n)$.
	Nevertheless, we want to show that $A\neq 0$. Assume by contradiction that $A=0$, then in particular $\|\rF^\varepsilon(y_n)\|\to 0$. This implies that 
	$$
	m_n:=\max |\sigma(\rF^\varepsilon(y_n))|=\|\rF^\varepsilon(y_n)\|\to 0,
	$$ 
	and therefore all the eigenvalues $\nu_i(n)$ of $\rF^\varepsilon(y_n)$ will converge to zero, since $|\nu_i(n)|\le m_n$. Nevertheless, this is not possible, because Theorem \ref{approximationeigenvalues} would imply that, for any $i=1,2,3,4$:
	\begin{equation}
	\begin{split}
	|2\pi\,\nu^-(\varepsilon)|&=\min\left\{|2\pi\,\nu^-(\varepsilon)|,|2\pi\,\nu^+(\varepsilon)|\right\}=\\
	&= \lim_{n\to \infty}  \min\left\{|\nu_i(n)-2\pi\,\nu^-(\varepsilon)|,|\nu_i(n)-2\pi\,\nu^+(\varepsilon)|\right\} \le\\
	&\le 2\,\cE (\V{e}_-,\varepsilon,x)^2< {|2\pi\,\nu^-(\varepsilon)|},
	\end{split}
	\end{equation}
	which is a contradiction. Now, notice that the orthogonal projector over $a_-(\sU_-)$ is given exactly by the function $\Psi$ defined in Proposition \ref{part2} (with $\boldsymbol{\psi}=\boldsymbol{u}=\boldsymbol{e}_-$  any Hilbert basis of $\sU_-\subset \hat{\mathrm{E}}(\hat{P}_-(\mathcal{S}_p(\bR^3,\bC^4)))$) i.e. $a_-(\sU_-)=\Psi[\sol^-]$. Moreover, it is not difficult to see that 
	$$
	\Psi\left[\hat{\mathrm{E}}(\hat{P}_-(\mathcal{S}(\bR^3,\bC^4))\right]\subset \hat{\mathrm{E}}(\hat{P}_-(\mathcal{S}(\bR^3,\bC^4))\quad\mbox{is dense within $\Psi[\sol^-]$}.
	$$ 
	At this point, the proof follows as in Theorem \ref{closedF}, replacing $\scH_m^-$ by $a_-(\sU_-)$ and $\hat{\mathrm{E}}(\hat{P}_-(\mathcal{S}(\bR^3,\bC^4))$ by $\Psi\left[\hat{\mathrm{E}}(\hat{P}_-(\mathcal{S}(\bR^3,\bC^4))\right]$.
	
	Now, suppose we add some positive energy solutions to the system, i.e. we consider the causal fermion system $a(\sU_-,\sU_+)$ with $\sU_+\neq \{0\}$.  Again, consider any compact set $K\subset\scF$ and let $\{x_n \}_n\subset H:=(\rF^\varepsilon)^{-1}(K)\subset\bR^{1,3}$. Since $K$ is compact, there exists a subsequence $\{y_n \}_n$ of $\{x_n\}_n$ and some $A\in K$ such that $\rF^\varepsilon(y_n)\to A$. Let us denote by $\rF_-^\varepsilon$ the local correlation function associated to the causal fermion system $a_-(\sU_-)$ analyzed in the first part of the proof. Thanks to Proposition \ref{propprojcetion}, we know that $\rF^\varepsilon_-(z)=\Pi_-\,\rF^\varepsilon(z)\, \Pi_-=:A_-$ (with $\Pi_-$ the orthogonal projector onto $a_-(\sU_-)$) for every $z\in\bR^{1,3}$ and therefore $\rF^\varepsilon_-(y_n)\to \Pi_-A\Pi_-$. The argument continues as in the first part of the proof,  replacing $\rF^\varepsilon$ by $\rF^\varepsilon_-$ and $A$ by $A_-$.
\end{proof}

\begin{remark}
In order to get some intuition, referring to the brute approximation carried out in (\ref{stime}), the estimate in (\ref{stima}) is approximately given by
	$$
	\cE(\boldsymbol{e}_-,\varepsilon,x)\lesssim\left(\frac{1}{16\pi^2}\left(\frac{2}{3\,\varepsilon^{3}}-\frac{m}{\varepsilon^2} \right)\right)^{1/2}=\left[\frac{1}{4\pi}\left(\frac{2}{3(m\varepsilon)^3}-\frac{1}{(m\varepsilon)^2}\right)^{1/2}\right] m^{3/2}.
	$$
	The right-hand term is very large under the assumption $m\varepsilon<m\varepsilon_{max}=10^{-15}$ if compared with (\ref{macroscopicond}).
\end{remark}

We are ready to state the most important result of this section. 

\begin{theorem}\label{teoremamanifold}
	Under the assumptions of Theorems \ref{theoreminjectivityparticle} and \ref{teoremaclosedparticle} the following statements hold.
	\vspace{0.1cm}
	\begin{itemize}[leftmargin=2.5em]
		\item[\rm{(i)}] $\rF^\varepsilon(\bR^{1,3})$ is a closed subset of $\scF$,\\[-0.3cm]
		\item[\rm{(ii)}] $\rF^\varepsilon:\bR^3\rightarrow \rF^\varepsilon(\bR^{1,3})$ is a homeomorphism,\\[-0.3cm]
		\item[\rm{(iii)}] $\rF^\varepsilon(\bR^{1,3})=\supp(\rF^\varepsilon)_*\mu$ is a 4-dimensional smooth manifold.
	\end{itemize}
\end{theorem}
\begin{proof}
	The proof follows directly from the injectivity, continuity and closedness of~$\rF^\varepsilon$.
\end{proof}

It is actually possible show that even the canonical foliation of $\bR^{1,3}$ into space and time is preserved. 
Let us introduce the following symbol:
$$
\hat{\Sigma}_t:=\rF^\varepsilon(\Sigma_t)\quad t\in\bR.
$$
Then the following result can be proved as the previous theorem.
\begin{corollary}\label{teoremamanifold2}
	Under the assumptions of Theorems \ref{theoreminjectivityparticle} and \ref{teoremaclosedparticle} the following statements hold for every $t\in\bR$.
	\vspace{0.1cm}
	\begin{itemize}[leftmargin=2.5em]
		\item[\rm{(i)}] $\hat{\Sigma}_t$ is a closed subset of $\scF$,\\[-0.3cm]
		\item[\rm{(ii)}] $\rF^\varepsilon\restr_{\Sigma_t}:\bR^3\rightarrow\Sigma_t$ is a homeomorphism,\\[-0.3cm]
		\item[\rm{(iii)}] $\hat{\Sigma}_t$ is a three-dimensional smooth manifold.
	\end{itemize}
	In particular, $\supp(\rF^\varepsilon)_*\mu$ admits a smooth foliation in terms of 3-dimensional submanifolds $\{\hat{\Sigma}_t\}_{t\in\bR}$.
\end{corollary}

As a conclusion of this section, we go back to the discussion started after Definition~\ref{defCFS}.
Consider the set $\scH_m^-$ of negative-energy physical solutions and let $\rF^\varepsilon_{vac}$ be its local correlation function. We want to give some arguments which support the choice of $(\rF^\varepsilon_{vac})_*\mu$ as Borel measure on $\scF$ realizing $\rF^\varepsilon_{vac}(\bR^{1,3})$ as support.

Proposition \ref{propositioninvariancetransl} and Proposition \ref{spectrumF} show that the vacuum causal fermion system carries a unitary representation of the translation group. 
It is sensible to assume that any Borel measure $\varrho$ which is meant to describe the vacuum is invariant under such transformations:
\begin{equation}\label{invariancemeasure}
\varrho((\mathrm{U}_a)^\dagger\, \Omega\,\mathrm{U}_a))=\varrho(\Omega)\ \mbox{ for all }\Omega\in\mathfrak{Bor}(\scF)\ \mbox{ and } a\in\bR^{1,3}. 
\end{equation}
This assumption does not leave much freedom in choosing $\rho$, as the following proposition clarifies.
\begin{proposition}\label{remarkunicitymeasure}
	Let $\mu$ denote the Lebesgue-Borel measure on $\bR^{1,3}$ and let $\varrho$ be a Borel measure on $\scF$ which fulfills (\ref{invariancemeasure}) and such that 
	$$
	\supp\varrho=\mathrm{F}_{vac}^\varepsilon(\bR^{1,3})
	\quad\mbox{and}\quad
	\varrho(K)<\infty\mbox{ for every compact $K\subset\scF$. }$$ Then there exists $\lambda\ge 0$ such that $\varrho=\lambda\, (\rF^\varepsilon_{vac})_*\mu$.
\end{proposition} 
\begin{proof}
	For simplicity of notation we drop the indices $\varepsilon$ and $vac$. To start with, we show that the push-forward measure does satisfy condition (\ref{invariancemeasure}). First, notice that Proposition \ref{spectrumF} implies $\mathrm{U}_a^\dagger\, \rF(\bR^{1,3})\,\mathrm{U}_a=\rF(\bR^{1,3})$ for every $a\in\bR^{1,3}$. Therefore, given that $\supp \rF_*\mu=\rF(\bR^{1,3})$ and the unitarity of $\mathrm{U}_a$, in order to prove (\ref{invariancemeasure}) we can stick to the Borel subsets of $\rF(\bR^{1,3})$. Indeed, let $\Omega$ be any Borel set of $\scF$, then:
	\begin{equation*}
	\begin{split}
	\rF_*\mu(\mathrm{U}_a^\dagger\,\Omega\, \mathrm{U}_a)&=\rF_*\mu(\mathrm{U}_a^\dagger\,\Omega\, \mathrm{U}_a\cap \rF(\bR^{1,3}))=\rF_*\mu(\mathrm{U}_a^\dagger\,\Omega\, \mathrm{U}_a\cap \mathrm{U}_a^\dagger\,\rF(\bR^{1,3})\,\mathrm{U}_a)=\\
	&=\rF_*\mu(\mathrm{U}_a^\dagger\,(\Omega\cap\rF(\bR^{1,3}))\, \mathrm{U}_a).
	\end{split}
	\end{equation*}
	Thus, we can focus on the Borel sets of $\rF(\bR^{1,3})$. Let $\Omega$ be any of them, then
	\begin{equation*}
	\begin{split}
	(\mathrm{U}_a)^\dagger\, \Omega\,\mathrm{U}_a&=\{(\mathrm{U}_a)^\dagger\rF(x)\mathrm{U}_a\:|\: x\in\rF^{-1}(\Omega)\}=\\
	&=\{\rF(x+a)\:|\: x\in\rF^{-1}(\Omega)\}=\rF(\rF^{-1}(\Omega)+a).
	\end{split}
	\end{equation*}
	Therefore, as the Lebesgue measure is invariant under translations, we get
	$$
	\rF_*\mu((\mathrm{U}_a)^\dagger\, \Omega\,\mathrm{U}_a)=\mu(\rF^{-1}(\Omega)+a)=\mu(\rF^{-1}(\Omega))=\rF_*\mu(\Omega).
	$$
	To conclude the proof we have to show uniqueness (up to a non-negative factor) of such a measure. Assume there exists another measure $\varrho$ as in the assumptions of the proposition. Since $\supp\varrho=\mathrm{F}_{vac}^\varepsilon(\bR^{1,3})$, again we can stick our analysis to the Borel subsets of $\rF(\bR^{1,3})$, which coincide with the images of the Borel subsets of $\bR^{1,3}$ through $\rF$, the latter being a homeomorphism. Now, fix a Borel set $\Delta\subset\bR^{1,3}$  and a vector $a\in\bR^{1,3}$, then we have
	$$
	\varrho(\rF(\Delta+a))=\varrho((\mathrm{U}_a)^\dagger\,\rF(\Delta)\,\mathrm{U}_a)=\varrho(\rF(\Delta)).
	$$
	Now, as $\rF$ is a homeomorphism onto its image, the function $\Delta\mapsto \varrho(\rF(\Delta))$ defines a Borel measure on $\bR^{1,3}$, which is translation invariant and is finite on compact subsets. It is a well-known fact (see \cite[Theorem 2.20]{rudin}) that every Borel measure fulfilling these properties must be a non-negative multiple of the Lebesgue-Borel measure $\mu$. Therefore, there exists $\lambda\ge 0$ such that, for any Borel set $\Omega\subset\rF(\bR^{1,3})$,
	$$
	\lambda\,\rF_*\mu(\Omega)=\lambda\,\mu(\rF^{-1}(\Omega))=\varrho(\rF(\rF^{-1}(\Omega)))=\varrho(\Omega),
	$$
	and the proof is complete.
\end{proof}
\section{Some concluding remarks}
%

In constructing a causal fermion systems $(\mathscr{H},\rF^\varepsilon)$, we start from Minkowski space $\bR^{1,3}$ and represent it in $\scF$ through the map $\mathrm{F}^\varepsilon$. Then, we push-forward the Lebesgue measure $\mu$ on $\bR^{1,3}$ to a measure $\varrho=(\rF^\varepsilon)_*\mu$ on $\scF$, which satisfies $\supp\varrho=\overline{\rF^\varepsilon(\bR^{1,3})}$. 

If now we focus on the special cases treated in the previous sections, for which the corresponding local correlation function is closed and homeomorphic onto its image, then it is possible to discuss the relations existing between different causal fermion systems without ever referring to Minkowski space. In this way, we may be able to highlight specific underlying structures which hold in more general settings than Minkowski space.

%

Consider some finite-dimensional subspaces $\sU_\pm^{(i)}\subset\sol^\pm,\ i=1,2 $ as in the previous section, and construct the corresponding causal fermion systems. Also, consider the vacuum system (this is, actually, a special case of the former with $\sU_\pm=\{0\}$, but it is convenient to consider it separately). Thus, we have
$$
\big(a\big(\sU_-^{(i)},\sU_\pm^{(i)}\big),\rF_i,\varrho_i\big),\ \ i=1,2\quad\mbox{and}\quad \big(\scH_m^-,\rF_{vac},\varrho_{vac}\big),
$$
where we dropped the index $\varepsilon$, for simplicity of notation. The first pair of casual fermion systems describes two possible configurations of matter, with a few particles and a few antiparticles. The second one represents the vacuum.

By definition, the corresponding measures are defined as the push-forward of the Lebesgue measure to $\scF$ through the local correlation functions:
$$
\varrho_i=(\mathrm{F}_i)_*\mu,\ \ i=1,2\quad\mbox{and}\quad \varrho_{vac}=(\rF_{vac})_*\mu.
$$
We also introduce the following notation for the supports:
$$
\sM_i:=\supp\varrho_i=\rF_i(\bR^{1,3}),\ \ i=1,2\quad\mbox{and}\quad \sM_{vac}=\supp\varrho_{vac}=\rF_{vac}(\bR^{1,3}).
$$
At this point, we put all these mathematical structures together in a commutative diagram:
%
%
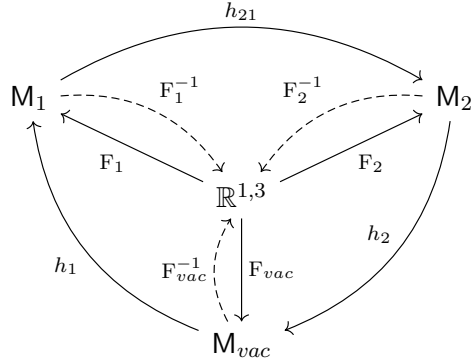
\begin{displaymath}
\begin{tikzcd}
\sM_1\arrow[rrrr, bend left, "h_{21}"]\arrow[drr, dashed, bend left, "\mathrm{F}_1^{-1}"] & & &	&	\sM_2 \arrow[dddll, bend left, "h_2"'] \arrow[dll, dashed, bend right, "\mathrm{F}_2^{-1}"']\\	
&	&\bR^{1,3}\arrow[ull,"\mathrm{F}_1"] \arrow[urr,"\mathrm{F}_2"'] \arrow[dd,"\mathrm{F}_{vac}"]\\
& &\\	
&& \sM_{vac}\arrow[uu, bend left, dashed, "\mathrm{F}_{vac}^{-1}"]	\arrow[uuull, bend left, "h_1"] 
\end{tikzcd}
\end{displaymath}
%
where we defined the homeomorphisms:
\begin{equation*}
\begin{split}
h_1:=\rF_1\circ\rF_{vac}^{-1}:&\sM_{vac}\rightarrow\sM_1,\quad h_2:=\rF_{vac}\circ\rF_{2}^{-1}:\sM_{2}\rightarrow\sM_{vac}\\
 &h_{21}:\rF_2\circ\rF_1^{-1}:\sM_1\rightarrow\sM_2.
\end{split}
\end{equation*}
By construction, every set appearing in the diagram above is a smooth manifold and every function is a diffeomorphism. 
Moreover, it is possible to rewrite every measure as the push-forward of one another, as the defining functions are homeomorphisms. 
More precisely, it follows by definition of push-forward that
$$
\varrho_1=(h_1)_*\varrho_{vac},\quad \varrho_{vac}=(h_2)_*\varrho_2,\quad \varrho_2=(h_{21})_*\varrho_1.
$$
The identities above translate into the mathematical language of causal fermion systems what is generally known as the action of a \textbf{creation or annihilation operator}.
More precisely, we can give the following interpretations.
\vspace{0.1cm}
\begin{itemize}[leftmargin=2.5em]
	\item[\rm{(i)}] The identity $\varrho_1=(f_1)_*\varrho_{vac}$ defines the creation of the particles of $\sU_+^{(1)}$ and  the antiparticles of $\sU_-^{(1)}$, starting from the vacuum.\\[-0.3cm]
	\item[\rm{(ii)}] The identity $\varrho_{vac}=(f_2)_*\varrho_2$  defines the annihilation of the particles of $\sU_+^{(2)}$ and  the antiparticles of $\sU_-^{(2)}$, ending up in the vacuum.\\[-0.3cm]
	\item[\rm{(iii)}] The identity $\varrho_2=(f_{21})_*\varrho_1$ defines the annihilation of $\sU_\pm^{(1)}$ and the creation of~$\sU_\pm^{(2)}$.
\end{itemize}

In the general case, when  the local correlation function is not a homeomorphism onto its image, the discussion presented here no longer applies. However, in some special cases it is possible to get similar results, even though the functions involved may only be measurable. This, though, goes beyond the scope of this paper and will not be discussed here.

\vspace{0,5cm}


\section*{Acknowledgments}
I would like to thank Felix Finster for the  fruitful discussions and the several advices and hints which helped me a lot in writing this paper. I am grateful to my colleagues, in particular Maximilian Jokel, Christoph Langer and Andreas Platzer, for their interest in this work and the useful exchange of ideas, and to anonymous referees whose comments and suggestions helped me to improve and clarify this manuscript. A final thank goes to a dear friend, Davide De Boni, for pointing out some unforgivable mistakes.

\section{{Appendix}}

\begin{proof}[Proof of Proposition \ref{causalpropag}]
	Point (i) is a direct consequence of the fact that two solutions of (\ref{solutiondirac}) which coincide on the past boundary of a lens-shaped region do coincide on the whole region (see the corresponding section in \cite{intro}; also see \cite[Section 2.1]{DM}). More precisely, any solution whose restriction to  $\Sigma_t$ has compact support must coincide with the trivial solution on the set complement of $\{t\}\times\supp\varphi + \{x\in\bR^{1,3}\:|\: \eta(x,x)\ge 0 \}$. The proof of (ii) follows from (i) and from Theorem \ref{existenceuniqueness}-(i).
\end{proof}
\begin{proof}[Proof of Proposition \ref{currentconserv}]
	The positive-definiteness of $(\cdot|\cdot)_t$ follows directly from the analogous property of $(\cdot|\cdot)_{\mathcal{L}^2}$ and the uniqueness of the solutions of (\ref{solutiondirac}) for given  initial data at time $t\in\bR$.  The function $\mathrm{E}_t$ is an isometry by construction. Finally, the independence of the inner product from the time variable can be proved by differentiating under the integral sign, writing the time derivative in terms of spatial derivatives by means of equation (\ref{solutiondirac}) and eventually applying the divergence theorem (note that the involved functions have spatially compact support).
\end{proof}

\begin{proof}[Proof of Lemma \ref{lemmaconverging}]
	Consider $f,g\in\Sol$ and fix any $T>0$. Then, from current conservation we have
	\begin{equation}\label{equazioniintegrali}
	(f\restr_{R_T}|g\restr_{R_T})_{\mathcal{L}^2}=\int_{R_T}f(x)^\dagger g(x)\, d^4x=\int_{-T}^Tdt\int_{\bR^3}f(t,\V{x})^\dagger g(t,\V{x})\, d^3\V{x}=2T(f|g)_0,
	\end{equation}
	which in particular gives identity (\ref{identityintegralstrip}) for functions in $\Sol$.
	From this we see that $\{f_n \}_n$ is Cauchy in $\Sol$ if and only if $\{f_n|_{R_T}\}_n$ is Cauchy in $\mathcal{L}^2(R_T,\bC^4)$ for any $T>0$.  
	In particular, since every open bounded set is contained within a suitable $R_T$, we see that every Cauchy sequence in $\Sol$ is also a Cauchy sequence in $\mathcal{L}^2_{loc}(\bR^{1,3},\bC^4)$ and the proof is complete.
\end{proof}

\begin{proof}[Proof of Lemma \ref{lemmaconverging2}]
	Let $\{f_n \}_n$ be a Cauchy sequence in $\Sol$ and let $u$ denote its limit in $\mathcal{L}^2_{loc}(\bR^{1,3},\bC^4)$.	Let us start by proving point (i). 
	As shown in the proof of Lemma \ref{lemmaconverging}, $\{f_n\}_n$ is also Cauchy in $\mathcal{L}^2(R_k,\bC^4)$ for every $k\in\bN$ and, therefore, for any $k\in\bN$ there must exist a function $\tilde{u}_k\in \mathcal{L}^2(R_k,\bC^4)$ such that $\|f_n\restr_{R_k}-\tilde{u}_k\|_{\mathcal{L}^2}\to 0$. Suppose that $l<k$, then $f_n\restr_{R_l}=(f_n\restr_{R_k})\restr_{R_l}$ and therefore it follows that that $\tilde{u}_k\restr_{R_l}=\tilde{u}_l$. In this way it is possible to construct a measurable function $\tilde{u}$ such that $\tilde{u}\restr_{R_k}=\tilde{u}_k$ for every $k\in\bN$. This function belongs to $\mathcal{L}^2_{loc}(\bR^{1,3},\bC^4)$ and $\|f_n\restr_{R_k}-\tilde{u}\restr_{R_k}\!\!\|\to 0$ for any $k\in\bN$. From this, it follows that $f_n$ converges to $\tilde{u}$  in the $\mathcal{L}^2$-norm  on every open bounded set $B$. By uniqueness of the limit in $\mathcal{L}^2_{loc}(\bR^{1,3},\bC^4)$ we see that $u=\tilde{u}$. In this way we have proved that $u\restr_{R_T}\in\mathcal{L}^2(R_T,\bC^4)$ . It remains to prove identity (\ref{identityintegralstrip}). Notice that we have already proved the same identity for the elements of $\Sol$ in (\ref{equazioniintegrali}). We also know that $\|f_n\restr_{R_k}-u\restr_{R_k}\!\!\|_{\mathcal{L}^2}\to 0$ for every $k\in\bN$. The proof of (\ref{identityintegralstrip}) follows then by the analogous property on $\Sol$ and the continuity of $\|\cdot\|$.
	Now, let us pass to the proof of (ii). We know from the proof of (i) that $f_n|_{R_k}\to u|_{R_k}$ for every $k\in\bN$ in $\mathcal{L}^2(R_k,\bC^4)$. Fix $k=1$, then there exists a subsequence $f_{\sigma_1(n)}$ and a null-set $N_1\subset R_1$ on whose complement $R_1\setminus N_1$ the pointwise convergence $f_{\sigma_1(n)}(x)\to u(x)$ holds. At this point, we apply the same argument to the Cauchy sequence $\{f_{\sigma_1(n)} \}_n$ and the integer $k=2$. This gives a subsequence $\{f_{\sigma_2(n)} \}_n$ and a null set $N_2\subset R_2$ such that $f_{\sigma_2(n)}(x)\to u(x)$ for every $x\in R_2\setminus N_2$. We proceed in this way for every $k\in\bN$. If we now define the subsequence $\sigma(n):=\sigma_n(n)$ and the  set $N=\bigcup_{k\in\bN} N_k$, then $N$ is again a null set and $f_{\sigma(n)}(x)\to u(x)$ on $\bR^{1,3}\setminus N$, concluding the proof.
\end{proof}

\begin{proof}[Proof of Theorem \ref{weaksol}]
	Let $u\in\scH_m$  and let $\{f_n\}_n\subset\Sol$ be any Cauchy sequence converging to $u$ in $\scH_m$. Then $\{f_n\}_n$ converges also in the $\mathcal{L}_{loc}^2$-topology and therefore, as $\supp\varphi\subset B(0,R)$ for some $R>0$, we have
	\begin{equation*}
	\begin{split}
&\int_{\bR^4} u(x)^\dagger(\cD^* \varphi) (x)\, d^4x=\int_{B(0,R)} u(x)^\dagger (\cD^* \varphi) (x)\, d^4x=\\
	&=\lim_{n\to \infty }\int_{B(0,R)} f_n(x)^\dagger (\cD^* \varphi)(x)\, d^4x=\lim_{n\to\infty}\int_{B(0,R)} (\cD f_n (x))^\dagger \varphi(x)\, d^4x=0.
	\end{split}
	\end{equation*}
	To conclude, assume that $u\in\mathcal{C}^\infty(\bR^{1,3},\bC^4)$. Integrating by parts, the proven identity can be improved to
	$$
	0=\int_{\bR^4} u(x)^\dagger(\cD^* \varphi) (x)\, d^4x=\int_{\bR^4} (\mathcal{D} u(x))^\dagger\varphi(x)\, d^4x\quad \mbox{for all }\varphi\in\mathcal{C}_0^\infty(\bR^{1,3},\bC^4).
	$$
	The fundamental theorem of variational calculus implies that $\mathcal{D} u=0$.
\end{proof}

\begin{proof}[Proof of Proposition \ref{continuityU}]
	Let us start with the group properties. By density and unitarity, it suffices to work on $\mathcal{C}_0^\infty(\bR^3,\bC^4)$. It is clear that $U_0=\bI$. Now, consider any solution $f\in\Sol$. Then  the function $f_p(u,\V{x}):=f(u+p,\V{x})$ is again an element of $\Sol$ and $f_p\restr_{\Sigma_q}=f\restr_{\Sigma_{p+q}}$. Take any $t,s\in\bR$, then we have
	$$
	U_{s+t}(f\restr_{\Sigma_0})=f\restr_{\Sigma_{s+t}}=f_s\restr_{\Sigma_{t}}=U_t(f_s\restr_{\Sigma_0})=U_t(f\restr_{\Sigma_s})=U_t(U_s(f\restr_{\Sigma_0})).
	$$ 
	Let us pass to the proof of  strong-continuity. Again, we can focus on $\mathcal{C}_0^\infty(\bR^3,\bC^4)$ (see \cite[Proposition 9.28]{moretti}). For simplicity of notation (we change notation), we write $f_t:=f\restr_{\Sigma_t}$ for every $f\in\Sol$. The proof follows from (b) of  \cite[Proposition 9.28]{moretti} if we manage to show that $(f_t|f_0)_{\mathcal{L}^2}\to (f_0|f_0)_{\mathcal{L}^2}$, with $f_t= U_t f_0$.
	To this aim, fix any $\epsilon>0$. Point (i) of Proposition \ref{causalpropag} ensures that $R_\epsilon\cap\supp f$ is compact. Therefore, by continuity of $f$, there exists $C>0$ such that
	$$
	|f_t(\V{x})^\dagger f_0(\V{x})|\le C|f_0(\V{x})|\quad\mbox{for all }t\in [-\epsilon,\epsilon],\ \V{x}\in\bR^3.
	$$
	Moreover, $f_t(\V{x})^\dagger f_0(\V{x})\to f_0(\V{x})^\dagger f_0(\V{x})$ as $t\searrow 0$ for any fixed $\V{x}\in \bR^3$, again by continuity. The claim follows by Lebesgue's dominated convergence theorem.
	
	To show the last statements let us stick again to the notation $f_t:=f\restr_{\Sigma_t}=f(t,\,\cdot\,)$. Also,  notice that $U_t(f_0)=f_t$. We want to show that 
	$$
	\frac{d U_s(f_0)}{ds}\bigg|_0= -i\mathcal{H}(f_0)\quad\mbox{for any $f\in\scH_m^{sc}$}.
	$$ 
The claim will then follow from  \cite[Theorem 1.1 ]{thaller}. The Dirac equation $\cD f=0$ can be restated as 
	$$
	(i\partial_tf) (s,\V{x})= (\mathcal{H} f_s) (\V{x})\quad\mbox{for all $f\in\Sol$ and $(s,\V{x})\in\bR^{1,3}$.}
	$$
	This implies that 
	$$
	\lim_{s\to 0}\frac{f(s,\V{x})-f(0,\V{x})}{s}+i\mathcal{H} f_0(\V{x})=0\quad\mbox{for all }\V{x}\in\bR^3.
	$$
	Now, fix $\epsilon>0$.  From Proposition \ref{causalpropag}-(i) again, there must exist some $\delta>0$ such that  
	$$
	R_\epsilon\cap\supp f\subset[-\epsilon,\epsilon]\times \overline{B(0,\delta)}:=A.
	$$ 
	Moreover,  $\supp \mathcal{H} f_0\subset\supp f_0\subset \overline{B(0,\delta)}$. Therefore, for every $(s,\V{x})\in R_\epsilon$ we have
	\begin{equation*}
	\begin{split}
	\left|\frac{f(s,\V{x})-f(0,\V{x})}{s}+i\mathcal{H} f_0(\V{x})\right|^2&=\chi_{\overline{B(0,\delta)}}(\V{x})\left|\frac{f(s,\V{x})-f(0,\V{x})}{s}+i\mathcal{H} f_0(\V{x})\right|^2\le\\
	&\le \left(\chi_{\overline{B(0,\delta)}}({\V{x}})\cdot\sup_{A}|\partial_t f| + |\mathcal{H} f_0(\V{x})|\right)^2,
	\end{split}
	\end{equation*}
	where  we made use of the mean value theorem in the time variable. At this point, the function on the right-hand side is integrable and, therefore, we can apply Lebesgue's dominated convergence theorem and show that $U_t f_0$ is differentiable at $t=0$ in $\mathcal{L}^2(\bR^3,\bC^4)$, with derivative $-i\mathcal{H} f_0$. Because the function $f$ was chosen arbitrarily, we have just proved that the self-adjoint generator of $U_t$ coincides with $\mathcal{H}$ on~$\mathcal{C}^\infty_c(\bR^3,\bC^4)$. 
\end{proof}

\begin{lemma}\label{chiusuraschwartz}
	Let $\omega(\V{k}):=\sqrt{\V{k}^2+m^2}$. Then the following statements hold.
	\vspace{0.1cm}
	\begin{itemize}[leftmargin=2.5em]
		\item[\rm{(i)}] $\omega^{-1}$ has bounded derivatives of any order $k\ge 0$,\\[-0.2cm]
		\item[\rm{(ii)}] $\omega$ has bounded derivatives of any order $k\ge 1$ and $\omega(\V{k})\le \sum_{i=1}^3|k_i|+m$.
	\end{itemize}
\end{lemma}
\begin{proof}
	The proof follows by direct inspection.
\end{proof}

	\begin{proof}[Proof of Theorem \ref{theoremfourierH}]
		The proof of the first part follows from the properties of the Fourier Transform (see \cite{thaller}, Section 1.4.4). The last statement follows directly from Stone Theorem and the uniqueness of the self-adjoint generator.
	\end{proof}
\begin{proof}[Proof of Theorem \ref{proporthogonalityp}]
	The result can be proved by direct inspection, bearing in mind the properties of the Dirac matrices: $(\gamma^\mu)^\dagger= \eta^{\nu\nu}\gamma^\mu$ and $\{\gamma^\mu,\gamma^\nu\}=2\eta^{\mu\nu}\bI_4$.
\end{proof}
\begin{proof}[Proof of Proposition \ref{propPpm}]
	The proof follows from the corresponding features of the matrices $p_\pm(\V{k})$. Just notice that the operators $\hat{P}_\pm$ are well-defined by virtue of the boundedness of $p_\pm$.
\end{proof}
\begin{proof}[Proof of Proposition \ref{lemmaproizionieschwat}]
	The proof follows directly from the density of the Schwartz space in the space of square-integrable functions and the fact that $\hat{P}_\pm(\mathcal{S}_p(\bR^3,\bC^4))\subset \mathcal{S}_p(\bR^3,\bC^4)$. This last inclusion follows from the second identity in (\ref{localprojectorposneg}), together with Lemma~\ref{chiusuraschwartz} and the fact that the Schwartz space is closed under multiplication by polynomials.
\end{proof}
\begin{proof}[Proof of Proposition \ref{theoremenergy}]
	The proof of point (i) follows immediately from $$\hat{P}_\pm(\varphi)\in\mathcal{S}_p(\bR^3,\bC^4)\subset D(\hat{\mathcal{H}})$$ and the definition of $\hat{\mathcal{H}}$.  Similarly, point (ii) can be proved exploiting the explicit action of $e^{-it\hat{\mathcal{H}}}$. Point (iii) can be found in \cite[Theorem 1.1]{thaller}.
\end{proof}

\begin{proof}[Proof of Proposition \ref{generalsolutionsmooth}]
	Consider a  generic function $\varphi\in \mathcal{S}_p(\bR^3,\bC^4)$ and denote by $\tilde{u}_\varphi$ the function defined in the right-hand side of (\ref{fourierexpansion}). We claim that $\tilde{u}_\varphi\in\ker\cD$. In order to prove this, it suffices to focus on functions $\varphi\in \hat{P}_\pm(\mathcal{S}_p(\bR^3,\bC^4))$.  We can just consider the negative-energy case, because the other one is analogous. Using the compact notation $\eta(x,k)=-\omega(\V{k})t-\V{k}\cdot\V{x}$, define
	$$
	{\varphi}_x:= \varphi\, e^{-i\eta(k,x)}\in\mathcal{L}_p^1(\bR^3,\bC^4)\quad\mbox{for any } x\in\bR^{1,3}.
	$$
	Notice that $x\mapsto \varphi_x(\V{k})$ is smooth for every fixed $\V{k}\in\bR^3$. Moreover, differentiating in $x$ gives, for any multi-index $\alpha \in\bN^4$,
	$$
	|D^\alpha \varphi_x(\V{k})| = |k^\alpha \varphi (\V{k})|\quad\mbox{for every }\V{k}\in\bR^3.
	$$
	Notice that $k^\alpha\varphi$ is an element of $\mathcal{S}_p(\bR^3,\bC^4)$ for any multi-index $\alpha$. This follows from the fact that the Schwartz space is closed under multiplication by polynomials and by Lemma \ref{chiusuraschwartz}.
	At this point, as $k^\alpha \varphi \in\mathcal{L}_p^1(\bR^3,\bC^4)$ and it does not depend on $x$, we can apply \cite[Theorem 1.88]{moretti} and conclude that the function
	\begin{equation}\label{equazionesmooth}
	\bR^{1,3}\ni x\mapsto  \int_{\bR^3} \frac{d^3\V{k}}{(2\pi)^{3/2}}\varphi(\V{k})e^{-i\eta(k,x)}\in\bC^4
	\end{equation}
	is differentiable to every order and that partial derivatives and integral can be switched.  In particular, if we apply the Dirac operator $\cD$ to (\ref{equazionesmooth}), it is not difficult to see that (\ref{equazionesmooth}) solves the Dirac equation. 
	
	In the special case of functions $\varphi\in\cF(\mathcal{C}_{0,x}^\infty(\bR^3,\bC^4))$, the function $\tilde{u}_\varphi$ must coincide with $u_\varphi=\mathrm{E}\circ\cF^{-1}(\varphi)\in\Sol$, for they are both smooth solutions of (\ref{solutiondirac}) and coincide at $t=0$.
	Let us go back to general functions $\varphi\in\mathcal{S}_p(\bR^3,\bC^4)$. It holds that  
	$$
	\tilde{u}_\varphi(t,\cdot)\in\mathcal{S}_x(\bR^3,\bC^4)\subset \mathcal{L}_x^2(\bR^3,\bC^4)\ \ \mbox{and}\ \  \|\tilde{u}_\varphi(t,\cdot)\|_{\mathcal{L}^2}=\|\tilde{u}_\varphi(0,\cdot)\|_{\mathcal{L}^2}\quad \mbox{for every }t\in\bR.
	$$ 
	The first statement follows from the fact that $\tilde{u}(t,\cdot)$ is, by definition, the inverse Fourier Transform of $\varphi_+\, e^{- i\omega t} +\varphi_-\, e^{ i\omega t}$, which is a Schwartz function (again, use Lemma \ref{chiusuraschwartz}). 
	In order to prove the second identity, notice once more that $\tilde{u}(t,\cdot)$ is the inverse Fourier Transform of $\varphi_+\, e^{- i\omega t} + \varphi_-\, e^{ i\omega t}$ and, therefore, using Parseval's identity (see \cite[Proposition 3.105]{moretti}) and the fact that the subspaces $W_{\V{k}}^\pm$ are orthogonal to each other (see Proposition \ref{proporthogonalityp}) we have:
	\begin{equation*}
	\begin{split}
	\|\tilde{u}_\varphi(t,\cdot)\|_{\mathcal{L}^2}^2&=\int_{\bR^3} |\varphi_+(\V{k})\, e^{- i\omega(\V{k})t} + \varphi_-(\V{k})\, e^{ i\omega(\V{k}) t}|^2\,d^3\V{k}=\\
	&=\int_{\bR^3} \big(|\varphi_+(\V{k})|^2 + |\varphi_-(\V{k})|^2\big)\,d^3\V{k}=\\
	&=\int_{\bR^3} |\varphi_+(\V{k}) + \varphi_-(\V{k})|^2\,d^3\V{k}=\left(\|\varphi\|^2_{\mathcal{L}^2}=\right)\|\tilde{u}_\varphi(0,\cdot)\|_{\mathcal{L}^2}^2
	\end{split}
	\end{equation*}
	At this point, reasoning as in (\ref{equazioniintegrali}) and making use of the identity between brackets above, it follows that, for every $T>0$,
	$$
	\tilde{u}_\varphi\restr_{R_T}\in\mathcal{L}^2(R_T,\bC^4)\ \ \mbox{and}\ \  \|\tilde{u}_\varphi\restr_{R_T}\|_{\mathcal{L}^2}=\sqrt{2T}\|\tilde{u}_\varphi(0,\cdot)\|_{\mathcal{L}^2}=\sqrt{2T}\|\varphi\|_{\mathcal{L}^2}.
	$$ 
	Since $\cF(\mathcal{C}_{0,x}^\infty(\bR^3,\bC^4))$ is dense in $\mathcal{S}_p(\bR^3,\bC^4)$ in the $\mathcal{L}^2$-norm, there must exist a sequence 
	$$
	\{\varphi_n \}_n\subset\cF(\mathcal{C}_{0,x}^\infty(\bR^3,\bC^4))\quad\mbox{such that}\quad \|\varphi_n- \varphi\|_{\mathcal{L}^2}\to 0.
	$$
	The sequence $\{u_{\varphi_n} \}_n\subset\Sol$ is then of Cauchy type, as $\{\varphi_n\}_n$ is Cauchy and $\mathrm{{E}_0\circ\cF^{-1}}$ is an isometry.
	Now, take an open bounded set $B\subset\bR^{1,3}$ and $T>0$ such that $\overline{B}\subset R_T$. Then, bearing in mind that $\tilde{u}_{\varphi_n}=u_{\varphi_n}$, we have
\begin{equation*}
\begin{split}
	\|\tilde{u}_{\varphi}\restr_{B}-u_{\varphi_n}\restr_{B}\|_{\mathcal{L}^2}&\le \|\tilde{u}_{\varphi}\restr_{R_T}-u_{\varphi_n}\restr_{R_T}\|_{\mathcal{L}^2}=\|\tilde{u}_{\varphi-\varphi_n}\restr_{R_T}\|_{\mathcal{L}^2}=\\
	&=\sqrt{2T}\|\varphi-\varphi_n\|_{\mathcal{L}^2}\to 0,
	\end{split}
	\end{equation*}
	which implies that $u_{\varphi_n}\to \tilde{u}_\varphi$ in $\mathcal{L}^2_{loc}(\bR^{1,3},\bC^3)$. 
	By definition, this means that $\tilde{u}_{\varphi}\in\sol$ (notice that $\tilde{u}_\varphi$ clearly is locally square-integrable, it being continuous) and 
	$$
	\|\tilde{u}_\varphi\|_m=\lim_{n\to \infty }\|u_{\varphi_n}\|_m=\lim_{n\to \infty }\|\varphi_n\|_{\mathcal{L}^2}=\|\varphi\|_{\mathcal{L}^2}.
	$$
	In particular, the last identity proves that the linear function $f:\mathcal{S}_p(\bR^3,\bC^4)\ni\varphi\mapsto \tilde{u}_\varphi\in\sol$ is continuous. 
	Using the notation $g:=\mathrm{E}\circ\cF^{-1}:	\mathcal{L}_p^2(\bR^3,\bC^4)\ni\psi\mapsto u_\psi\in\sol$, and exploiting the uniqueness of the extension of continuous linear operators (see \cite[Proposition 2.47]{moretti}), we get
	$$
	g\restr_{\cF(\mathcal{C}_{0,x}^\infty(\bR^3,\bC^4))}=f\restr_{\cF(\mathcal{C}_{0,x}^\infty(\bR^3,\bC^4))}\ \Longrightarrow\ g\restr_{\mathcal{S}_p(\bR^3,\bC^4)}= f\restr_{\mathcal{S}_p(\bR^3,\bC^4)}.
	$$
	In other words, we have just proved that $\tilde{u}_\varphi=u_\varphi$ for every $\varphi\in\mathcal{S}_p(\bR^3,\bC^4)$.
\end{proof}

\begin{proof}[Proof of Proposition \ref{mostgeneralsolution}]
	Identity (\ref{representationsol}) can be obtained as in (\ref{abstractrestatement}). Concerning the identity in (\ref{identificationsolutions}), the inclusion $\supset$ was already proved in (\ref{abstractrestatement}): just take $f:=\cF^{-1}(\tilde{\varphi})$. In order to prove the inclusion $\subset$, take any element $P_\pm(\,\cdot\,,f)$  and define the functions
	$$
	\varphi_\pm(\V{k}):=\pm (2\pi)^{-1/2}\,  p_\pm(\V{k}) \gamma^0\, \cF(f)(\pm\omega(\V{k}),\V{k}),
	$$
	which belong to $\hat{P}_\pm(\mathcal{S}_p(\bR^3,\bC^4))$. From (\ref{representationsol}) and Proposition \ref{generalsolutionsmooth} we get
	\begin{equation}\label{key}
	P_\pm((t,\V{x}),f)=\int_{\bR^3}\frac{d^3\V{k}}{(2\pi)^{3/2}}\,\varphi_\pm(\V{k})\, e^{-i(\pm\omega(\V{k})t-\V{k}\cdot\V{x})}=u_{\varphi_\pm}(t,\V{x})=\hat{\mathrm{E}}(\varphi_\pm)(t,\V{x}).
	\end{equation}
	Proposition \ref{generalsolutionsmooth} ensures that $P_\pm(\,\cdot\,,f)$ belong to the subspace  $\scH_m^\pm\cap\mathcal{C}^\infty(\bR^{1,3},\bC^4)$.
\end{proof}

\begin{proof}[Proof of Proposition \ref{propdistribproof}]
	Let $x\in\bR^{1,3}$ and $f\in\mathcal{S}_x(\bR^{1,3},\bC^4)$. By definition of Schwartz space, there exists a constant $C$ such that:
	$$
	(1+|k|^4) \:|\cF(f)(k)|\le C\:\|\cF(g)\|_{4,0}\quad\mbox{for all $k\in\bR^{1,3}$}
	$$
	(where $\|\cdot\|_{p,q}$ denote the Schwartz norms).  Usingthis inequality and identities (\ref{localprojectorposneg}), (\ref{representationsol}), we see that
	\begin{equation*}
	\begin{split}
		|P_\pm(x,f)| &\le \left(C\int_{\bR^3}\frac{d^3\V{k}}{(2\pi)^2}\frac{\omega(\V{k})\|\gamma^0\|_2+\sum_{i=1}^3|k^i|\|\gamma^i\|_2+m}{2\omega(\V{k})(1+(\omega(\V{k})^2+\V{k}^2)^2)} \right)\|\cF(f)\|_{4,0} \:.
	\end{split}
\end{equation*}
	The integral in brackets is well-defined and convergent.
	Now, exploiting \cite[Lemma 8.2.2 and eq.~(8.2.2)]{friedlander2} (where a different convention for the Schwartz norm is adopted), one finds that, for some constant $K>0$,
	$$
	\|\cF(h)\|_{4,0}\le K \:\|h\|_{6,4}\quad\mbox{for all }h\in \mathcal{S}_x(\bR^{1,3},\bC^4).
	$$
	Putting all together, we find a constant $B>0$ such that $|P_\pm(x,f)|\le B\|f\|_{6,4}$ for all $f\in\mathcal{S}_x(\bR^{1,3},\bC^4)$. This proves that $P_\pm$ are indeed tempered distributions.
\end{proof}

\begin{proof}[Proof of Proposition \ref{realizationallsolut}]
	The first statement follows directly from the second one, which can be proved using (\ref{abstractrestatement}) and (\ref{identificationsolutions}).
\end{proof}
\begin{proof}[Proof of Proposition \ref{kernelP}]
	The case of $P_-$ can be proved following the discussion in \cite[Section 1.2.5]{FF} (see Lemma 1.2.9). In particular, the function $\mathcal{P}_-$ is given by
	$
	\mathcal{P}_-=(i\slashed{\partial}+m) T_{m^2}
	$
	with $T_{m^2}$ the function defined in (1.2.29). The case of positive energy is analogous.
\end{proof}
\begin{proof}[Proof of Proposition \ref{spindecomposition}]
	Take any $\lambda_{\uparrow,\downarrow}^\pm\in\mathcal{S}_p(\bR^3,\bC)$. Then, reasoning as in the proof of Lemma  \ref{lemmaproizionieschwat}, it is possible to show that the functions $\lambda_{\uparrow,\downarrow}\,\chi_{\uparrow,\downarrow}^\pm$ belong to $\mathcal{S}_p(\bR^3,\bC^4)$, more precisely to $\hat{P}_\pm(\mathcal{S}_p(\bR^3,\bC^4))$. Together with Proposition \ref{generalsolutionsmooth}, this proves that the functions in (\ref{basicsolution}) do define elements of $\hat{\mathrm{E}}(\hat{P}_\pm(\mathcal{S}_p(\bR^3,\bC^4)))$.  In order to prove the other inclusion, let us stick to the positive-energy case (the negative one is analogous) and take any $u_\varphi\in \hat{\mathrm{E}}(\hat{P}_+(\mathcal{S}_p(\bR^3,\bC^4)))$ with $\varphi\in \hat{P}_+(\mathcal{S}_p(\bR^3,\bC^4))\subset\mathcal{S}_p(\bR^3,\bC^4)$. Since the vectors $\chi_{\uparrow,\downarrow}^+(\V{k})$ define a basis of $W_\V{k}^+$ at every $\V{k}\in\bR^3$, we can write 
	$$
	\varphi(\V{k})=\lambda_\uparrow^+(\V{k})\chi_\uparrow^+(\V{k})+\lambda_\downarrow^+(\V{k})\chi_\downarrow^+(\V{k})=
	\big(
	\lambda_\uparrow^+(\V{k}),
	\lambda_\downarrow^+(\V{k}),
	A(\V{k})
	\big)^t,
	$$
	for some vector $A(\V{k})\in\bC^2$. By definition of $\mathcal{S}(\bR^3,\bC^4)$, all the components of $\varphi$ must be (complex-valued) Schwartz functions and, therefore, in particular, $\lambda_{\uparrow,\downarrow}^+\in\mathcal{S}_p(\bR^3,\bC)$. 
\end{proof}

\begin{lemma}\label{estimate}
	Referring to Proposition \ref{spindecomposition} the following estimates hold.
	\vspace{0.15cm}
	\begin{itemize}[leftmargin=2.5em]
		\item[\rm{(i)}] $\|u_{\uparrow,\downarrow}^\pm\|_m\le \sqrt{2}\,\|\lambda_{\uparrow,\downarrow}^\pm\|_{\mathcal{L}^2}$\\[-0.15cm]
		\item[\rm{(ii)}] $|\partial_\mu u_{\uparrow,\downarrow}^\pm(x)|\le \sqrt{2}\,(2\pi)^{-3/2}\,\|k_\mu \lambda_{\uparrow,\downarrow}^\pm\|_{\mathcal{L}^1}$
	\end{itemize}
\end{lemma}
\begin{proof}
	For the proof of (i) notice that $\chi^\dagger \chi=2\omega/(\omega+m)\le 2$ and, therefore, Plancherel Theorem and a calculation similar as in the proof of Proposition \ref{generalsolutionsmooth} yields
	\begin{equation*}
	\begin{split}
	\|u^\pm\|_m^2\!&=\!\|u^\pm(0,\cdot)\|_{\mathcal{L}^2}^2\!=\!\int_{\bR^3}d^3\V{k}\,|\lambda^\pm(\V{k})|^2|\chi^\pm(\V{k})|^2\le 2\int_{\bR^3}d^3\V{k}\,|\lambda^\pm(\V{k})|^2={2}\|\lambda^\pm\|_{\mathcal{L}^2}^2.
	\end{split}
	\end{equation*}
	Concerning the first derivatives, it follows from (\ref{basicsolution}) that
	\begin{equation}
	|\partial_\mu u^\pm(x)|\le \frac{1}{(2\pi)^{3/2}}\int_{\bR^3}|\lambda^\pm(\V{k})||\chi^\pm(\V{k})||k_\mu|\, d^3\V{k}\le \frac{\sqrt{2}}{(2\pi)^{3/2}}\|k_\mu \lambda^\pm \|_{\mathcal{L}^1}.
	\end{equation}
	where $k=(\pm\omega(\V{k}),\V{k})$.
\end{proof}

\begin{proof}[Proof of Lemma \ref{lemmamollifieroperator}]
	Let us start with point (i). Take any $h\in\mathscr{M}(\bR^{1,3})$ with $\supp h\subset B(0,\delta)$ and let $f\in\Sol$ be arbitrary. The function $f*h$ is smooth (see \cite[Theorem 1.6.1]{ziemer}). From Proposition \ref{causalpropag} we know that $\supp f$ is contained in the causal propagation of the support of the initial data. By definition of convolution, it holds that $\supp (f*h)\subset \supp f + \overline{B(0,\delta)}$ and therefore $\supp ( f*h|_{\Sigma_0})$ must be compact.  If we manage to prove that $f*h$ belongs to $\ker\cD$, then the proof of (i) is finished. This is true. Indeed, applying again \cite[Theorem 1.6.1]{ziemer}, we get
	\begin{equation*}
	\begin{split}
	&i\gamma^\mu\partial_\mu (h*f)=\partial_{\mu}(h*(i\gamma^\mu f))=(\partial_{\mu}h)*(i\gamma^\mu f)=\\
	&=\int_{\bR^4}\left(\frac{\partial h}{\partial s^\mu} \right)(x-y)\cdot i\gamma^\mu f(y)\, d^4y= -\int_{\bR^4}\left(\frac{\partial }{\partial y^\mu}h(x-y)\right)\cdot i\gamma^\mu f(y)\, d^4y \stackrel{(*)}{=}\\
	&= \int_{\bR^4}h(x-y)\,\left(i\gamma^\mu\frac{\partial }{\partial y^\mu} f(y)\right)\, d^4y=m (h*f),
	\end{split}
	\end{equation*}
	where in  $(*)$ we used the divergence theorem and the fact that $h$ is compactly supported.
	%
	
	Now, let us prove point (ii). First, notice that the convolution of any $\mathcal{L}^2_{loc}$ function with a mollifier $h$ always yields a smooth function (see \cite[Theorem 1.6.1]{ziemer}). In order to proceed, we need a technical result about mollification.
	\begin{lemma}\label{lemmatecnico}
		Let $v\in\mathcal{L}_{loc}^1(\bR^n,\bC^m)$, $B\subset\bR^n$ be an open set and define the open set
		$
		B_\delta:=B\cup\bigcup_{x\in\partial B} B(x,\delta).
		$
		Then $\|h*v\restr_B\|_{\mathcal{L}^2}\le \|v\restr_{B_\delta}\|_{\mathcal{L}^2}.$
	\end{lemma}
	\begin{proof}
		Exploiting H\"older's inequality, $\|h\|_{\mathcal{L}^1}=1$ and Fubini Theorem for positive functions, we have:
		\begin{equation*}
		\begin{split}
		&\int_B|v*h(x)|^2\, d^nx=\int_B\, d^nx\left|\int_{\bR^n}\sqrt{h(x-y)}\sqrt{h(x-y)}v(y)\,d^ny\right|^2\le \\
		&\le \int_B d^nx \int_{\bR^n} h(x-y)|v(y)|^2\, d^ny= \int_{\bR^n} |v(y)|^2\, d^ny\int_B h(x-y)\, d^nx\stackrel{(*)}{=}\\
		&=\int_{B_\delta} |v(y)|^2\, d^ny\int_B h(x-y)\, d^nx\le \int_{B_\delta} |v(y)|^2\, d^ny.
		\end{split}
		\end{equation*}
		The only non-obvious step in the above chain of inequalities is identity (*). In order to understand it, take any $y\in \bR^n\setminus B_\delta$. By definition, $|y-x|>\delta$ for every $x\in B$ and therefore $h(x-y)=0$ for every $x\in B$. This shows that $y$ gives no contribution to the integral on the left-hand side of~$(*)$.
	\end{proof}
	We can go back to the proof of point (ii). Let $u\in\sol$ and $\{f_n \}_n\subset\Sol$ be a Cauchy sequence converging to $u$ in $\mathcal{L}_{loc}^2(\bR^{1,3},\bC^4)$, as in the definition of $\sol$.
	First, exploiting Lemma \ref{lemmaconverging2}, Lemma \ref{lemmatecnico} and the fact that $f_n*h\in\Sol$, we get for any $T>0$ that
	\begin{equation*}
	\begin{split}
	\sqrt{2T}\|f_n*h-f_m*h\|_m&=\|f_n*h\restr_{R_T}-f_m*h\restr_{R_T}\|_{\mathcal{L}^2}=\|(f_n-f_m)*h\restr_{R_T}\|_{\mathcal{L}^2}\le \\
	&\le\|(f_n-f_m)\restr_{R_{T+\delta}}\|_{\mathcal{L}^2}=\sqrt{2{(T+\delta)} }\|f_n-f_m\|_m
	\end{split}
	\end{equation*}
	which shows that $\{f_n*h\}_n$ is a Cauchy sequence in $\Sol$.\textbf{}
	Second, exploiting Lemma \ref{lemmatecnico},  we see that, for any bounded open set $B\subset\bR^{1,3}$,
	$$
	\|f_n*h\restr_B-u*h\restr_B\|_{\mathcal{L}^2}\le \|f_n\restr_{B_\delta}-u\restr_{B_\delta}\|_{\mathcal{L}^2}\to 0,
	$$
	where we used the fact that $B_\delta$ is open and bounded. 
	By definition of $\sol$, this implies that $u*h\in\sol$.
	
	To prove point (iii), fix any $T>0$ and notice that (see Lemma \ref{lemmaconverging2})
	$$\sqrt{2T}\|u*h\|_m=\|u*h\restr_{R_T}\|_{\mathcal{L}^2}\le \|u\restr_{R_{T+\delta}}\|_{\mathcal{L}^2}=\sqrt{2(T+\delta)}\|u\|_m,$$
	for every $u\in\sol$. This gives the continuity.
\end{proof}

\begin{proof}[Proof of Proposition \ref{representationP}]
	The theorem follows once we have proved that ($\gamma^0$ is unitary)
	$$
	\pm 2\pi\,a^\dagger P_\pm(x,y)\gamma^0 b=\sum_{n=1}^\infty a^\dagger\, \gR u_n(x)\,\Sl\gR u_n(y)\,|\,	\gamma^0 b\Sr\quad \mbox{for any $a,b\in\bC^4$}.
	$$ 
	By definition, we have $u_n=\hat{\mathrm{E}}(\varphi_n)$ for some Hilbert basis $\{\varphi_n\}_n$ of $\hat{P}_\pm(\mathcal{L}_p^2(\bR^3,\bC^4))$ which is contained in $\hat{P}_\pm(\mathcal{S}_p(\bR^3,\bC^4))$. Let $N\in\bN$, then, exploiting Proposition \ref{generalsolutionsmooth} and the definition of $\gR_{\varepsilon}$, we have
	$$
	u_n(x):=\int_{\bR^3}\frac{d^3\V{k}}{(2\pi)^{3/2}}\varphi_n(\V{k})e^{-i\eta(x,k)},\quad \gR_\varepsilon u_n(x)=\int_{\bR^3}\frac{d^3\V{k}}{(2\pi)^{3/2}}\mathfrak{g}_\varepsilon(\V{k})\varphi_n(\V{k})e^{-i\eta(x,k)},
	$$	
	where we used the compact notation $\eta(k,x)=\pm\omega(\V{k})t-\V{k}\cdot\V{x}$. Now, since $\varphi_n(\V{k})=p_\pm(\V{k})\varphi_n(\V{k})$, we get the following chain of identities (with obvious notation):
	\begin{equation*}
	\begin{split}
	&\sum_{n=1}^Na^\dagger\, \gR u_n(x)\,\Sl\gR u_n(y)|\,\gamma^0b\Sr=\sum_{n=1}^N\int_{\bR^3}\frac{d^3\V{k}}{(2\pi)^{3/2}}\left(\mathfrak{g}_\varepsilon(\V{k})p_\pm(\V{k})a\, e^{i\eta(x,k)}\right)^\dagger\! \varphi_n(\V{k})\,\cdot\\
	&\qquad\qquad\qquad\qquad\qquad\qquad\cdot\int_{\bR^3}\frac{d^3\V{p}}{(2\pi)^{3/2}}\varphi_n(\V{p})^\dagger \left(\mathfrak{g}_\varepsilon(\V{p})p_\pm(\V{p})\gamma^0 b\, e^{i\eta(y,p)}\right)\!=\\
	&=\frac{1}{(2\pi)^3}\sum_{n=1}^N \left(\mathfrak{g}_\varepsilon\, p_\pm\,a\, e^{i\eta(x,\,\cdot\,)}\,\bigg|\,\varphi_{n}\right)_{\mathcal{L}^2}\left(\varphi_n\,\bigg|\,\mathfrak{g}_\varepsilon\, p_\pm\, \gamma^0\, b\, e^{i\eta(y,\,\cdot\,)}\right)_{\mathcal{L}^2}.
	\end{split}
	\end{equation*}
	At this point, noticing that the functions $\varphi_n$ are form a Hilbert basis, the completeness relations gives
	\begin{equation}
	\begin{split}
	&\sum_{n=1}^\infty a^\dagger\, \gR u_n(x)\,\Sl\gR u_n(y)|\,b\Sr=\frac{1}{(2\pi)^3} \left(\mathfrak{g}_\varepsilon\, p\,a\, e^{i\eta(\,\cdot\,,x)}\bigg|\mathfrak{g}_\varepsilon\, p\, \gamma^0\, b\, e^{i\eta(y,\,\cdot\,)}\right)_{\mathcal{L}^2}=\\
	&=\pm a^\dagger \left(\int_{\bR^3}\frac{d^3\V{k}}{(2\pi)^3}\, \mathfrak{g}_\varepsilon(\V{k})^2\, p_\pm(\V{k})\,(\pm \gamma^0)\, e^{-i\eta(x-y,k)}\bigg|_{k^0=-\omega(\V{k})}\right)\, b=\\
	&= \pm 2\pi\,a^\dagger P_{\pm}^{2\varepsilon}(x,y)\,b,
	\end{split}
	\end{equation}
	where the last equality follows by reasoning as in (\ref{abstractrestatement}).
\end{proof}

\begin{proof}[Proof of Lemma \ref{lemmalinearlyinde}]
	Let $\lambda^1,\dots,\lambda^n\in\bR$ be such that $\sum_{i=1}^n\lambda^i v_i=0$.  Then
	\begin{equation*}
		\begin{split}
		\sum_{j=1}^n|\lambda^j| =\sum_{j=1}^n\left|\sum_{i=1}^n\lambda^i\langle e_j|v_i-e_i\rangle\right|\le n\sum_{i=1}^n|\lambda^i|\|v_i-e_i\|\le n\varepsilon\sum_{i=1}^n|\lambda^i|.
		\end{split}
	\end{equation*}
If any of the coefficients $\lambda^j$ is different from zero we can divide by $\sum_{i}|\lambda^i|$ and get $1\le \varepsilon n<1$, which is not possible. Therefore, $\lambda^i=0$ for all $i=1,\dots,n$ and the claim is proved.
\end{proof}

\begin{proof}[Proof of Proposition \ref{part1}]
	Let $\{u_1,\dots,u_n\}$ be an orthonormal basis of $\sU\subset\sol^-$ and let $\epsilon >0$. Exploiting Lemma \ref{lemmaproizionieschwat}, for any $\epsilon'>0$ there must exist unit vectors $\{\phi_1,\dots,\phi_n\}\subset \hat{\mathrm{E}}(\hat{P}_-( \mathcal{S}_p(\bR^3,\bC^4))$ such that $\|u_i-\phi_i\|<\epsilon'/2<\epsilon'$ for any $i=1,\dots,n$. 
	Choosing $\epsilon'$ sufficiently small, we infer that the functions $\phi_i$ are linearly independent, thanks to Lemma \ref{lemmalinearlyinde}.
	As a consequence, for any choice of $i\neq j$ we also have
	\begin{equation}\label{bound}
	\begin{split}
	|\langle\phi_i|\phi_j\rangle|&=|\langle\phi_i|\phi_j\rangle-\langle u_i|u_j\rangle|=|\langle\phi_i|\phi_j\rangle-\langle\phi_i|u_j\rangle+\langle\phi_i|u_j\rangle-\langle u_i|u_j\rangle|=\\
	&\le \|\phi_i\|\|\phi_j-u_j\|+\|\phi_i-u_i\|\|u_j\|<\epsilon'.
	\end{split}
	\end{equation}
	At this point, we can apply the Gram-Schmidt algorithm and get an orthogonal set $\{\varphi_1,\dots,\varphi_n\}$ which spans the same linear space of $\{\phi_1,\dots,\phi_n\}$. 
	These orthogonal vectors can be defined in a non-recursive way as follows (apply for example \cite[Section 56.3]{vecchio}). For any $i=1,\dots,n$:
	\begin{equation}
	\varphi_i = \dfrac{1}{\det G_{i,\hat{i}}}\,\sum_{k=1}^i(-1)^{i+k}\big(\!\det G_{i,\hat{k}}\,\big)\, \phi_k,\quad
	G_i=\begin{bmatrix}
	(\phi_1|\phi_1) & \cdots & (\phi_i|\phi_1)\\
	\vdots & \ddots & \vdots\\
	(\phi_1|\phi_{i-1})& \cdots & (\phi_{i}|\phi_{i-1})
	\end{bmatrix}\:,
	\end{equation}
	where $G_{i,\hat{k}}$ is the square matrix obtained by removing the $k$-th column from $G_i$ and $\det G_{1,\hat{1}}:=1$. 
	Now, suppose that $ n!\epsilon'<1$. Exploiting Leibniz' formula for the determinant, together with (\ref{bound}) and $\|\phi_k\|=1$, it is not difficult to see that    (notice that $(\epsilon')^m<\epsilon'$):
	$$
	|\det G_{i,\hat{i}}|\ge 1-n!\epsilon',\quad 	|\det G_{i,\hat{k}}|\le n!\epsilon'\ \mbox{ for all } k\neq i.
	$$
	Exploiting these inequalities, we conclude the proof of point (i).
	Indeed, let $i=1,\dots,n$, then
	\begin{equation}
	\begin{split}
	\|\varphi_i-u_i\|\le \|\varphi_i-\phi_i\|+\|\phi_i-u_i\|\le \epsilon' +\sum_{k=1}^{i-1}\left|\frac{\det G_{i,\hat{k}}}{\det G_{i,\hat{i}}}\right|\le \epsilon' + (i-1)\frac{n!\epsilon'}{1-n!\epsilon'}.
	\end{split}
	\end{equation}
	Fix now any $\epsilon''>0$. Then we can always choose $\epsilon'$ small enough so that $\|\varphi_i-u_i\|<\epsilon''$ for any $i=1,\dots,n$. To conclude, define the orthonormal set $\psi_i=\varphi_i/\|\varphi_i\|$ for $i=1,\dots,n$. At this point, choosing $\epsilon''$ sufficiently small, it can be arranged that $\|u_i-\psi_i\|<\epsilon$. The second statement of point (i) can be proved easily exploiting the first one.

	Now, let us prove point (ii). Fix any $0<\epsilon<n^{-1/2}$ as in the assumptions, and take a corresponding set $\{\psi_1,\dots,\psi_n \}$ as in the proof of point (i). Then, we want to show that there is no non-vanishing $n$-ple $\{\lambda_1,\dots,\lambda_n\}\subset\bC$ such that $\sum_{i=1}\lambda_i\psi_i\perp\sU$ (in this sense the space spanned by the $\psi_i$ is "close enough" to $\sU$). Suppose by contradiction that this is not true.  Then we have
	\begin{equation*}
	\left\|\sum_{i=1}^n\lambda_i\psi_i-\sum_{i=1}^n\lambda_iu_i\right\|^2=\left\|\sum_{i=1}^n\lambda_i \psi_i\right\|^2+\left\|\sum_{i=1}^n\lambda_i u_i\right\|^2=\underbrace{\left\|\sum_{i=1}^n\lambda_i \psi_i\right\|^2}_{A>0}+\underbrace{\sum_{i=1}^n|\lambda_i|^2}_{B>0},
	\end{equation*}
	where $A$ must be strictly positive, because we assumed that at least some of the scalars $\lambda_i$ do no vanish and that the vectors $\psi_i$ are linearly independent.
	At the same time, using H\"older's inequality, we have
	\begin{equation*}
	\begin{split}
	\left\|\sum_{i=1}^n\lambda_i\psi_i-\sum_{i=1}^n\lambda_iu_i\right\|^2&=\left\|\sum_{i=1}^n\lambda_i (\psi_i-u_i)\right\|^2\le\left(\sum_{i=1}^n|\lambda_i|\|\psi_i-u_i\|\right)^2\le\\
	&\le \epsilon^2 \left(\sum_{i=1}^n|\lambda_i|\right)^2\le\epsilon^2 \left(\sum_{i=1}^n|\lambda_i|^2\right) n=n\epsilon^2 B.
	\end{split}
	\end{equation*}
	Putting all together we have
	$
	A+B\le n\epsilon^2 B
	$
	from which $1\le1+\frac{A}{B}\le n\epsilon^2<1$, which is a contradiction.
	This proves that $\sum_{i=1}^n\lambda_i\psi_i\perp\sU$ implies $\lambda_i=0$ for all $i=1,\dots,n$.
\end{proof}

\begin{proof}[Proof of Proposition \ref{part2}]
	Let $\{\psi_1,\dots,\psi_n\}$ be an approximating set for an orthonormal basis $\{u_1,\dots,u_n \}$ of $\sU$ as in Proposition \ref{part1}  and take any $\varphi\in \hat{\mathrm{E}}(\hat{P}_-(\mathcal{S}_p(\bR^3,\bC^4)))$. We want to find scalars $\{\lambda_1,\dots,\lambda_n\}\subset\bC$ such that the linear combination
	\begin{equation}\label{Psi}
	\Psi[\varphi]:=\varphi-\sum_{i=1}^n\lambda_i \psi_i\in\hat{\mathrm{E}}(\hat{P}_-(\mathcal{S}_p(\bR^3,\bC^4)))\subset\sol^- 
	\end{equation}
	is orthogonal to the subspace $\sU$, and show that these are uniquely determined by the function $\varphi$. This is equivalent to the requirement
	\begin{equation}\label{equazionelinear}
	0=\langle u_j|\Psi[\varphi]\rangle=\langle u_j|\varphi\rangle-\sum_{i=1}^n\lambda_i \langle u_j|\psi_i\rangle\quad \forall j=1,\dots,n.
	\end{equation}
	Suppose first that $\varphi\perp\sU$, then it must be $\sum_{i=1}^n\lambda_i \psi_i\perp\sU$, which is only possible if $\lambda_i=0$ for any $i=1,\dots,n$, as follows from (ii) in Proposition \ref{part1}.  On the other hand, it is clear that $\lambda_i=0$ for any $i=1,\dots,n$ implies $\varphi\perp\sU$. This proves the last statement of the proposition.
	
	Now, suppose on the contrary that $\varphi\not\perp\sU$. Then there must exist at least one basis element, say $u_1$, such that $(\varphi|u_1)\neq 0$. Equation (\ref{equazionelinear}) is equivalent to the linear system
	\begin{equation}\label{linearsystemmatrix}
	\begin{cases}
	\sum_{i=1}^n \lambda_i \langle u_1|\psi_i\rangle=\langle u_1|\varphi\rangle\\
	\vdots\\
	\sum_{i=1}^n \lambda_i \langle u_n|\psi_i\rangle=\langle u_n|\varphi\rangle
	\end{cases}
	\mbox{ or }\ \ 
	\underbrace{
		\left[
		\begin{matrix}
		\langle u_1|\psi_1\rangle & \cdots & \langle u_1|\psi_n\rangle\\
		\vdots & \ddots & \vdots\\
		\langle u_n|\psi_1\rangle & \cdots & \langle u_n|\psi_n\rangle
		\end{matrix}
		\right]
	}_{\sM}
	\underbrace{
		\left[
		\begin{matrix}
		\lambda_1\\
		\vdots\\
		\lambda_n
		\end{matrix}
		\right]}_{\lambda}
	=
	\underbrace{
		\left[
		\begin{matrix}
		\langle u_1|\varphi\rangle\\
		\vdots\\
		\langle u_n|\varphi\rangle
		\end{matrix}
		\right].
	}_{\sD}
	\end{equation}
	The matrix $\sM$ on the right-hand side is nonsingular. Indeed, suppose there is some $(\lambda_1,\dots,\lambda_n)$ in its kernel, then we would have 
	$$
	0=\sum_{i=1}^n\lambda_i\langle u_j|\psi_i\rangle=\left\langle u_j\bigg|\sum_{i=1}^n \lambda_i\psi_i\right\rangle\quad\mbox{for all } j=1,\dots,n,
	$$
	which implies $\sum_{i=1}^n\lambda_i\psi_i\perp\sU$ and so, as above, it must be $\lambda_i=0$ for all $i=1,\dots,n$. 
	Thus, the linear system (\ref{linearsystemmatrix}) has one and only one solution given by
	\begin{equation}\label{equationsystem}
	\lambda:=\left[
	\begin{matrix}
	\lambda_1\\
	\vdots\\
	\lambda_n
	\end{matrix}
	\right]
	=
	\left[
	\begin{matrix}
	\langle u_1|\psi_1\rangle & \cdots & \langle u_1|\psi_n\rangle\\
	\vdots & \ddots & \vdots\\
	\langle u_n|\psi_1\rangle & \cdots & \langle u_n|\psi_n\rangle
	\end{matrix}
	\right]^{-1}
	\left[
	\begin{matrix}
	\langle u_1|\varphi\rangle\\
	\vdots\\
	\langle u_n|\varphi\rangle
	\end{matrix}
	\right].
	\end{equation}
\end{proof}
\begin{proof}[Proof of Proposition \ref{part3}]
	Let us fix an orthonormal basis $\{u_1,\dots,u_n \}$ of $\sU$ and fix for any $0<\epsilon<n^{-1/2}$ an approximating set $\{\psi_1^\epsilon,\dots,\psi_n^\epsilon\}$ as in Proposition \ref{part1} (where the dependence on $\epsilon$ has been made explicit). Then, in particular, it holds that 
	$
	|\langle u_i|\psi_j^\epsilon\rangle-\delta_{ij}|<\epsilon
	$
	for any $i,j=1,\dots,n$ and therefore\footnote{Note that $\|A\|_2\le \sqrt{\sum_{i,j=1}^n |a_{ij}|^2}=:\|A\|_{HS}$ (\textit{Hilbert-Schmidt norm}) for any $A\in\bM(4,\bC)$.}
	$$
	\|\sM_\epsilon-\bI_n\|_2\le \sqrt{\sum_{i,j=1}^n|\langle u_i|\psi^\epsilon_j\rangle-\delta_{ij}|^2}<n\epsilon.
	$$
	At this point, exploiting the continuity of the inverse-matrix function, it is possible to choose $\epsilon$ small enough so that $\|\sM_\epsilon^{-1}-\bI_n\|_2<1$. This implies that $\|\sM_\epsilon^{-1}\|_2\le \|\sM_\epsilon^{-1}-\bI_n\|_2+\|\bI_n\|_2<2$ and therefore, exploiting (\ref{equationsystem}), we get
	$$
	|\lambda|\le \|\sM_\epsilon^{-1}\|_2|(\langle u_1|\varphi\rangle,\dots,\langle u_n|\varphi\rangle)|\le  2\sqrt{\sum_{i=1}^n|\langle u_i|\varphi\rangle|^2}\le 2\|\varphi\|,
	$$
	where we exploited Bessel's inequality in the last inequality.
\end{proof}

%
%

\newpage 

\thispagestyle{empty}

\end{document}